\documentclass[12pt]{article}
\usepackage{amsmath, amsfonts, amsthm, amssymb}
\usepackage{geometry, graphicx, subfig, hyperref, cite}
\newtheorem{theorem}{Theorem}
\newtheorem{proposition}{Proposition}
\numberwithin{equation}{section}

\begin{document}

\begin{titlepage}

\begin{flushright}
YITP-SB-17-54
\end{flushright}
\begin{center}
\vspace{1.0cm}
\Large{\textbf{The unitary subsector of generalized minimal models}}\\
\vspace{0.8cm}
\small{\textbf{Connor Behan}}\\
\vspace{0.5cm}
\textit{C. N. Yang Institute for Theoretical Physics, Stony Brook University, \\ Stony Brook, NY 11794, USA}
\end{center}

\vspace{1.0cm}
\begin{abstract}
We revisit the line of non-unitary theories that interpolate between the Virasoro minimal models. Numerical bootstrap applications have brought about interest in the four-point function involving the scalar primary of lowest dimension. Using recent progress in harmonic analysis on the conformal group, we prove the conjecture that global conformal blocks in this correlator appear with positive coefficients. We also compute many such coefficients in the simplest mixed correlator system. Finally, we comment on the status of using global conformal blocks to isolate the truly unitary points on this line.
\end{abstract}

\end{titlepage}

\tableofcontents

\section{Introduction}
Conformal field theories (CFTs) in two dimensions enjoy invariance under two copies of the Virasoro algebra --- an algebra defined by
\begin{equation}
\left [ L_m, L_n \right ] = (m - n) L_{m + n} + \frac{c}{12} m(m-1)(m+1) \delta_{m + n, 0} \; , \label{virasoro}
\end{equation}
where $c$ is the central charge. The power of this infinite-dimensional symmetry was perhaps most famously demonstrated in \cite{bpz84} with the discovery of the minimal models. In addition to providing an exact solution, representation theory of the Virasoro algebra enabled \cite{fqs84, fqs86, gko85, gko86} to show that these models are the only unitary CFTs in two dimensions with $c < 1$. However, it has become known more recently that one can see hints of the special role played by minimal models without exploiting Virasoro symmetry at all \cite{rrtv08, rv09, cr09, rrv10a, rrv10b, v12, ps10, psv12}. The method in question is the numerical bootstrap, which uses only the global conformal transformations --- two copies of $\mathfrak{sl}(2)$ in this case. Exclusion plots, based on crossing symmetry and unitarity, are shown in Figure \ref{single-multi} where a straight line containing the minimal models is clearly visible.\footnote{The algorithm used to generate Figure \ref{single-multi} will become important in section 4. These details are summarized in Appendix B.} A kink is present at the Ising point $(\Delta_\sigma, \Delta_\epsilon) = \left ( \frac{1}{8}, 1 \right )$ but, unlike in the three-dimensional case \cite{kps14, s15, kpsv16}, this kink does not sharpen into an island when three correlators are used to restrict the number of relevant operators. Following \cite{lrv13}, it is worthwhile to see which features of the $\frac{1}{8} \leq \Delta_\sigma \leq \frac{1}{2}$ solution can be predicted analytically. The goal of this work is to put the one-correlator upper bound on a more rigorous footing and to explain why the three-correlator upper bound is unchanged.
\begin{figure}[t!]
\centering
\subfloat[][One correlator]{\includegraphics[scale=0.45]{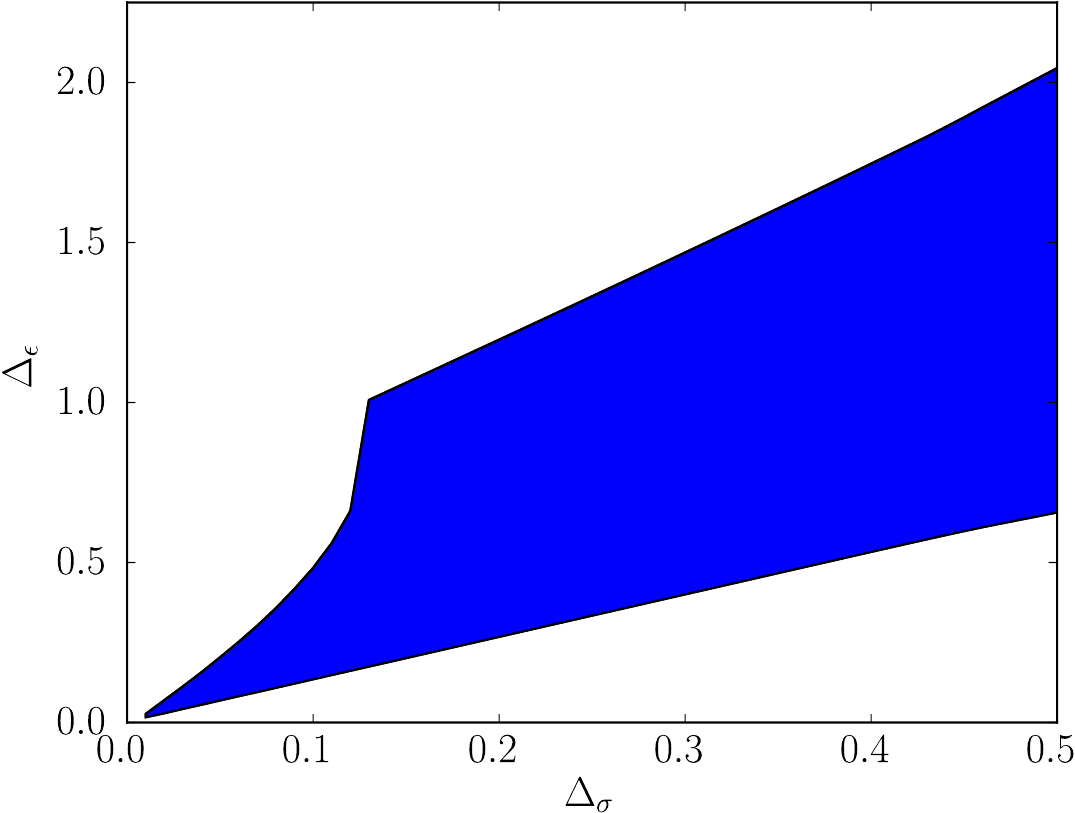}}
\subfloat[][Three correlators]{\includegraphics[scale=0.45]{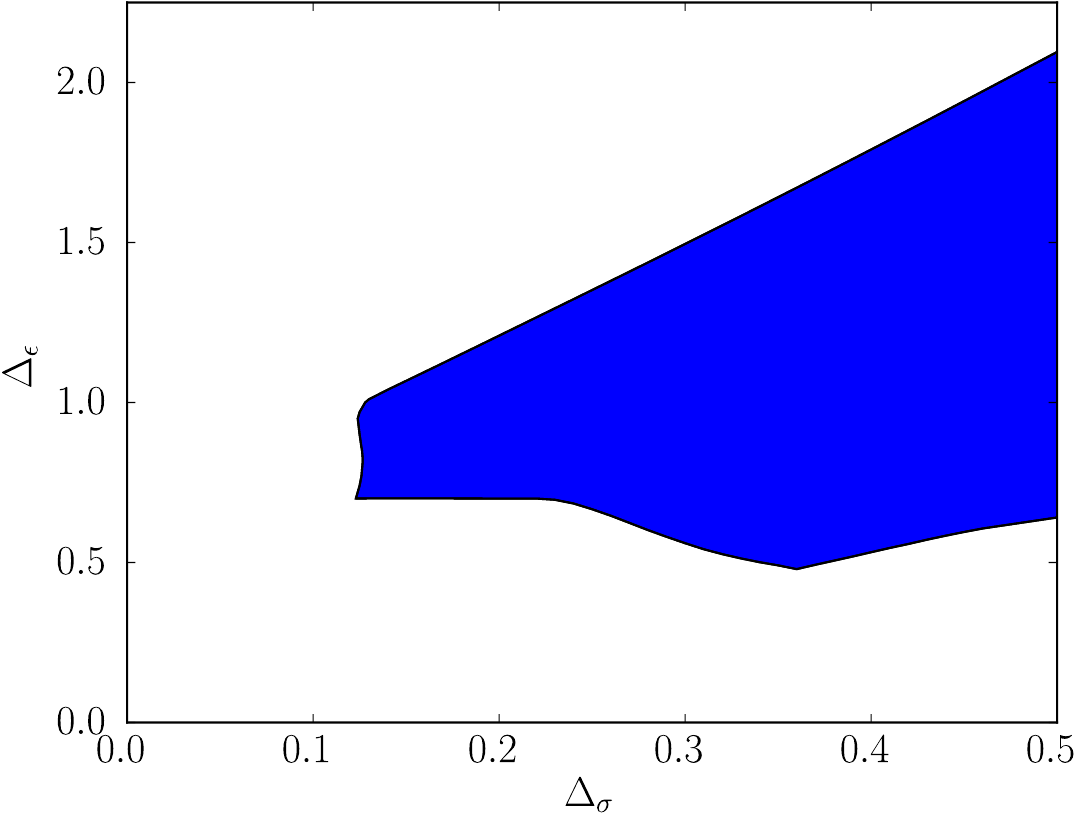}}
\caption{Allowed regions for the dimensions of $\sigma$ and $\epsilon$ --- the $\mathbb{Z}_2$-odd scalar of smallest dimension and the $\mathbb{Z}_2$-even scalar of smallest dimension respectively. The left plot follows from the constraints of crossing symmetry and unitarity on the four-point function $\left < \sigma\sigma\sigma\sigma \right >$. The right plot comes from the same constraints on $\left < \sigma\sigma\sigma\sigma \right >$, $\left < \sigma\sigma\epsilon\epsilon \right >$ and $\left < \epsilon\epsilon\epsilon\epsilon \right >$. In both cases, all OPEs are restricted to contain only one relevant scalar.}
\label{single-multi}
\end{figure}

In order to have a unitary 2D CFT with $c < 1$, it is necessary that all primary operators have conformal weights equal to Kac's formula $h_{r,s}(c)$ for some $(r, s)$. The Kac table of degenerate weights is given by
\begin{eqnarray}
c &=& 1 - \frac{6}{m(m + 1)} \;\;\;\;\;\;\;\;\;\;\;\;\;\; m > 2 \nonumber \\
h_{r, s} &=& \frac{[(m + 1)r - ms]^2 - 1}{4m(m + 1)} \;\;\; r, s \in \mathbb{Z}_{>0} \; . \label{kac-table}
\end{eqnarray}
Each of these Verma modules has a null state at level $rs$. In the operator product expansion (OPE) of primary operators $\phi_{r, s}$ and $\phi_{r^\prime, s^\prime}$, the new conformal families that appear are captured in the fusion rule
\begin{equation}
\phi_{r, s} \times \phi_{r^\prime, s^\prime} = \sum_{k = 0}^{\left \lfloor \frac{r + r^\prime - |r - r^\prime| - 2}{2} \right \rfloor} \sum_{l = 0}^{\left \lfloor \frac{s + s^\prime - |s - s^\prime| - 2}{2} \right \rfloor} \phi_{|r - r^\prime| + 2k + 1, |s - s^\prime| + 2l + 1} \; . \label{fusion}
\end{equation}
For generic values of $m$, this leads to an infinite discrete spectrum. All OPEs are finite, but as we raise the values of $r$ and $s$, these sums become arbitrarily long. A special situation occurs when $m$ is an integer. This precisely describes a central charge for which $h_{r, s}(c) = h_{m - r, m + 1 - s}(c)$. The Kac table for these doubly degenerate weights can be shown to truncate, allowing us to consider only $0 < r < m$ and $0 < s < m + 1$. This leads to a finite number of primary operators and, as it turns out, a unitary theory. This theory, called a (unitary) minimal model, is often denoted $\mathcal{M}(m + 1, m)$. Since Figure \ref{single-multi} only shows a kink for $m = 3$, it is evident that non-integer values of $m$ are still important for the bootstrap. Analytically continuing the unitary minimal models in this way is not new. In \cite{z05, bz05}, the generic $\mathcal{M}(m + 1, m)$ was found to be a solvable consistent theory and referred to as a \textit{generalized minimal model}.\footnote{This should not be confused with \textit{non-unitary minimal model}, which describes a non-unitary $c < 1$ theory with finitely many primaries. This discrete set is denoted $\mathcal{M}(p, q)$ with $p$ and $q$ relatively prime.} It obeys the 2D CFT axioms of associativity and Virasoro symmetry but not unitary \cite{r14}. The upper bound from the bootstrap
\begin{equation}
\Delta_\epsilon = \frac{1}{3} (8\Delta_\sigma + 2) \; , \label{upper-bound}
\end{equation}
is realized by the generalized minimal model four-point function $\left < \sigma\sigma\sigma\sigma \right >$ if we identify $\sigma \equiv \phi_{1, 2}$ and $\epsilon \equiv \phi_{1, 3}$.\footnote{This notation differs from the statistical physics literature in which it is natural to regard $\phi_{2, 2}$ as the spin-field.} The authors of \cite{lrv13} observed squared OPE coefficients of quasiprimaries in this correlator that were all positive. If this conjecture is correct, we must conclude that in an arbitrary $\mathcal{M}(m + 1, m)$ theory, the non-unitarity is mild enough that it cannot be diagnosed from the correlator of four $\sigma$ operators. What this means is that for a non-degenerate $\mathcal{O}$ appearing in the $\sigma \times \sigma$ OPE, the state $\mathcal{O}(x_1) \sigma(x_2) \sigma(x_3) \left | 0 \right >$ on the cylinder will have a positive norm. Similarly, within a degenerate subspace, all such norms will sum to a positive number. We use the terminology that $\left < \sigma\sigma\sigma\sigma \right >$ comprises a \textit{unitary subsector} of the theory since other OPEs, \textit{e.g.} $\sigma \times \epsilon$, are needed to construct negative norms. As we will show, the conjecture can indeed be proven with the help of a new formalism in \cite{hv17, h17} for deriving $\mathfrak{sl}(2)$ block expansions. This proof, along with a systematic look at the other two correlators, forms the main result of this work.\footnote{Looking ahead, the decompositions (\ref{4sigma-result-id}) and (\ref{4sigma-result-eps}) are essential for the positivity proof. We have learned that they were previously obtained, through a slightly different method, in unpublished work by Mikhail Isachenkov and Volker Schomerus.}
\begin{table}[h]
\centering
\begin{tabular}{l|l|l}
Correlator & $3 < m < 4$ & $4 < m < \infty$ \\
\hline
$\left < \sigma\sigma\sigma\sigma \right >$ & All coefficients $\geq 0$ & All coefficients $\geq 0$ \\
$\left < \sigma\sigma\epsilon\epsilon \right >$ & One checked coefficient $< 0$ & All checked coefficients $\geq 0$ \\
$\left < \epsilon\epsilon\epsilon\epsilon \right >$ & Infinitely many coefficients $< 0$ & All checked coefficients $\geq 0$
\end{tabular}
\caption{The status of three-correlator $\mathfrak{sl}(2)$ block coefficients in the generalized minimal models. Statements about $\left < \sigma\sigma\sigma\sigma \right >$ apply rigorously to the full set of coefficients. For the other two correlators, we have manually decomposed them up to order $15$.}
\label{result-table}
\end{table}
The summary of results in Table \ref{result-table} shows that for $3 < m < 4$, $\left < \sigma\sigma\sigma\sigma \right >$ is the only unitary subsector we have found in the sense described above. However, it appears that the generalized minimal models with $m > 4$ have larger unitary subsectors that include the other two four-point functions. The last line shows a surprising tension with Figure \ref{single-multi}. Despite the fact that $\left < \epsilon\epsilon\epsilon\epsilon \right >$ displays significant unitarity violation for $m < 4$, the line (\ref{upper-bound}) in the three-correlator exclusion plot is uninterrupted. The numerics are telling us that there is a partial solution to crossing, other than $\mathcal{M}(m + 1, m)$, which fills in this region. Using the properties of minimal models, we will show that the existence of this solution can be concluded from a \textit{simpler} numerical setup. It would be nice to eventually find a fully analytic construction.

Before deriving the results in Table \ref{result-table}, it is instructive to consider the lower bound
\begin{equation}
\Delta_\epsilon = \frac{4}{3} \Delta_\sigma \; , \label{lower-bound}
\end{equation}
which appears on the left side of Figure \ref{single-multi}. This gives us a more straightforward opportunity to use the techniques in \cite{hv17, h17}. The explicit solution for $\left < \sigma\sigma\sigma\sigma \right >$ along this line was found in \cite{efr16}, which focused on its special role in the non-unitary (severe truncation) bootstrap of \cite{g13, gr14, glmr15}.\footnote{The fact that it also appears in the unitary bootstrap has not received much attention. In \cite{s16}, it was mentioned that the lower bound at $\Delta_\sigma = \frac{1}{8}$ was somewhat close to $\Delta_\epsilon = \frac{1}{6}$. Interestingly, the region below the bound coincides with the region where standard OPE maximization techniques cannot constrain the central charge \cite{v11}.} This solution exhibits Virasoro symmetry with a central charge given by
\begin{equation}
c = 1 + 16 \Delta_\sigma \label{non-kac-table}
\end{equation}
but no Virasoro identity block. To find an analytic explanation for why this four-point function appears in the unitary bootstrap, one must be able to show that this vacuum decoupling is the only sign of non-unitarity that appears at the level of a single correlator. In other words, one must be able to repeat the logic of \cite{lrv13} and find positive squared OPE coefficients for all of the quasiprimaries that do appear. These will be seen as very large coefficients by the numerics because the algorithm used for Figure \ref{single-multi} fixes $\lambda_{\sigma\sigma I} = 1$. This is indeed what we find from the extremal functional method of \cite{ep12}.

This paper is organized as follows. In section 2, we study the lower bound (\ref{lower-bound}) as a warm-up. In this case, it is particularly easy to invert the OPE and find that all $\mathfrak{sl}(2)$ block coefficients are positive. The methods involved prepare us for our main interest, which is the upper bound (\ref{upper-bound}). Turning to this upper bound in section 3, we calculate the global block coefficients summarized in Table \ref{result-table}. In the case of $\left < \sigma\sigma\sigma\sigma \right >$, we find closed-form expressions. In showing that they are positive, we prove the conjecture made in the appendix of \cite{lrv13}. For the other two correlators, we expand them to high order recursively and conjecture that the $\mathfrak{sl}(2)$ block coefficients are positive again when $4 < m < \infty$. Even though the Virasoro blocks in $\left < \epsilon\epsilon\epsilon\epsilon \right >$ appear to have positive $\mathfrak{sl}(2)$ expansions everywhere, the unitarity violation for $3 < m < 4$ arises because of the coefficients multiplying the Virasoro blocks themselves. In section 4, we go back to the bootstrap and discuss what these patterns in the OPE coefficients mean for the results in Figure \ref{single-multi}. In particular, we perform a semi-analytic treatment of the problematic correlator. The result is that one does not need to perform a three-correlator bootstrap to predict that (\ref{upper-bound}) survives --- the search for an upper bound may be reduced to a one-correlator problem. Before concluding, in section 5, we discuss other 2D theories which might be possible to study using more correlators or more assumptions on the spectrum.

\section{The lower line: A warm-up}
The four-point function along the line (\ref{lower-bound}) consists of a single Virasoro block $V(h_i, h, c ; z)$. It was found in \cite{efr16} via the Coulomb gas formalism which writes the central charge as $c = 1 - 24\alpha_0^2$ and places a background charge of $2\alpha_0$ at infinity. This allows a number of four-point functions to be realized as correlators of vertex operators with additional insertions of screening charges. The simplest of these is a correlator of four scalars that all have charge $\frac{\alpha_0}{2}$. Since the neutrality condition for this is satisfied without any screening charges, one finds the manifestly crossing symmetric
\begin{eqnarray}
\left < \sigma(0) \sigma(z, \bar{z}) \sigma(1) \sigma(\infty) \right > &=& V \left ( -\frac{3}{4} \alpha_0^2, -\alpha_0^2, 1 - 24\alpha_0^2 ; z \right ) V \left ( -\frac{3}{4} \alpha_0^2, -\alpha_0^2, 1 - 24\alpha_0^2 ; \bar{z} \right ) \nonumber \\
&=& |z (1 - z)|^{\alpha_0^2} \; , \label{lower-4sigma1}
\end{eqnarray}
where we have used $h = \alpha (\alpha - 2\alpha_0)$. Expressing (\ref{lower-4sigma1}) in terms of $\Delta_\sigma$,
\begin{eqnarray}
\left < \sigma(0) \sigma(z, \bar{z}) \sigma(1) \sigma(\infty) \right > &=& \frac{g(z) g(\bar{z})}{|z|^{2\Delta_\sigma}} \nonumber \\
g(z) &=& z^{\frac{2}{3}\Delta_\sigma} (1 - z)^{-\frac{1}{3}\Delta_\sigma} \; . \label{lower-4sigma2}
\end{eqnarray}

Our task now is to expand $g(z)$ into $\mathfrak{sl}(2)$ blocks:
\begin{eqnarray}
g(z) &=& \sum_{n = 0}^\infty c_n K_{\frac{2}{3}\Delta_\sigma + n}(z) \nonumber \\
K_h(z) &\equiv& z^h {}_2F_1(h, h; 2h; z) \; . \label{sl2-block}
\end{eqnarray}
This is guaranteed to be an expansion in even integers due to the Bose symmetry of the $\sigma \times \sigma$ OPE. To proceed by the brute-force approach, we expand the hypergeometric function and switch the order of two sums.
\begin{eqnarray}
g(z) &=& \sum_{n = 0}^\infty \sum_{m = 0}^\infty c_n \frac{\left ( \frac{2}{3} \Delta_\sigma + n \right )_m^2}{\left ( \frac{4}{3} \Delta_\sigma + 2n \right )_m} \frac{z^{\frac{2}{3} \Delta_\sigma + n + m}}{m!} \nonumber \\
&=& \sum_{k = 0}^\infty \sum_{n = 0}^k c_n \frac{\left ( \frac{2}{3} \Delta_\sigma + n \right )_{k - n}^2}{\left ( \frac{4}{3} \Delta_\sigma + 2n \right )_{k - n}} \frac{z^{\frac{2}{3} \Delta_\sigma + k}}{(k - n)!} \label{sieve1}
\end{eqnarray}
We may now compare the inner finite sums to the Taylor coefficients of (\ref{lower-4sigma2}), given by $b_k = \frac{(\frac{1}{3} \Delta_\sigma)_k}{k!}$. Since the lower triangular system for $c_n$ yields to back-substitution,
\begin{equation}
c_{2k} = b_{2k} - \sum_{n = 0}^{k - 1} c_{2n} \frac{\left ( \frac{2}{3} \Delta_\sigma + 2n \right )_{2(k - n)}^2}{\left ( \frac{4}{3} \Delta_\sigma + 4n \right )_{2(k - n)}} \frac{1}{(2k - 2n)!} \; . \label{sieve2}
\end{equation}

Rather than using this recursive procedure, we will now review a method for computing the $c_n$ directly. The blocks, defined in (\ref{sl2-block}), are eigenfunctions of the conformal Casimir
\begin{eqnarray}
D K_h(z) &=& h(h - 1) K_h(z) \nonumber \\
D &=& z^2(1 - z)\frac{\partial^2}{\partial z^2} - z^2\frac{\partial}{\partial z} \; . \label{casimir}
\end{eqnarray}
It is well known that $D$ is self-adjoint on $[0, 1]$ with respect to the measure $z^{-2}$. The authors of \cite{hv17} used this fact to develop the Sturm-Liouville theory of this operator and construct the orthogonal eigenfunctions
\begin{equation}
\Psi_h(z) = \frac{\Gamma(1 - 2h)}{\Gamma(1 - h)^2} K_h(z) + (h \leftrightarrow 1 - h) \; . \label{eigenfunction}
\end{equation}
It is convenient to set $h = \frac{1}{2} + \alpha$ in which case (\ref{eigenfunction}) becomes a function $\Psi_\alpha(z)$ which is even in $\alpha$. In order for it to have a finite norm, $\alpha$ cannot be real. We must go to imaginary dimension space and take $\alpha \in i\mathbb{R}$.\footnote{There is another name for this space as evidenced by the title of \cite{hv17}.} The result is that to any four-point function $f(z)$, we may associate a density $\hat{f}(\alpha) = \hat{f}(-\alpha)$ via the invertible transform
\begin{eqnarray}
f(z) &=& \frac{1}{2\pi i} \int_{-i\infty}^{i\infty} \hat{f}(\alpha) \Psi_\alpha(z) \frac{\textup{d}\alpha}{N(\alpha)} \nonumber \\
N(\alpha) &\equiv& \frac{\Gamma(\alpha)\Gamma(-\alpha)}{2\pi \Gamma(\frac{1}{2} + \alpha) \Gamma(\frac{1}{2} - \alpha)} \; . \label{alpha-space}
\end{eqnarray}
It is now clear that OPE coefficients may be read off from the residues of $\hat{f}(\alpha)$ whenever its poles are on the real axis. A formula that \cite{hv17, h17} derived using this method is
\begin{equation}
z^p (1 - z)^{-q} = \sum_{n = 0}^\infty \frac{(p)_n^2}{(2p + n - 1)_n n!} {}_3F_2 \left ( \begin{tabular}{c} $-n, 2p + n - 1, p - q$ \\ $p, p$ \end{tabular} ; 1 \right ) K_{p + n}(z) \; . \label{power-law}
\end{equation}
We will use this in the current section and the next one.

Specializing (\ref{power-law}) to the four-point function (\ref{lower-4sigma2}), we immediately find
\begin{equation}
c_n = \frac{\left ( \frac{2}{3} \Delta_\sigma \right )_n^2}{\left ( \frac{4}{3} \Delta_\sigma + n - 1\right )_n n!} {}_3F_2 \left ( \begin{tabular}{c} $-n, \frac{4}{3} \Delta_\sigma + n - 1, \frac{1}{3} \Delta_\sigma$ \\ $\frac{2}{3} \Delta_\sigma, \frac{2}{3} \Delta_\sigma$ \end{tabular} ; 1 \right ) \; . \label{continuous-hahn-before}
\end{equation}
There are two ways to assess the positivity of (\ref{continuous-hahn-before}). The first is to recall the definition of a continuous Hahn polynomial \cite{ks96}.
\begin{equation}
\tilde{P}_n(a, b, c, d; x) = {}_3F_2 \left ( \begin{tabular}{c} $-n, n + a + b + c + d - 1, a + x$ \\ $a + c, a + d$ \end{tabular} ; 1 \right ) \label{continuous-hahn}
\end{equation}
Clearly,
\begin{equation}
c_n = \frac{\left ( \frac{2}{3} \Delta_\sigma \right )_n^2}{\left ( \frac{4}{3} \Delta_\sigma + n - 1\right )_n n!} \tilde{P}_n \left ( \frac{2}{3} \Delta_\sigma, \frac{2}{3} \Delta_\sigma, 0, 0; -\frac{1}{3} \Delta_\sigma \right ) \; , \label{continuous-hahn-after}
\end{equation}
is a valid rewriting of (\ref{continuous-hahn-before}).\footnote{In the notation of \cite{gkss17}, we would write $c_n = \frac{2^{-n}}{n!} Q_{n, 0}^{\frac{4}{3} \Delta_\sigma + n} \left ( -\frac{1}{3} \Delta_\sigma \right )$.} Suppressing their parameters, the polynomials $\tilde{P}_n(a, b, c, d; x)$ satisfy the following recurrence relation:
\begin{eqnarray}
(x + a)\tilde{P}_n(x) &=& A_n \tilde{P}_{n + 1}(x) - (A_n + B_n) \tilde{P}_n(x) + B_n \tilde{P}_{n - 1}(x) \nonumber \\
A_n &\equiv& -\frac{(n + a + b + c + d - 1)(n + a + c)(n + a + d)}{(2n + a + b + c + d - 1)(2n + a + b + c + d)} \nonumber \\
B_n &\equiv& \frac{n(n + b + c - 1)(n + b + d - 1)}{(2n + a + b + c + d - 2)(2n + a + b + c + d - 1)} \; . \label{continuous-hahn-recursion}
\end{eqnarray}
For our parameters, we may easily check that $A_n + B_n + a + x = 0$. It then follows by induction that all $c_{2k + 1}$ vanish. Once we know this, (\ref{continuous-hahn-recursion}) is effectively a two-term recursion. Seeing a positive constant of proportionality in
\begin{eqnarray}
\tilde{P}_{2k} \left ( -\frac{1}{3}\Delta_\sigma \right ) &=& -\frac{B_{2k - 1}}{A_{2k - 1}} \tilde{P}_{2k - 2} \left ( -\frac{1}{3}\Delta_\sigma \right ) \nonumber \\
&=& \frac{3(2k - 1)(\Delta_\sigma + 3k - 3)}{(2\Delta_\sigma + 3k - 3)(2\Delta_\sigma + 6k - 3)} \tilde{P}_{2k - 2} \left ( -\frac{1}{3}\Delta_\sigma \right ) \; , \label{effective-recursion}
\end{eqnarray}
we conclude that the sequence $c_{2k}$ decays to zero monotonically from above. It is therefore imperative that the bootstrap single out a region that includes (\ref{lower-bound}). To derive this without referring to continuous Hahn polynomials, one may instead express $c_n$ in terms of gamma functions. This is possible because Watson's theorem \cite{c12},
\begin{equation}
{}_3F_2 \left ( \begin{tabular}{c} $a, b, c$ \\ $\frac{a + b + 1}{2}, 2c$ \end{tabular} ; 1 \right ) = \frac{\Gamma(\frac{1}{2})\Gamma(\frac{1+a+b}{2})\Gamma(\frac{1}{2} + c)\Gamma(\frac{1-a-b}{2} + c)}{\Gamma(\frac{1+a}{2})\Gamma(\frac{1+b}{2})\Gamma(\frac{1-a}{2} + c)\Gamma(\frac{1-b}{2} + c)} \label{watson}
\end{equation}
applies whenever $p = 2q$ in (\ref{power-law}). Although $\Re(a + b - 2c) < 1$ is usually needed for convergence, we may drop this requirement for a hypergeometric function that terminates.

This analysis does not explain why (\ref{lower-bound}) saturates the lower bound in the one-correlator result of Figure \ref{single-multi}. However, it is encouraging that bounds of this form in one dimension have been proven in \cite{m16}. The Coulomb gas formalism does not yield an obvious way to solve for the correlators $\left < \sigma\sigma\epsilon\epsilon \right >$, $\left < \epsilon\epsilon\epsilon\epsilon \right >$ or even to verify that they exist. Because the three-correlator plot in Figure \ref{single-multi} excludes this line, any theory to which $\left < \sigma\sigma\sigma\sigma \right >$ could extend would have to be highly non-unitary.

In this section, we have seen two methods for proving that global block coefficients in (\ref{lower-4sigma2}) are positive. One uses Watson's theorem and the other uses a recurrence relation for orthogonal polynomials. We will need both of these methods when we prove positivity in the generalized minimal models. Before moving on, there is an interesting way to check our results in a spacetime with Minkowski signature. Even though (\ref{lower-4sigma2}) is not strictly a correlation function in a unitary theory, it is still bounded in the Regge limit. Its $\mathfrak{sl}(2)$ block expansion should therefore be calculable with the conformal Froissart-Gribov formula \cite{c17} which in our case reads
\begin{eqnarray}
c(\Delta, \ell) &=& \kappa_{\Delta + \ell} \int_0^1 \int_0^1 K_{\frac{\Delta + \ell}{2}}(z) K_{\frac{\ell - \Delta + 2}{2}}(\bar{z}) \mathrm{dDisc}[|z|^{2\Delta_\sigma}\left < \sigma(0) \sigma(z, \bar{z}) \sigma(1) \sigma(\infty) \right >] \frac{\textup{d}z}{z^2} \frac{\textup{d}\bar{z}}{\bar{z}^2} \nonumber \\
\kappa_\beta &\equiv& \frac{\Gamma(\frac{\beta}{2})^4}{2\pi^2 \Gamma(\beta - 1)\Gamma(\beta)} \; . \label{sch}
\end{eqnarray}
To define the double-discontinuity, we must treat $z, \bar{z}$ as independent variables and rotate around the $\bar{z} = 1$ branch point. Since this can be done in two ways, we subtract the average from our four-point function to find
\begin{equation}
\mathrm{dDisc} \left [ g(z) g(\bar{z}) \right ] = 2\sin^2 \left ( \frac{\pi \Delta_\sigma}{3} \right ) (z \bar{z})^{\frac{2}{3}\Delta_\sigma} [(1 - z)(1 - \bar{z})]^{-\frac{1}{3}\Delta_\sigma} \; . \label{ddisc}
\end{equation}
Performing the factored integrals yields a spectral density given by
\begin{eqnarray}
\frac{c(\Delta, \ell) \Gamma(1 - \frac{1}{3} \Delta_\sigma)^{-2}}{2\sin^2 \left ( \frac{\pi \Delta_\sigma}{3} \right ) \kappa_{\Delta + \ell}} &=& \frac{\Gamma(\frac{2}{3} \Delta_\sigma + \frac{\ell - \Delta}{2})}{\Gamma(\frac{1}{3} \Delta_\sigma + \frac{\ell - \Delta + 2}{2})} {}_3F_2 \left ( \begin{tabular}{c} $\frac{\ell - \Delta + 2}{2}, \frac{\ell - \Delta + 2}{2}, \frac{2}{3} \Delta_\sigma + \frac{\ell - \Delta}{2}$ \\ $\ell - \Delta + 2, \frac{1}{3} \Delta_\sigma + \frac{\ell - \Delta + 2}{2}$ \end{tabular} ; 1 \right ) \nonumber \\
&& \frac{\Gamma(\frac{2}{3} \Delta_\sigma + \frac{\Delta + \ell - 2}{2})}{\Gamma(\frac{1}{3} \Delta_\sigma + \frac{\Delta + \ell}{2})} {}_3F_2 \left ( \begin{tabular}{c} $\frac{\Delta + \ell}{2}, \frac{\Delta + \ell}{2}, \frac{2}{3} \Delta_\sigma + \frac{\Delta + \ell - 2}{2}$ \\ $\Delta + \ell, \frac{1}{3} \Delta_\sigma + \frac{\Delta + \ell}{2}$ \end{tabular} ; 1 \right ) \label{sch-density}
\end{eqnarray}
with poles at $\Delta - \ell = \frac{4}{3} \Delta_\sigma + 2n$. As none of these are integers, the correct prescription for finding OPE coefficients is to simply take the residue \cite{c17}.
\begin{eqnarray}
-\mathrm{Res} \left ( c(\Delta, \ell), n \right ) &=& \frac{(-1)^n}{n!} \frac{8 \sin^2 \left ( \frac{\pi \Delta_\sigma}{3} \right ) \kappa_{\frac{4}{3}\Delta_\sigma + 2\ell + 2n} \Gamma(1 - \frac{1}{3}\Delta_\sigma)^2 \Gamma(\frac{4}{3}\Delta_\sigma + n + \ell - 1)}{\Gamma(\Delta_\sigma + \ell + n)\Gamma(1 - \frac{1}{3}\Delta_\sigma - n)} \nonumber \\
&& {}_3F_2 \left ( \begin{tabular}{c} $\frac{4}{3}\Delta_\sigma + \ell + n - 1, \frac{2}{3}\Delta_\sigma + \ell + n, \frac{2}{3}\Delta_\sigma + \ell + n$ \\ $\Delta_\sigma + \ell + n, \frac{4}{3}\Delta_\sigma + 2\ell + 2n$ \end{tabular} ; 1 \right ) \nonumber \\
&& {}_3F_2 \left ( \begin{tabular}{c} $-n, 1 - n - \frac{2}{3}\Delta_\sigma, 1 - n - \frac{2}{3}\Delta_\sigma$ \\ $1 - n -\frac{1}{3}\Delta_\sigma, 2 - 2n - \frac{4}{3}\Delta_\sigma$ \end{tabular} ; 1 \right ) \label{sch-residue}
\end{eqnarray}
Although it is not obvious, we have checked that (\ref{sch-residue}) is equal to $c_n c_{n + \ell}$ by using Watson's theorem twice.

\section{The upper line: One and three correlators}
Each state $\left | h_{r, s} \right >$ dual to a degenerate operator in a generalized minimal model has a null descendant $\left | \chi_{r, s} \right >$ at level $rs$ called a singular vector. Four-point functions may be calculated once the necessary singular vectors are known. We simply convert $\chi_{r, s}$ to a differential operator with the Ward identities
\begin{eqnarray}
L_{-n} \mapsto \mathcal{L}_{-n} &=& \sum_{i \neq 1} \frac{h_i (n - 1)}{(z_i - z_1)^n} - \frac{\partial_i}{(z_i - z_1)^{n - 1}} \nonumber \\
L_{-1} \mapsto \mathcal{L}_{-1} &=& \partial_1 \; , \label{differential}
\end{eqnarray}
and use the fact that this operator must annihilate any correlation function involving $\phi_{r, s}$ in the first position. The resulting equation, known as a BPZ differential equation \cite{bpz84}, has $rs$ linearly independent solutions representing the exchanged multiplets. One may read off their dimensions by looking at the $O(z^h)$ behaviour as $z \rightarrow 0$. While expressions for singular vectors are generally non-trivial, with some appearing only recently \cite{dlm12, bbt13, adm13}, the ones we need are relatively simple:
\begin{eqnarray}
\left | \chi_{1, s} \right > &=& \sum_{p_1 + \dots + p_k = s} \frac{(-t)^{s - k} [(s - 1)!]^2}{\prod_{i = 1}^{k - 1} (p_1 + \dots + p_i)(s - p_1 - \dots - p_i)} L_{-p1} \dots L_{-p_k} \left | h_{1, s} \right > \nonumber \\
t &\equiv& \frac{m}{m + 1} \; . \label{singular}
\end{eqnarray}
\begin{table}[h]
\centering
\begin{tabular}{l|l}
Fusion rules & Weights \\
\hline
$\phi_{1,2} \times \phi_{1,2} = \phi_{1,1} + \phi_{1,3}$ & $h_{1,1} = 0$ \\
$\phi_{1,2} \times \phi_{1,3} = \phi_{1,2} + \phi_{1,4}$ & $h_{1,2} = \frac{\Delta_\sigma}{2}$ \\
$\phi_{1,3} \times \phi_{1,3} = \phi_{1,1} + \phi_{1,3} + \phi_{1,5}$ & $h_{1,3} = \frac{4\Delta_\sigma + 1}{3} \equiv \frac{\Delta_\epsilon}{2}$ \\
& $h_{1,4} = \frac{5\Delta_\sigma + 2}{2}$ \\
& $h_{1,5} = 4\Delta_\sigma + 2$
\end{tabular}
\caption{Operators that can appear in $\left < \sigma\sigma\sigma\sigma \right >$, $\left < \sigma\sigma\epsilon\epsilon \right >$, $\left < \epsilon\epsilon\epsilon\epsilon \right >$ and their holomorphic weights. The fusion rules would shorten \textit{e.g.} in the Ising model $m = 3$ and tricritical Ising model $m = 4$, but we are interested in $\mathcal{M}(m + 1, m)$ for real $m$.}
\label{fusion-ops}
\end{table}
It is clear that we may confine ourselves to the border of the Kac table when studying correlators of $\sigma \equiv \phi_{1, 2}$ and $\epsilon \equiv \phi_{1, 3}$. The set of operators $\phi_{1, s}$ which closes under fusion is called the Verlinde subalgebra. We will continue to parametrize the generalized minimal model by $\Delta_\sigma$ --- the horizontal axis of Figure \ref{single-multi}. For convenience, Table \ref{fusion-ops} summarizes the OPEs that are important for the simplest mixed correlator system.

\subsection{All global block coefficients in the simplest correlator}
We will now derive new expressions for the squared OPE coefficients in $\left < \sigma\sigma\sigma\sigma \right >$ along (\ref{upper-bound}). Positivity, as predicted by \cite{lrv13}, will then follow from methods analogous to those in the last section. Since this correlator solves a second-order BPZ equation, we may write it as
\begin{eqnarray}
&& \left < \sigma(z_1, \bar{z}_1) \sigma(z_2, \bar{z}_2) \sigma(z_3, \bar{z}_3) \sigma(z_4, \bar{z}_4) \right > = \frac{G(z, \bar{z})}{|z_{12} z_{34}|^{2\Delta_\sigma}} \nonumber \\
&& G(z, \bar{z}) = G^{\sigma\sigma\sigma\sigma}_{(1, 1)}(z) G^{\sigma\sigma\sigma\sigma}_{(1, 1)}(\bar{z}) + C^{(1, 3)}_{(1, 2)(1, 2)} G^{\sigma\sigma\sigma\sigma}_{(1, 3)}(z) G^{\sigma\sigma\sigma\sigma}_{(1, 3)}(\bar{z}) \; , \label{upper-4sigma}
\end{eqnarray}
where the functions of $z$ are Virasoro blocks.\footnote{For the moment, we will be concerned with the global OPE coefficients contained within each one. The overall coefficients $C_{(r_1,s_1)(r_2,s_2)}^{(r_3,s_3)}$ are the generalized minimal model structure constants that were obtained with the Coulomb gas formalism in \cite{df84, df85a, df85b}.} The specific operator, read off from (\ref{singular}), is
\begin{equation}
\frac{3}{2(\Delta_\sigma + 1)} \mathcal{L}_{-1}^2 - \mathcal{L}_{-2} \; . \label{level2}
\end{equation}
Acting on (\ref{upper-4sigma}) with (\ref{level2}), we arrive at a PDE in terms of $(z_1, z_2, z_3, z_4)$. To reduce it to an ODE, we map these points to $(0, z, 1, \infty)$ by a global conformal transformation:
\begin{equation}
\frac{3}{2} z(z - 1)^2 \frac{\partial^2 G}{\partial z^2} + (z - 1) [(2 - \Delta_\sigma)z + 2\Delta_\sigma - 1] \frac{\partial G}{\partial z} - \frac{1}{2} \Delta_\sigma (\Delta_\sigma + 1) zG = 0 \; . \label{4sigma-ode}
\end{equation}
Well known solutions, which have the expected asymptotic behaviour, are
\begin{eqnarray}
G^{\sigma\sigma\sigma\sigma}_{(1, 1)}(z) &=& (1 - z)^{-\Delta_\sigma} {}_2F_1 \left ( -2\Delta_\sigma, \frac{1 - 2\Delta_\sigma}{3} ; \frac{2 - 4\Delta_\sigma}{3} ; z \right ) \nonumber \\
G^{\sigma\sigma\sigma\sigma}_{(1, 3)}(z) &=& z^{\frac{1 + 4\Delta_\sigma}{3}} (1 - z)^{-\Delta_\sigma} {}_2F_1 \left ( \frac{1 - 2\Delta_\sigma}{3}, \frac{2 + 2\Delta_\sigma}{3} ; \frac{4 + 4\Delta_\sigma}{3} ; z \right ) \; . \label{4sigma-solutions}
\end{eqnarray}

First, let us look at the identity block. After a quadratic transformation, the hypergeometric function becomes a series in $\frac{z^2}{4z - 4}$. This makes the formula (\ref{power-law}) applicable if we set $p = 2n$ and $q = n$.
\begin{eqnarray}
&& G^{\sigma\sigma\sigma\sigma}_{(1, 1)}(z) = {}_2F_1 \left ( -\Delta_\sigma, \frac{1 + \Delta_\sigma}{3} ; \frac{5 - 4\Delta_\sigma}{6} ; \frac{1}{4} \frac{z^2}{z - 1} \right ) \label{quadratic-id} \\
&& = \sum_{n = 0}^\infty \left ( -\frac{1}{4} \right )^n \frac{(-\Delta_\sigma)_n \left ( \frac{1 + \Delta_\sigma}{3} \right )_n}{\left ( \frac{5 - 4\Delta_\sigma}{6} \right )_n n!} \sum_{m = 0}^\infty \frac{(2n)_m^2}{(4n + m - 1)_m m!} {}_3F_2 \left ( \begin{tabular}{c} $-m, 4n + m - 1, n$ \\ $2n, 2n$ \end{tabular} ; 1 \right ) K_{2n + m}(z) \nonumber
\end{eqnarray}
Exchanging the two sums, we find
\begin{equation}
c_{2k} = \sum_{n = 0}^k \left ( -\frac{1}{4} \right )^n \frac{(-\Delta_\sigma)_n \left ( \frac{1 + \Delta_\sigma}{3} \right )_n}{\left ( \frac{5 - 4\Delta_\sigma}{6} \right )_n n!} \frac{(2n)^2_{2k - 2n}}{(2n + 2k - 1)_{2k - 2n} (2k - 2n)!} {}_3F_2 \left ( \begin{tabular}{c} $2n - 2k, 2n + 2k - 1, n$ \\ $2n, 2n$ \end{tabular} ; 1 \right ) \label{4sigma-coeff-id}
\end{equation}
for the non-vanishing global block coefficients.\footnote{We derived this by applying a quadratic transformation to (\ref{4sigma-solutions}) which makes the Bose symmetry manifest. By leaving the function in its original form, or by applying Euler / Pfaff transformations, we can derive other sums that are non-trivially equivalent to (\ref{4sigma-coeff-id}).} Since $p = 2q$, this hypergeometric is in a form that can be treated with Watson's theorem. When we apply this to the $c_{2k}$, one pole in the numerator cancels another in the denominator.\footnote{Because we have not allowed for other poles, our expression for $c_{2k}$ will not be correct for $k = 0$. There is no need to treat this case separately as it is already clear that $c_0 = 1$.}
\begin{eqnarray}
{}_3F_2 \left ( \begin{tabular}{c} $2n - 2k, 2n + 2k - 1, n$ \\ $2n, 2n$ \end{tabular} ; 1 \right ) &=& \frac{(-1)^{n + k} \Gamma(\frac{1}{2})\Gamma(n+\frac{1}{2})\Gamma(2n)}{\Gamma(n)\Gamma(n-k+\frac{1}{2})\Gamma(2n + 2k - 1)} \frac{(k - 1)!}{(n - 1)!} \label{watson-result} \\
&=& \frac{(-1)^{n + k} \Gamma(2n)^2}{\Gamma(2n + 2k - 1)} \frac{\Gamma(n + k - \frac{1}{2})\Gamma(\frac{1}{2})}{\Gamma(n - k + \frac{1}{2})\Gamma(k + \frac{1}{2})} \lim_{\delta \rightarrow 0} \frac{2^{2k - 1} (k - 1)!}{[\Gamma(\delta)(\delta)_n]^2} \nonumber
\end{eqnarray}
Above, we can easily see three gamma functions that will cancel when we multiply by $\frac{(2n)^2_{2k - 2n}}{(2n + 2k - 1)_{2k - 2n} (2k - 2n)!}$. We have also written $\frac{1}{\Gamma(n)^2}$ in a limiting form for later convenience. After substituting (\ref{watson-result}), one must use the identities
\begin{equation}
\Gamma(x + n) = \Gamma(x)(x)_n \;\;\; , \;\;\; \Gamma(x - n) = (-1)^n \frac{\Gamma(x)}{(1 - x)_n} \label{gamma-ids}
\end{equation}
until each term of (\ref{4sigma-coeff-id}) only depends on $n$ through the Pochhammer symbol. This leads to
\begin{eqnarray}
c_{2k} &=& \binom{4k - 2}{2k - 1}^{-1} \lim_{\delta \rightarrow 0} \frac{1}{k(2k - 1)\Gamma(\delta)^2} {}_4F_3 \left ( \begin{tabular}{c} $-k, k - \frac{1}{2}, -\Delta_\sigma, \frac{1 + \Delta_\sigma}{3}$ \\ $\delta, \delta, \frac{5 - 4\Delta_\sigma}{6}$ \end{tabular} ; 1 \right ) \nonumber \\
&=& \binom{4k - 2}{2k - 1}^{-1} \frac{\Delta_\sigma (1 + \Delta_\sigma)}{5 - 4\Delta_\sigma} {}_4F_3 \left ( \begin{tabular}{c} $1 - k, k + \frac{1}{2}, 1 - \Delta_\sigma, \frac{4 + \Delta_\sigma}{3}$ \\ $1, 2, \frac{11 - 4\Delta_\sigma}{6}$ \end{tabular} ; 1 \right ) \; . \label{4sigma-result-id}
\end{eqnarray}

The $\epsilon$ block can be analyzed in the same way. Doing so will in fact be easier since we will not have to pass to the $\delta \rightarrow 0$ limit. Starting from
\begin{equation}
G^{\sigma\sigma\sigma\sigma}_{(1, 3)}(z) = z^{\frac{1 + 4\Delta_\sigma}{3}} (1 - z)^{-\frac{1 + 4\Delta_\sigma}{6}} {}_2F_1 \left ( \frac{1 + 2\Delta_\sigma}{2}, \frac{1 - 2\Delta_\sigma}{6} ; \frac{7 + 4\Delta_\sigma}{6} ; \frac{1}{4} \frac{z^2}{z - 1} \right ) \; , \label{quadratic-eps}
\end{equation}
we may use (\ref{power-law}) with $p = \frac{1 + 4\Delta_\sigma}{3} + 2n$ and $q = \frac{1 + 4\Delta_\sigma}{6} + n$. This leads to
\begin{eqnarray}
c_{2k} &=& \sum_{n = 0}^k \left ( -\frac{1}{4} \right )^n \frac{\left ( \frac{1 + 2\Delta_\sigma}{6} \right )_n \left ( \frac{1 + 2\Delta_\sigma}{2} \right )_n}{\left ( \frac{7 + 4\Delta_\sigma}{6} \right )_n n!} \frac{\left ( \frac{1 + 4\Delta_\sigma}{3} + 2n \right )^2_{2k - 2n}}{\left ( \frac{2 + 8\Delta_\sigma}{3} + 2n + 2k - 1 \right)_{2k - 2n} (2k - 2n)!} \nonumber \\
&& {}_3F_2 \left ( \begin{tabular}{c}  $2n - 2k, \frac{2 + 8\Delta_\sigma}{3} + 2n + 2k - 1, \frac{1 + 4\Delta_\sigma}{6} + n$ \\ $\frac{1 + 4\Delta_\sigma}{3} + 2n, \frac{1 + 4\Delta_\sigma}{3} + 2n$ \end{tabular} ; 1 \right ) \label{4sigma-coeff-eps}
\end{eqnarray}
which again allows us to use Watson's theorem. Employing (\ref{gamma-ids}) to perform the sum, we arrive at
\begin{equation}
c_{2k} = \frac{1}{4^k k!} \frac{\left ( \frac{1 + 4\Delta_\sigma}{6} \right )^2_k}{\left ( \frac{1 + 4\Delta_\sigma}{3} + k - \frac{1}{2} \right )_k} {}_4F_3 \left ( \begin{tabular}{c} $-k, \frac{1 + 4\Delta_\sigma}{3} + k - \frac{1}{2}, \Delta_\sigma + \frac{1}{2}, \frac{1 - 2\Delta_\sigma}{6}$ \\ $\frac{1 + 4\Delta_\sigma}{6}, \frac{1 + 4\Delta_\sigma}{6}, \frac{7 + 4\Delta_\sigma}{6}$ \end{tabular} ; 1 \right ) \; . \label{4sigma-result-eps}
\end{equation}

A useful observation about the hypergeometric functions in (\ref{4sigma-result-id}) and (\ref{4sigma-result-eps}) is that they both have a parameter excess of 1. This means that they are Wilson polynomials \cite{ks96}.\footnote{Note that we are using the normalization in \cite{w77}. There is another common normalization that makes the Wilson polynomial symmetric in $(a, b, c, d)$.}
\begin{equation}
P_n(a, b, c, d; x) = {}_4F_3 \left ( \begin{tabular}{c} $-n, n + a + b + c + d - 1, a + x, a - x$ \\ $a + b, a + c, a + d$ \end{tabular} ; 1 \right ) \label{wilson}
\end{equation}
Invoking this notation, our results are
\begin{eqnarray}
c^{\sigma\sigma(1, 1)\sigma\sigma}_{2n} &= \binom{4n - 2}{2n - 1}^{-1} \frac{\Delta_\sigma (1 + \Delta_\sigma)}{5 - 4\Delta_\sigma} P_{n - 1} \left ( \frac{7 - 2\Delta_\sigma}{6}, \frac{4 - 2\Delta_\sigma}{6}, -\frac{1 - 2\Delta_\sigma}{6}, \frac{5 + 2\Delta_\sigma}{6} ; \frac{1 + 4\Delta_\sigma}{6} \right ) \nonumber \\
c^{\sigma\sigma(1, 3)\sigma\sigma}_{2n} &= \frac{1}{4^n n!} \frac{\left ( \frac{1 + 4\Delta_\sigma}{6} \right )^2_n}{\left ( \frac{1 + 4\Delta_\sigma}{3} + n - \frac{1}{2} \right )_n} P_n \left ( \frac{1 + \Delta_\sigma}{3}, \frac{5 + 2\Delta_\sigma}{6}, -\frac{1 - 2\Delta_\sigma}{6}, -\frac{1 - 2\Delta_\sigma}{6} ; \frac{1 + 4\Delta_\sigma}{6} \right ) \; . \label{wilson-poly}
\end{eqnarray}
As with continuous Hahn polynomials, there is a recurrence relation that the Wilson polynomials satisfy.
\begin{eqnarray}
(x^2 - a^2) P_n(x) &=& A_n P_{n + 1}(x) - (A_n + B_n) P_n(x) + B_n P_{n - 1}(x) \nonumber \\
A_n &\equiv& \frac{(n + a + b + c + d - 1)(n + a + b)(n + a + c)(n + a + d)}{(2n + a + b + c + d - 1)(2n + a + b + c + d)} \nonumber \\
B_n &\equiv& \frac{n(n + b + c - 1)(n + b + d - 1)(n + c + d - 1)}{(2n + a + b + c + d - 2)(2n + a + b + c + d - 1)} \label{wilson-recursion}
\end{eqnarray}
Appendix A uses this recursion to solve for the asymptotic behaviour of Wilson polynomials and prove that the ones in (\ref{wilson-poly}) are positive.

\subsection{Some global block coefficients in the other correlators}
Given our success at explaining the one-correlator results, the next logical step is to find the global block decompositions applicable to three correlators. We will start with $\left < \epsilon\epsilon\epsilon\epsilon \right >$. Although this is another four-point function of identical scalars, the main difference compared to $\left < \sigma\sigma\sigma\sigma \right >$ is the lack of a closed-form Virasoro block. Because a generic $\epsilon$ only has a null descendant at level 3, we will have to work with a third-order BPZ equation which does not have simple solutions analogous to (\ref{4sigma-solutions}). This makes the formula (\ref{sieve2}) important for finding low-lying OPE coefficients. The singular vector that must annihilate
\begin{eqnarray}
&& \left < \epsilon(z_1, \bar{z}_1) \epsilon(z_2, \bar{z}_2) \epsilon(z_3, \bar{z}_3) \epsilon(z_4, \bar{z}_4) \right > = \frac{G(z, \bar{z})}{|z_{12} z_{34}|^{2\Delta_\epsilon}} \label{upper-4epsilon} \\
&& G(z, \bar{z}) = G^{\epsilon\epsilon\epsilon\epsilon}_{(1, 1)}(z) G^{\epsilon\epsilon\epsilon\epsilon}_{(1, 1)}(\bar{z}) + C^{(1, 3)}_{(1, 3)(1, 3)} G^{\epsilon\epsilon\epsilon\epsilon}_{(1, 3)}(z) G^{\epsilon\epsilon\epsilon\epsilon}_{(1, 3)}(\bar{z}) + C^{(1, 5)}_{(1, 3)(1, 3)} G^{\epsilon\epsilon\epsilon\epsilon}_{(1, 5)}(z) G^{\epsilon\epsilon\epsilon\epsilon}_{(1, 5)}(\bar{z}) \nonumber
\end{eqnarray}
has four terms that can be read off from (\ref{singular}). The Virasoro commutation relations reduce it to the three term expression
\begin{equation}
\frac{4}{\Delta_\epsilon(\Delta_\epsilon + 2)} \mathcal{L}_{-1}^3 - \frac{4}{\Delta_\epsilon} \mathcal{L}_{-2} \mathcal{L}_{-1} + \mathcal{L}_{-3} \; . \label{level3}
\end{equation}

\begin{figure}[t!]
\centering
\subfloat[][$c^{\epsilon\epsilon(1, 1)\epsilon\epsilon}_{2n}$]{\includegraphics[scale=0.45]{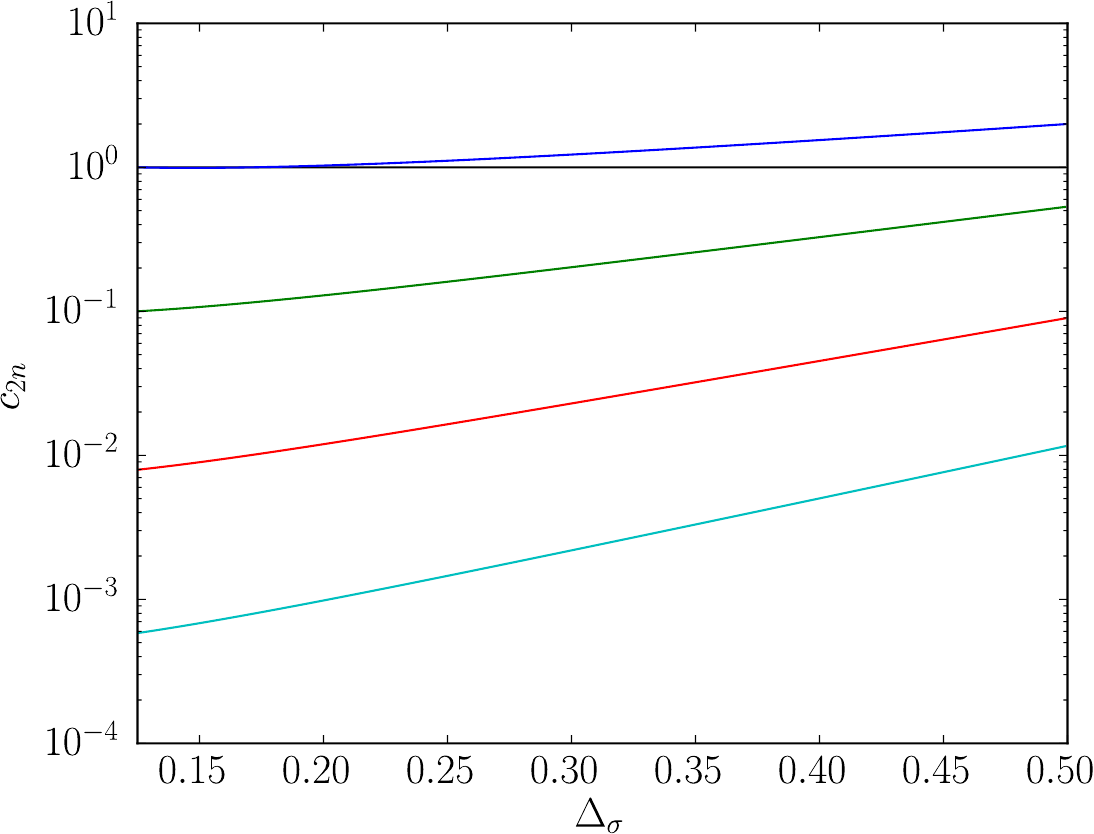}}
\subfloat[][$c^{\epsilon\epsilon(1, 3)\epsilon\epsilon}_{2n}$]{\includegraphics[scale=0.45]{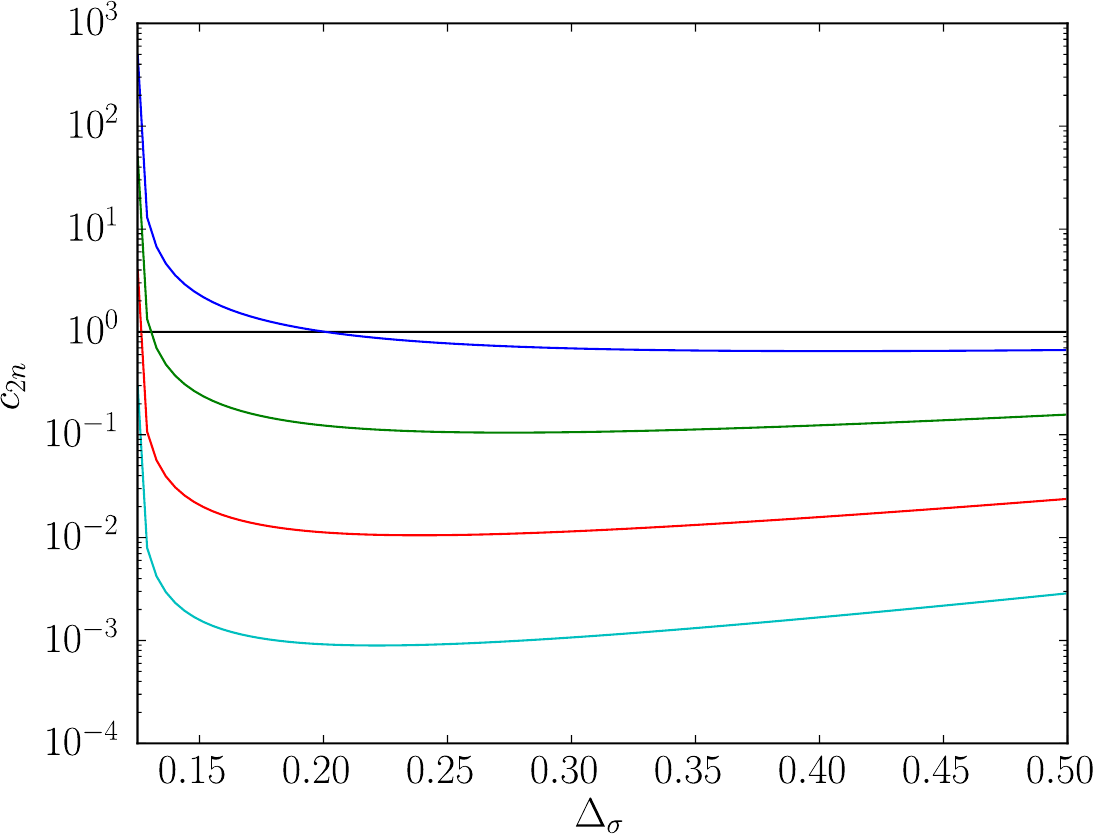}}\\
\subfloat[][$c^{\epsilon\epsilon(1, 5)\epsilon\epsilon}_{2n}$]{\includegraphics[scale=0.45]{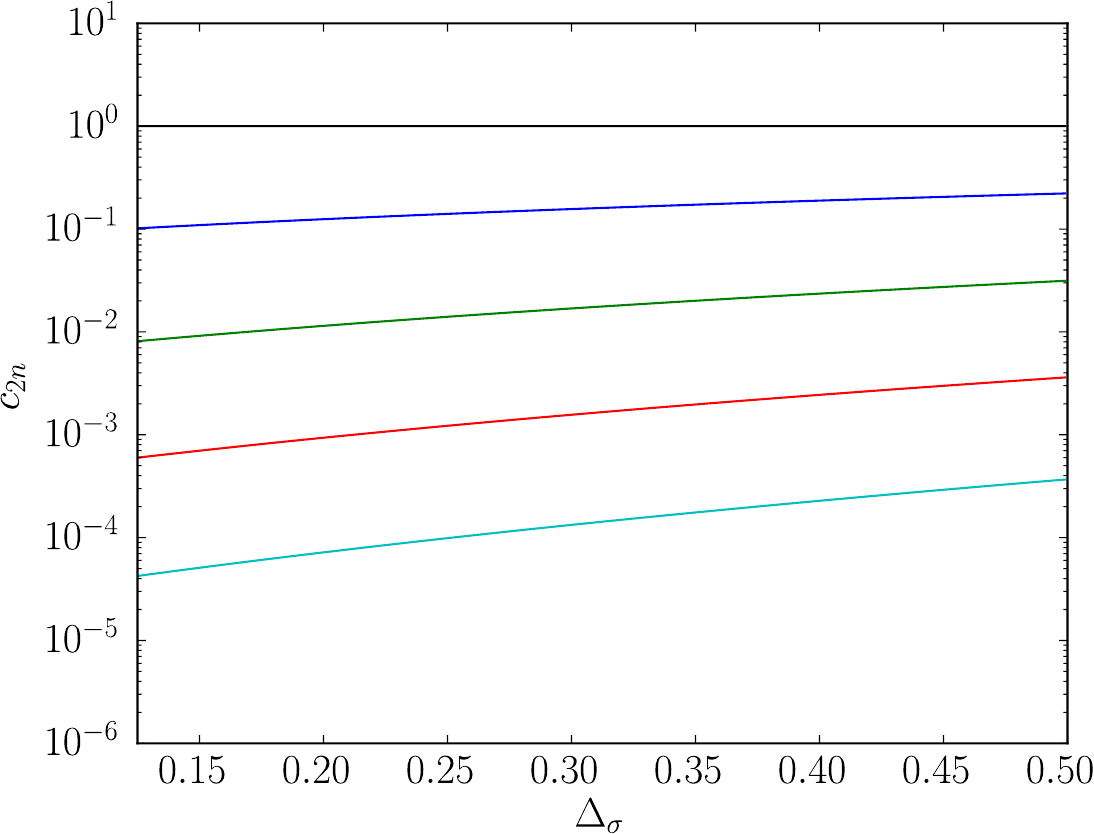}}
\;\;\; \subfloat[][Legend]{\includegraphics[scale=0.95]{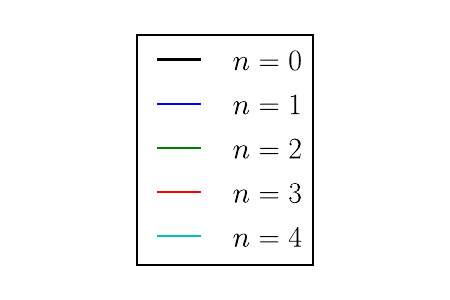}\vspace{0.5cm}} \;\;\;
\caption{Log-scale plots of $c^{\epsilon\epsilon(1, s)\epsilon\epsilon}_{2n}$ showing that the first five are all positive on the interval $\frac{1}{8} \leq \Delta_\sigma \leq \frac{1}{2}$.}
\label{4epsilon-plot}
\end{figure}

The null state condition that follows from (\ref{upper-4epsilon}) and (\ref{level3}) is
\begin{eqnarray}
&& 4z^2(z - 1)^3 \frac{\partial^3 G}{\partial z^3} + 4z(z - 1)^2[(4 - \Delta_\epsilon)z + 2\Delta_\epsilon - 2] \frac{\partial^2 G}{\partial z^2} \nonumber \\
&& - (z - 1)[(\Delta_\epsilon^2 + 10\Delta_\epsilon - 8)z^2 + (3\Delta_\epsilon^2 - 14\Delta_\epsilon + 8)z - 3\Delta_\epsilon(\Delta_\epsilon - 2)] \frac{\partial G}{\partial z} \nonumber \\
&& + \Delta_\epsilon^2(\Delta_\epsilon + 2) z(z - 2) G = 0 \; . \label{4epsilon-ode}
\end{eqnarray}
To approximate the Virasoro blocks that solve this, it will be helpful to use the Frobenius method.\footnote{Even though they describe exchanged weights of $h_{1, 1}$ and $h_{1, 3}$, we cannot reuse either of the expressions in (\ref{4sigma-solutions}). Unlike global blocks which only see dimension differences, Virasoro blocks depend on the external weights individually \cite{s09}.} Inserting the ansatz $G(z) = \sum_{k = -\infty}^\infty b_k z^{r + k}$, we may reindex the sum so that all terms carry the same power of $z$. This gives a recurrence relation for the coefficients.
\begin{eqnarray}
&& [-4(k + r)(k + r - 1) + 8(\Delta_\epsilon - 1)(k + r) + 3\Delta_\epsilon(2 - \Delta_\epsilon)](k + r + 1)b_{k + 1} \nonumber \\
&& + 2[6(k + r - 1)(k + r - 2) + 2(8 - 5\Delta_\epsilon)(k + r - 1) + 3\Delta_\epsilon^2 - 10\Delta_\epsilon + 4](k + r)b_k \nonumber \\
&& - 2 \left [ 6(k + r - 1)(k + r - 2)(k + r - 3) - 4(2\Delta_\epsilon - 5)(k + r - 1)(k + r - 2) \right. \nonumber \\
&& \left. + (\Delta_\epsilon^2 - 12\Delta_\epsilon + 8)(k + r - 1) + \Delta_\epsilon^2(\Delta_\epsilon + 2) \right ] b_{k - 1} \nonumber \\
&& + \left [ 4(k + r - 2)(k + r - 3)(k + r - 4) + 4(4 - \Delta_\epsilon)(k + r - 2)(k + r - 3) \right. \nonumber \\
&& \left. - (\Delta_\epsilon^2 + 10\Delta_\epsilon - 8)(k + r - 2) + \Delta_\epsilon^2(\Delta_\epsilon + 2) \right ] b_{k - 2} = 0 \label{4epsilon-recursion}
\end{eqnarray}
We will set $b_0 = 1$ and $b_k = 0$ for all $k < 0$. The values of $r$ that make this consistent (called roots of the indicial equation) are $h_{1, 1}$, $h_{1, 3}$ and $h_{1, 5}$ as expected. We find them by demanding that $b_0$ drop out of (\ref{4epsilon-recursion}) when $k = -1$. For each value of $r$, it is straightforward to iterate (\ref{4epsilon-recursion}) and then feed the results into (\ref{sieve2}). Some of the OPE coefficients that follow from this are written in Table \ref{4epsilon-table}. Due to the appearance of the upper bound (\ref{upper-bound}) in the three-correlator bootstrap, we expect the coefficients to be positive when $1 \leq \Delta_\epsilon \leq 2$, at least up to some high order. Figure \ref{4epsilon-plot} shows that this is indeed the case.
\begin{table}[h]
\centering
\begin{tabular}{|l|l|}
\hline
$n$ & $c_{2n}^{\epsilon\epsilon(1, 1)\epsilon\epsilon}$ \\
\hline
$0$ & $1$ \\
$1$ & $-\frac{\Delta_\epsilon^2 (\Delta_\epsilon + 2)}{(\Delta_\epsilon - 4)(3\Delta_\epsilon - 2)}$ \\
$2$ & $\frac{\Delta_\epsilon^2 (\Delta_\epsilon + 2)^2 [5\Delta_\epsilon + 2]}{30(\Delta_\epsilon - 8)(\Delta_\epsilon - 4)(3\Delta_\epsilon - 2)}$ \\
\hline
$n$ & $c_{2n}^{\epsilon\epsilon(3, 1)\epsilon\epsilon}$ \\
\hline
$0$ & $1$ \\
$1$ & $\frac{\Delta_\epsilon^2 (\Delta_\epsilon + 2) [5\Delta_\epsilon + 2]}{16(\Delta_\epsilon - 1)(\Delta_\epsilon + 1)(\Delta_\epsilon + 4)}$ \\
$2$ & $\frac{\Delta_\epsilon^2 (\Delta_\epsilon + 2) [25\Delta_\epsilon^5 + 167\Delta_\epsilon^4 - 66\Delta_\epsilon^3 - 1904\Delta_\epsilon^2 - 2752\Delta_\epsilon - 384]}{512(\Delta_\epsilon - 3)(\Delta_\epsilon - 1)(\Delta_\epsilon + 3)(\Delta_\epsilon + 4)(\Delta_\epsilon + 5)(\Delta_\epsilon + 8)}$ \\
\hline
$n$ & $c_{2n}^{\epsilon\epsilon(5, 1)\epsilon\epsilon}$ \\
\hline
$0$ & $1$ \\
$1$ & $\frac{\Delta_\epsilon (\Delta_\epsilon + 2) [7\Delta_\epsilon + 6]}{48(\Delta_\epsilon + 1)(\Delta_\epsilon + 3)}$ \\
$2$ & $\frac{\Delta_\epsilon (\Delta_\epsilon + 2) [441\Delta_\epsilon^5 + 5121\Delta_\epsilon^4 + 20732\Delta_\epsilon^3 + 37796\Delta_\epsilon^2 + 31056\Delta_\epsilon + 8640]}{1536(\Delta_\epsilon + 3)(\Delta_\epsilon + 5)(3\Delta_\epsilon + 5)(3\Delta_\epsilon + 7)(3\Delta_\epsilon + 10)}$ \\
\hline
\end{tabular}
\caption{The first three global block coefficients in the $(1, 1)$, $(1, 3)$ and $(1, 5)$ contributions found by taking the $\epsilon \times \epsilon$ OPE twice.}
\label{4epsilon-table}
\end{table}

We should also be able to find positive squared OPE coefficients in the mixed correlator. The four-point function
\begin{eqnarray}
&& \left < \sigma(z_1, \bar{z}_1) \epsilon(z_2, \bar{z}_2) \sigma(z_3, \bar{z}_3) \epsilon(z_4, \bar{z}_4) \right > = \left ( \frac{|z_{24}|}{|z_{13}|} \right )^{\Delta_{\sigma\epsilon}} \frac{G(z, \bar{z})}{|z_{12}|^{\Delta_\sigma + \Delta_\epsilon} |z_{34}|^{\Delta_\sigma + \Delta_\epsilon}} \nonumber \\
&& G(z, \bar{z}) = C^{(1, 2)}_{(1, 2)(1, 3)} G^{\sigma\epsilon\sigma\epsilon}_{(1, 2)}(z) G^{\sigma\epsilon\sigma\epsilon}_{(1, 2)}(\bar{z}) + C^{(1, 4)}_{(1, 2)(1, 3)} G^{\sigma\epsilon\sigma\epsilon}_{(1, 4)}(z) G^{\sigma\epsilon\sigma\epsilon}_{(1, 4)}(\bar{z}) \label{upper-2epsilon2sigma}
\end{eqnarray}
satisfies second-order and third-order BPZ equations.

\begin{figure}[b!]
\centering
\subfloat[][$c^{\sigma\epsilon(1, 2)\sigma\epsilon}_{n}$]{\includegraphics[scale=0.45]{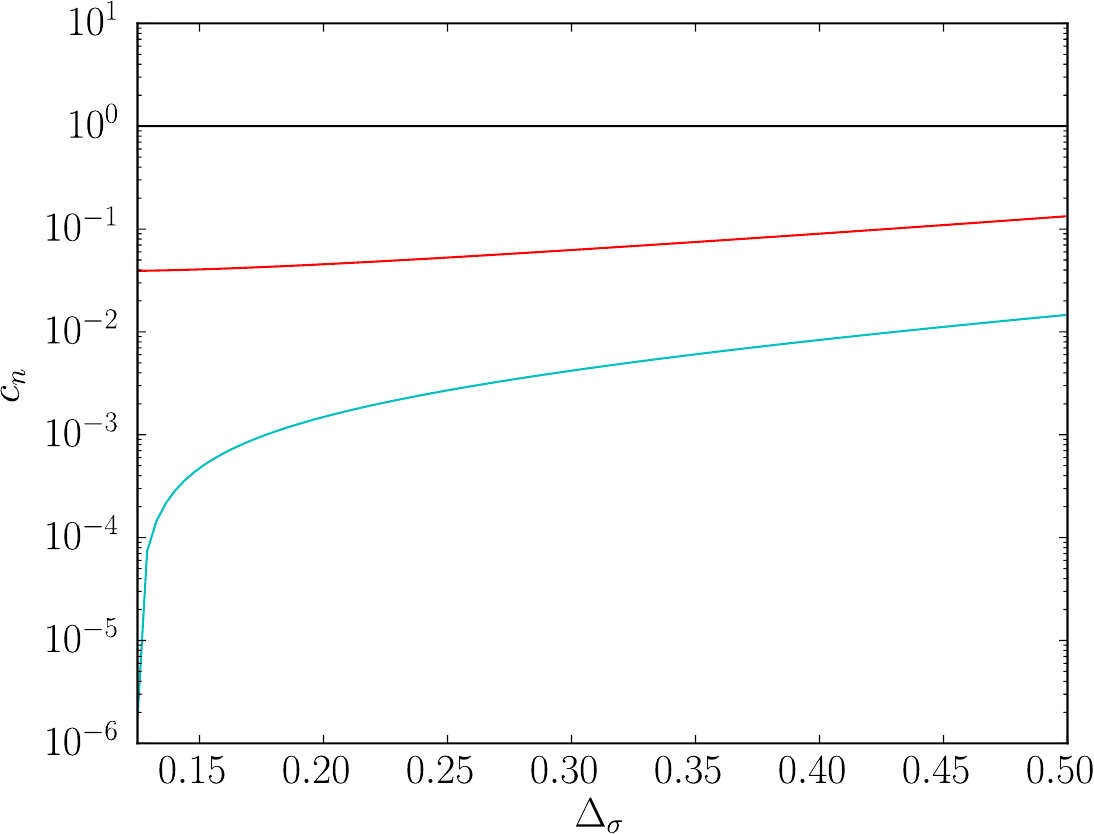}}
\subfloat[][$c^{\sigma\epsilon(1, 4)\sigma\epsilon}_{n}$]{\includegraphics[scale=0.45]{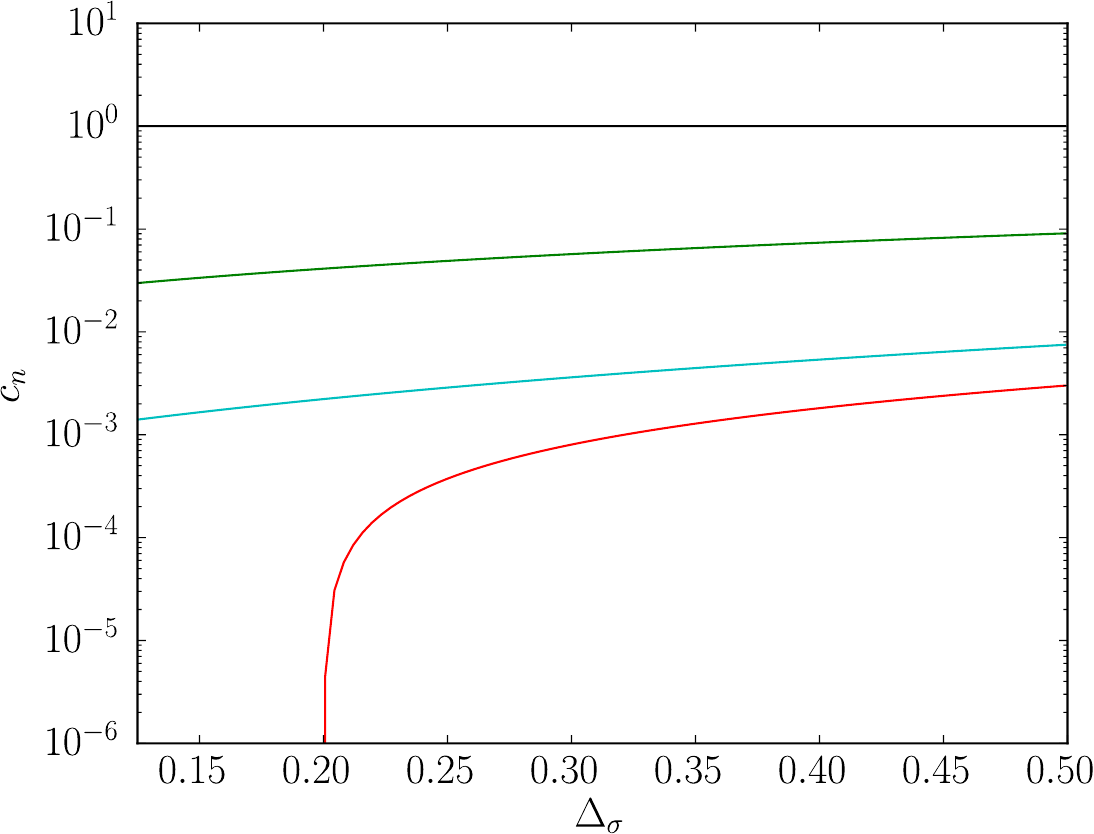}}
\caption{Log-scale plots of the $c^{\sigma\epsilon(1, s)\sigma\epsilon}_{n}$ that are non-zero. The legend is the same as that of Figure \ref{4epsilon-plot}.}
\label{2epsilon2sigma-plot}
\end{figure}

For simplicity, we will consider the second-order equation
\begin{eqnarray}
&& \frac{3}{2} z^2(z - 1)^2 \frac{\partial^2 G}{\partial z^2} + \frac{1}{2} z(z - 1) [(2 - 7\Delta_\sigma)z + 9\Delta_\sigma] \frac{\partial G}{\partial z} \nonumber \\
&& + \frac{1}{24} [3\Delta_\sigma(11\Delta_\sigma + 2)z^2 - 2(5\Delta_\sigma + 2)(11\Delta_\sigma + 2)z + 9\Delta_\sigma(5\Delta_\sigma + 2)]G = 0 \label{2epsilon2sigma-ode}
\end{eqnarray}
which has the recurrence relation
\begin{eqnarray}
&& 9[4(k + r + 1)(k + r) - 12\Delta_\sigma(k + r + 1) + \Delta_\sigma(5\Delta_\sigma + 2)] b_{k + 1} \label{2epsilon2sigma-recursion} \\
&& - 2[36(k + r)(k + r - 1) - 12(8\Delta_\sigma - 1)(k + r) + (5\Delta_\sigma + 2)(11\Delta_\sigma + 2)] b_k \nonumber \\
&& + 3[12(k + r - 1)(k + r - 2) + 4(2 - 7\Delta_\sigma)(k + r - 1) + \Delta_\sigma(11\Delta_\sigma + 2)] b_{k - 1} = 0 \; . \nonumber
\end{eqnarray}
It is easily seen that $r \in \{ h_{1, 2}, h_{1, 4} \}$ is the solution of the indicial equation for (\ref{2epsilon2sigma-recursion}). Because the product $\sigma \times \epsilon$ no longer has Bose symmetry, the procedure by which we extract the conformal block expansion this time is somewhat different. We must include dimension differences in the hypergeometric function and sum over all integers whether even or odd.
\begin{eqnarray}
G(z) &=& \sum_{n = 0}^\infty (-1)^n c_n z^{r + n} {}_2F_1 \left ( r + n - \frac{1}{2} \Delta_{\sigma\epsilon}, r + n + \frac{1}{2} \Delta_{\sigma\epsilon} ; 2(r + n); z \right ) \nonumber \\
&=& \sum_{n = 0}^{\infty} \sum_{m = 0}^{\infty} (-1)^n c_n \frac{\left ( r + n - \frac{1}{2} \Delta_{\sigma\epsilon} \right )_m \left ( r + n + \frac{1}{2} \Delta_{\sigma\epsilon} \right )_m}{(2r + 2n)_m} \frac{z^{r + n + m}}{m!} \nonumber \\
&=& \sum_{k = 0}^\infty \sum_{n = 0}^k (-1)^n c_n \frac{\left ( r + n - \frac{1}{2} \Delta_{\sigma\epsilon} \right )_{k - n} \left ( r + n + \frac{1}{2} \Delta_{\sigma\epsilon} \right )_{k - n}}{(2r + 2n)_{k - n}} \frac{z^{r + k}}{(k - n)!} \label{sieve3}
\end{eqnarray}
The lower triangular system from this leads to the recursion
\begin{equation}
(-1)^k c_k = b_k - \sum_{n = 0}^{k - 1} (-1)^n c_n \frac{\left ( r + n - \frac{1}{2} \Delta_{\sigma\epsilon} \right )_{k - n} \left ( r + n + \frac{1}{2} \Delta_{\sigma\epsilon} \right )_{k - n}}{(2r + 2n)_{k - n} (k - n)!} \; . \label{sieve3}
\end{equation}
\begin{table}[t]
\centering
\begin{tabular}{|l|l|}
\hline
$n$ & $c_n^{\sigma\epsilon(1, 2)\sigma\epsilon}$ \\
\hline
$0$ & $1$ \\
$1$ & $0$ \\
$2$ & $0$ \\
$3$ & $-\frac{(\Delta_\sigma + 1)(4\Delta_\sigma + 1)^2}{729\Delta_\sigma(\Delta_\sigma - 1)(\Delta_\sigma + 2)} (5\Delta_\sigma + 2)$ \\
$4$ & $-\frac{4(\Delta_\sigma + 1)^2(4\Delta_\sigma + 1)^2}{729\Delta_\sigma(\Delta_\sigma + 3)(\Delta_\sigma + 6)(2\Delta_\sigma - 3)} (8\Delta_\sigma - 1)$ \\
\hline
$n$ & $c_n^{\sigma\epsilon(1, 4)\sigma\epsilon}$ \\
\hline
$0$ & $1$ \\
$1$ & $0$ \\
$2$ & $\frac{\Delta_\sigma + 1}{9(2\Delta_\sigma + 3)(5\Delta_\sigma + 3)} (2\Delta_\sigma + 1)(10\Delta_\sigma + 1)$ \\
$3$ & $\frac{4(\Delta_\sigma + 1)(5\Delta_\sigma + 2)}{729(\Delta_\sigma + 2)(5\Delta_\sigma + 4)(5\Delta_\sigma + 6)} (5\Delta_\sigma - 1)(7\Delta_\sigma + 4)$ \\
$4$ & $\frac{(\Delta_\sigma + 1)^2(5\Delta_\sigma + 2)}{81(5\Delta_\sigma + 7)(5\Delta_\sigma + 8)} (10\Delta_\sigma + 1)$ \\
\hline
\end{tabular}
\caption{The first five global block coefficients in the $(1, 2)$ and $(1, 4)$ contributions found by taking the $\sigma \times \epsilon$ OPE twice.}
\label{2epsilon2sigma-table}
\end{table}
Some low-lying global block coefficients found with (\ref{sieve3}) are listed in Table \ref{2epsilon2sigma-table}. While all of them are non-negative above the tricritical Ising value $\Delta_\sigma = \frac{1}{5}$, there is actually one that takes on negative values for $\frac{1}{8} < \Delta_\sigma < \frac{1}{5}$ as shown in Figure \ref{2epsilon2sigma-plot}. We may explain this by noticing that $\phi_{1, 4}$ is also $\phi_{3, 1}$ in the minimal model $\mathcal{M}(5, 4)$. This field has exactly one quasiprimary descendant at level 3. The presence of a null state is therefore enough to conclude that $c^{\sigma\epsilon(1, 4)\sigma\epsilon}_3 = 0$.

\begin{table}[h]
\centering
\begin{tabular}{|l|l|}
\hline
$n$ & $c_{2n}^{\sigma\sigma(1, 1)\epsilon\epsilon}$ \\
\hline
$0$ & $1$ \\
$1$ & $\frac{(\Delta_\sigma + 1) (4\Delta_\sigma + 1)}{3(5 - 4\Delta_\sigma)}$ \\
$2$ & $\frac{2(\Delta_\sigma + 1)^2 (4\Delta_\sigma + 1) [5\Delta_\sigma + 2]}{45(5 - 4\Delta_\sigma)(11 - 4\Delta_\sigma)}$ \\
\hline
$n$ & $c_{2n}^{\sigma\sigma(1, 3)\epsilon\epsilon}$ \\
\hline
$0$ & $1$ \\
$1$ & $\frac{2(\Delta_\sigma + 1) (4\Delta_\sigma + 1) [5\Delta_\sigma + 2]}{3(7 + 4\Delta_\sigma)(5 + 8\Delta_\sigma)}$ \\
$2$ & $\frac{(\Delta_\sigma + 1)^2 (4\Delta_\sigma + 1) [400\Delta_\sigma^3 + 1548\Delta_\sigma^2 + 1644\Delta_\sigma + 361]}{18(7 + 4\Delta_\sigma)(13 + 4\Delta_\sigma)(11 + 8\Delta_\sigma)(17 + 8\Delta_\sigma)}$ \\
\hline
\end{tabular}
\caption{The first three global block coefficients in the $(1, 1)$ and $(1, 3)$ contributions found by taking the $\sigma \times \sigma$ and $\epsilon \times \epsilon$ OPEs.}
\label{2sigma2epsilon-table}
\end{table}
Alternatively, we can take the OPE in the other channel by permuting operator positions in (\ref{upper-2epsilon2sigma}). This yields the global block coefficients in Table \ref{2sigma2epsilon-table} which are related to the $\lambda_{\sigma\sigma\mathcal{O}} \lambda_{\epsilon\epsilon\mathcal{O}}$ CFT data. There is no reason for these numbers to be positive, but we expect
\begin{equation}
c^{\sigma\sigma(1, s)\sigma\sigma}_{2n} c^{\epsilon\epsilon(1, s)\epsilon\epsilon}_{2n} \geq \left ( c^{\sigma\sigma(1, s)\epsilon\epsilon}_{2n} \right )^2 \label{agm}
\end{equation}
to be obeyed.\footnote{For another quick check of our results, $\Delta_\epsilon^2 c^{\sigma\sigma(1, 1)\sigma\sigma}_0 c^{\sigma\sigma(1, 1)\sigma\sigma}_2 = \Delta_\sigma^2 c^{\epsilon\epsilon(1, 1)\epsilon\epsilon}_0 c^{\epsilon\epsilon(1, 1)\epsilon\epsilon}_2$ holds as it must by the Ward identity.} When $\sigma \times \sigma$ and $\epsilon \times \epsilon$ share a set of operators $\mathcal{S}$ with the same quantum numbers, the left-hand side is a product of sums. The right-hand side is a sum of products and therefore smaller by the arithmetic-geometric mean inequality. A departure from (\ref{agm}) would be a violation of unitarity since the matrix
\begin{equation}
\sum_{\mathcal{O} \in \mathcal{S}} \left [ \lambda_{\sigma\sigma\mathcal{O}} \; \lambda_{\epsilon\epsilon\mathcal{O}} \right ] \left [ \begin{tabular}{c} $\lambda_{\sigma\sigma\mathcal{O}}$ \\ $\lambda_{\epsilon\epsilon\mathcal{O}}$ \end{tabular} \right ] = \left [ \begin{tabular}{cc} $c^{\sigma\sigma(1, s)\sigma\sigma}_{2n}$ & $c^{\sigma\sigma(1, s)\epsilon\epsilon}_{2n}$ \\ $c^{\sigma\sigma(1, s)\epsilon\epsilon}_{2n}$ & $c^{\epsilon\epsilon(1, s)\epsilon\epsilon}_{2n}$ \end{tabular} \right ] \label{agm-matrix}
\end{equation}
would not be positive-definite.

\subsection{Virasoro block coefficients}
Our analysis so far has been focused on $c^{\sigma\sigma(1, s)\sigma\sigma}_{2n}$, $c^{\epsilon\epsilon(1, s)\epsilon\epsilon}_{2n}$ and $c^{\sigma\epsilon(1, s)\sigma\epsilon}_{n}$, which encode the decomposition of a Virasoro block into $\mathfrak{sl}(2)$ blocks. With the sole exception of $c^{\sigma\epsilon(1, 4)\sigma\epsilon}_3$, which we could imagine to have a small effect, we have found that these coefficients are non-negative when $\frac{1}{8} \leq \Delta_\sigma \leq \frac{1}{2}$. However, this is only meaningful if the same property holds for the structure constants that unite holomorphic and anti-holomorphic halves of a four-point function. A single $C_{(r_1,s_1)(r_2,s_2)}^{(r_3,s_3)} < 0$ for instance would give rise to an infintie number of negative contributions in $\left < \phi_{r_1, s_1} \phi_{r_2, s_2} \phi_{r_1, s_1} \phi_{r_2, s_2} \right >$, severely complicating the interpretation of Figure \ref{single-multi}.

In a given correlation function, the Virasoro block coefficients that appear must be compatible with crossing symmetry and single-valuedness. By briefly reviewing the method of \cite{df84, df85a, df85b}, we will show that this condition is enough to fix them uniquely. For definiteness, consider the $\left < \epsilon\epsilon\epsilon\epsilon \right >$ correlator in the generalized minimal model with central charge $\frac{13}{21}$. The holomorphic part of this function comes from the kernel of the operator (\ref{level3}), which is three-dimensional. Solving (\ref{4epsilon-recursion}) has given us a basis for this kernel in which each function vanishes at $z = 0$. We could equally well have chosen any of the regular singular points $0$, $1$ and $\infty$, corresponding to the $s$, $t$ and $u$ channels for Virasoro blocks. When going from $z = 0$ to $z = 1$, there is a special matrix $F$ called the \textit{crossing matrix} (or \textit{fusion matrix}) that accomplishes $G^a(z) = F^a_{\;\;\; b} G^b(1 - z)$. It is a special case of the \textit{crossing kernel} which applies to theories with a continuous spectrum.\footnote{Several interchangeable terms have proliferated over the years. When replacing the $z \mapsto 1 - z$ map with $z \mapsto \frac{1}{z}$, the words \textit{crossing} and \textit{fusion} become \textit{exchange} and \textit{braiding}. Outside the CFT context, one says that a linear ODE has a \textit{monodromy matrix} or \textit{connection matrix}. For ODEs that have less structure than a BPZ equation, finding this matrix is often a difficult problem.} Since it represents a change of basis, the expression for $F$ is unique. In this case, it is given by
\begin{equation}
F = \left [
\begin{tabular}{ccc}
$0.7422$ & $0.3124$ & $-0.1563$ \\
$2.3762$ & $-1.8795$ & $1.4405$ \\
$1.8760$ & $-2.2733$ & $2.1372$
\end{tabular}
\right ] \; , \label{crossing-matrix}
\end{equation}
which was obtained in \cite{fgp90} via the Coulomb gas formalism. Later reviews are \cite{pss14, efr16}. In the following, we will use $G^1$, $G^2$ and $G^3$ to denote $G^{\epsilon\epsilon\epsilon\epsilon}_{(1, 1)}$, $G^{\epsilon\epsilon\epsilon\epsilon}_{(1, 3)}$ and $G^{\epsilon\epsilon\epsilon\epsilon}_{(1, 5)}$ respectively.

The constraints from (\ref{crossing-matrix}) are best phrased in terms of a metric on the space of conformal blocks; $G(z, \bar{z}) = W_{ab} G^a(z) G^b(\bar{z})$. We have three conditions that $W_{ab}$ must satisfy:
\begin{enumerate}
\item $W_{11} = 1$
\item $W_{cd} = W_{ab} F^a_{\;\;\; c} F^b_{\;\;\; d}$
\item $W_{ab} = 0$ for $a \neq b$
\end{enumerate}
The first of these is an obvious consequence of having unit-normalized operators. The second comes from writing $G(1 - z, 1 - \bar{z}) = W_{ab} G^a(1 - z) G^b(1 - \bar{z})$ in terms of $s$-channel blocks and setting it equal to $G(z, \bar{z})$. The third ensures that the four-point function has trivial monodromy under $z \mapsto e^{2\pi i}z$, $\bar{z} \mapsto e^{-2\pi i}\bar{z}$. This rules out a non-diagonal metric since we only find $h_{r,s}$ weights that differ by integers at special values of $m$.\footnote{It is well known that this happens in discrete minimal models. Some of the non-diagonal theories so constructed also satisfy the stronger requirement of modular invariance.} We build solutions out of the left eigenvectors of $F$,
\begin{equation}
v^1 = \left [
\begin{tabular}{c} $0.9537$ \\ $-0.1157$ \\ $0.2776$ \end{tabular}
\right ] \; , \;
v^2 = \left [
\begin{tabular}{c} $0.9805$ \\ $0.1758$ \\ $-0.0879$ \end{tabular}
\right ] \; , \;
v^3 = \left [
\begin{tabular}{c} $0.5938$ \\ $-0.7196$ \\ $0.3600$ \end{tabular}
\right ] \; , \label{crossing-vectors}
\end{equation}
which have eigenvalues of $1$, $1$ and $-1$ respectively. Clearly $F^2 = 1$, which follows from $z \leftrightarrow 1 - z$ being an involution, requires all eigenvalues to be $\pm 1$. The following form is invariant under two multiplications by $F$:
\begin{equation}
W_{ab} = c_{11} v^1_a v^1_b + c_{12} v^1_a v^2_b + c_{21} v^2_a v^1_b + c_{22} v^2_a v^2_b + c_{33} v^3_a v^3_b \; . \label{w-form}
\end{equation}
There are six off-diagonal components of $W$ that need to vanish. If we set $c_{21} = c_{12}$, (\ref{w-form}) becomes manifestly symmetric and we can use three more parameters, $c_{11}$, $c_{12}$ and $c_{22}$, to make $W$ diagonal. As the single remaining parameter, $c_{33}$ is used to rescale $W$ so that its leading component is $1$.

\begin{figure}[h]
\centering
\includegraphics[scale=0.6]{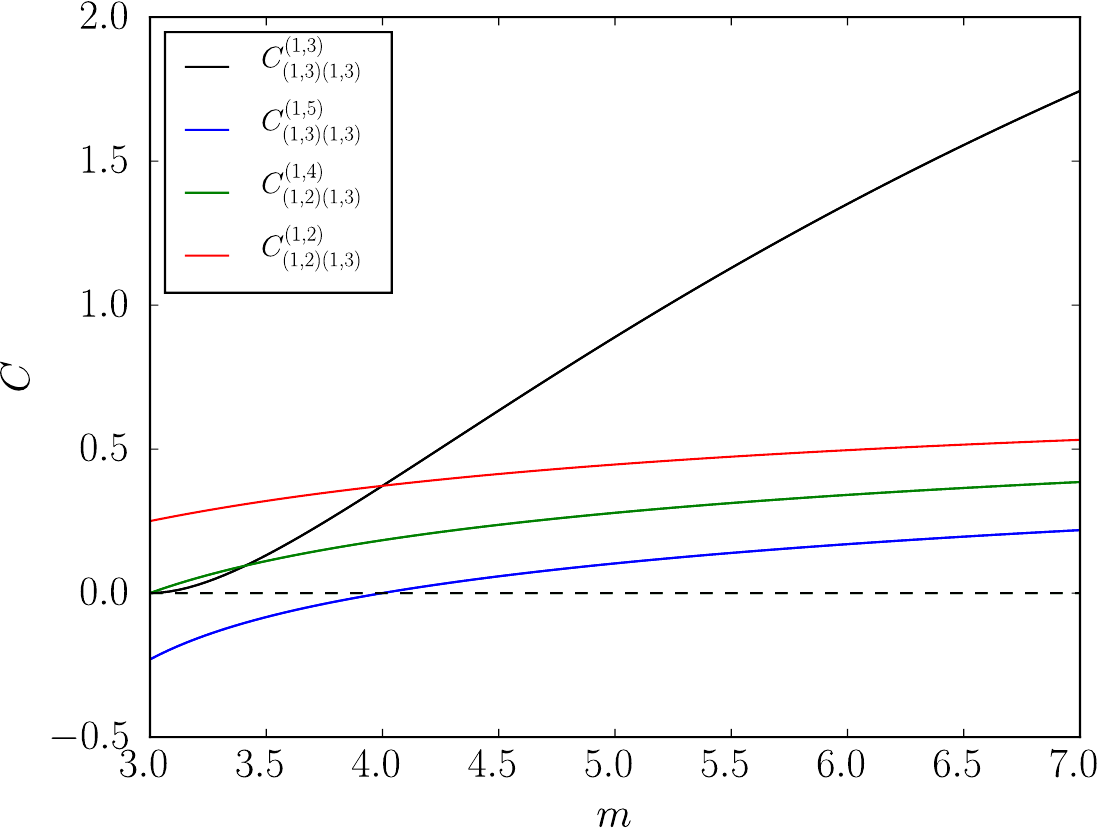}
\caption{Low-lying squared OPE coefficients for Virasoro primaries in the generalized minimal models. Between the Ising model at $m = 3$ ($\Delta_\sigma = \frac{1}{8}$) and the tricritical Ising model at $m = 4$ ($\Delta_\sigma = \frac{1}{5}$), there is one that becomes negative.}
\label{ope-plot}
\end{figure}

We may summarize this discussion by stating that there is no freedom in the three-correlator bootstrap equations once the Virasoro blocks involving $\sigma$ and $\epsilon$ are specified. Knowledge of these blocks fully determines $\left < \sigma\sigma\sigma\sigma \right >$, $\left < \sigma\sigma\epsilon\epsilon \right >$ and $\left < \epsilon\epsilon\epsilon\epsilon \right >$, whether or not we demand consistency conditions for other correlators. This means that the generalized minimal model structure constants are the only valid choices for the coefficients in (\ref{upper-4sigma}), (\ref{upper-4epsilon}) and (\ref{upper-2epsilon2sigma}). Defining $\gamma(x) = \Gamma(x) / \Gamma(1 - x)$ and $t = \frac{m}{m + 1}$, the expressions we need are
\begin{eqnarray}
C^{(1, 3)}_{(1, 3)(1, 3)} &=& \gamma(t)^3\gamma(4t - 1)^2\gamma(1 - 2t)^3\gamma(2 - 2t)\gamma(2 - 3t) \nonumber \\
C^{(1, 5)}_{(1, 3)(1, 3)} &=& \gamma(t)\gamma(2t)\gamma(4t - 1)\gamma(5t - 1)\gamma(1 - 3t)\gamma(1 - 4t)\gamma(2 - 2t)\gamma(2 - 3t) \nonumber \\
C^{(1, 4)}_{(1, 2)(1, 3)} &=& \gamma(t)\gamma(4t - 1)\gamma(1 - 3t)\gamma(2 - 2t) \nonumber \\
C^{(1, 2)}_{(1, 2)(1, 3)} &=& C^{(1, 3)}_{(1, 2)(1, 2)} = \gamma(t)\gamma(3t - 1)\gamma(1 - 2t)\gamma(2 - 2t) \; . \label{known-ope-coeffs}
\end{eqnarray}
The last coefficient in (\ref{known-ope-coeffs}) is clearly the one that was rederived in \cite{lrv13}. Plotting these in Figure \ref{ope-plot}, we see that $C^{(1, 5)}_{(1, 3)(1, 3)} < 0$ for the generalized minimal models between $\mathcal{M}(4, 3)$ and $\mathcal{M}(5, 4)$. This reveals a problem with our strategy for proving that the allowed region in Figure \ref{single-multi} must be large enough to include (\ref{upper-bound}). Constructing the generalized minimal model solution to crossing symmetry only accomplishes this in the one-correlator case. We must therefore conclude that there is at least one other way to extend the unitary subsector $\left < \sigma\sigma\sigma\sigma \right >$ into a consistent three-correlator system. This solution to crossing should have positive squared OPE coefficients wherever it exists, not just in $\frac{1}{5} \leq \Delta_\sigma \leq \frac{1}{2}$. As the solution might be very different from the theories discussed above, it is worth using the numerical bootstrap to see what else can be learned about it.

\section{Lessons for the bootstrap}
We saw in the last section that above central charge $\frac{7}{10}$, the generalized minimal models exhibit the restricted notion of unitarity that allows them to appear in Figure \ref{single-multi}. On the other hand, for $\frac{1}{2} < c < \frac{7}{10}$, they are highly non-unitary at the level of three correlators; the global coefficient $c^{\sigma\epsilon(1, 4)\sigma\epsilon}_3$ and the Virasoro coefficient $C^{(1, 5)}_{(1, 3)(1, 3)}$ both become negative in this region. Working around this problem, the machinery of the bootstrap has filled in this region with another solution whose $\sigma \times \sigma$ OPE agrees with that of a generalized minimal model. In this section, we will give an intuitive argument for why this should be possible. Beyond this, we will discuss two issues related to the replacement solution.

This first is whether it can be found uniquely. A technique called the extremal functional method is often used to extract a unique solution to crossing symmetry and unitarity whenever a dimension gap or OPE coefficient is extremized \cite{ep12, epprsv14, ep16, s17a, cl17, cflw17, cll17}. Based on this, we might expect to find a single line of exotic solutions that smoothly joins the $\mathcal{M}(m + 1, m)$ line at $\Delta_\sigma = \frac{1}{5}$. We will actually find the opposite --- a boundary of Figure \ref{single-multi} that has many possible choices for the local CFT data outside $\sigma \times \sigma$. To reconcile this with the standard lore about extremality, one has to remember that the bootstrap equations take on a more intricate form when there are multiple correlators.

The second is the prospect of excluding the above solution with further numerics. One reason for doing this with global conformal blocks is simply the technical challenge posed by Virasoro conformal blocks. There has indeed been recent progress in using the full Virasoro symmetry to carve out $c > 1$ CFTs \cite{ckly17}. However, tractable four-point functions with extended supersymmetry appear to be limited to those of BPS operators \cite{lsswy15, lswy16}. Global blocks were therefore a necessary ingredient of \cite{cls17}, a program which aims to constrain the space of superconformal theories using external operators in long multiplets. There has also been recent interest in conformal theories that have no locality and therefore no Virasoro algebra \cite{prvz15, brrz17a, brrz17b}. These provide a different motivation for shrinking the regions in Figure \ref{single-multi}.

\subsection{Reduction to one correlator}
The well known bootstrap constraints for three correlators with $\mathbb{Z}_2$ symmetry take the form of five crossing equations. As reviewed in Appendix B, the vector of equations has one component for $\left < \sigma\sigma\sigma\sigma \right >$, one component for $\left < \epsilon\epsilon\epsilon\epsilon \right >$ and three components for $\left < \sigma\sigma\epsilon\epsilon \right >$. Given a generic solution to crossing, it is easy to see that four (three) sum rules will break when an even (odd) operator is removed from the theory. However, there is a pleasing non-generic property that holds for generalized minimal models; $\epsilon \times \epsilon$ contains more operators than $\sigma \times \sigma$. Because of this, only one crossing equation is disturbed when we remove the $\phi_{1, 5}$ conformal family. This is the source of almost all unitarity violation in the system built from $\phi_{1, 2}$ and $\phi_{1, 3}$. Since negativity of $c^{\sigma\epsilon(1, 4)\sigma\epsilon}_3$ only affects spinning operators with $\Delta > \frac{45}{8}$, it is possible that the numerics are largely insensitive to it \cite{hrv16}. Assuming that problems with the mixed correlator are negligible, we will focus on
\begin{eqnarray}
&& \sum_{\mathcal{O}} \lambda^2_{\epsilon\epsilon\mathcal{O}} F^{\epsilon\epsilon ; \epsilon\epsilon}_{-, \mathcal{O}}(u, v) = 0 \nonumber \\
&& F^{\epsilon\epsilon ; \epsilon\epsilon}_{-, \Delta, \ell}(u, v) \equiv v^{\Delta_\epsilon} g^{0, 0}_{\Delta, \ell}(u, v) - u^{\Delta_\epsilon} g^{0, 0}_{\Delta, \ell}(v, u) \label{epsilon-focus}
\end{eqnarray}
as the single condition that needs to be restored. Once a solution to (\ref{epsilon-focus}) is found, one can incorporate it into the three-correlator problem by choosing $\lambda_{\sigma\sigma\mathcal{O}} = 0$ for new operators.\footnote{We are using four crossing equations to derive (\ref{upper-bound}) analytically and then claiming that the fifth crossing equation can be satisfied for free. This is different from what happens in three dimensions. We have checked that the island in \cite{kps14} merges with the rest of the allowed region once the fifth crossing equation is dropped.}

\begin{figure}[h]
\centering
\includegraphics[scale=0.6]{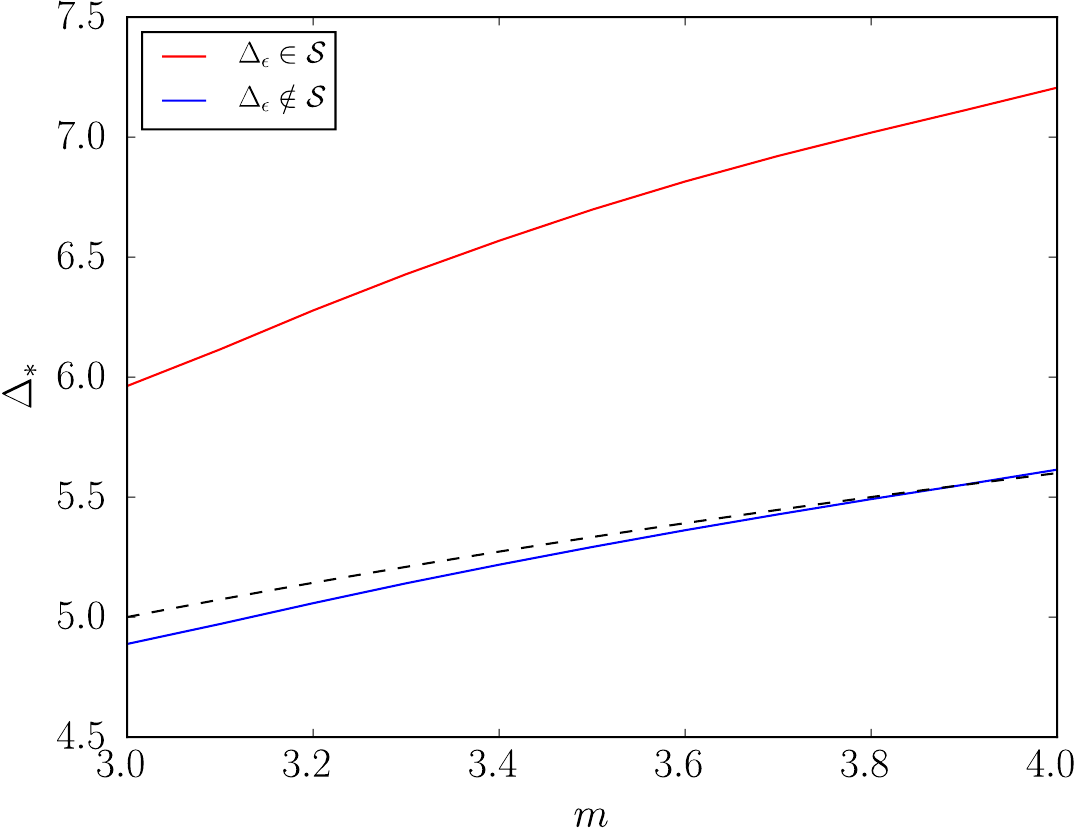}
\caption{Dimension bounds for irrelevant operators in (\ref{replace-15}). The dotted line shows the dimension of the primary scalar $\phi_{1, 5}$ whose multiplet needs to be replaced for $3 < m < 4$.}
\label{reduction}
\end{figure}

Checking the solvability of (\ref{epsilon-focus}) for real $\lambda_{\epsilon\epsilon\mathcal{O}}$ is the simplest numerical bootstrap problem. Emphasizing the contributions of operators that are already present, we may write
\begin{eqnarray}
&& \sum_{\mathcal{O}} \lambda^2_{\epsilon\epsilon\mathcal{O}} F^{\epsilon\epsilon ; \epsilon\epsilon}_{-, \mathcal{O}}(u, v) = - F_{1, 1}(u, v) - C^{(1, 3)}_{(1, 3)(1, 3)} F_{1, 3}(u, v) \nonumber \\
&& F_{1, 1}(u, v) \equiv \sum_{n, \bar{n}} c^{\epsilon\epsilon(1, 1)\epsilon\epsilon}_n c^{\epsilon\epsilon(1, 1)\epsilon\epsilon}_{\bar{n}} F^{\epsilon\epsilon ; \epsilon\epsilon}_{-, n + \bar{n}, |n - \bar{n}|}(u, v) \nonumber \\
&& F_{1, 3}(u, v) \equiv \sum_{n, \bar{n}} c^{\epsilon\epsilon(1, 3)\epsilon\epsilon}_n c^{\epsilon\epsilon(1, 3)\epsilon\epsilon}_{\bar{n}} F^{\epsilon\epsilon ; \epsilon\epsilon}_{-, \Delta_\epsilon + n + \bar{n}, |n - \bar{n}|}(u, v) \; . \label{replace-15}
\end{eqnarray}
Keeping operators with $\Delta \leq 30$, we have used the results of the last section to approximate the right-hand side of (\ref{replace-15}). In the sum over operators, there will be a continuum of irrelevant scalars not in $\phi_{1, 1} \cup \phi_{1, 3}$ which begins at some gap $\Delta_* > 2$. If this is the only set of scalars on the left-hand side of (\ref{replace-15}), we are dealing with the set $\mathcal{S} = \{ \Delta > \Delta_* \}$. Alternatively, we could allow the dimension $\Delta_\epsilon$ to appear again and enlarge it to $\mathcal{S} = \{ \Delta > \Delta_* \} \cup \{ \Delta = \Delta_\epsilon \}$. The second choice is the one applicable to Figure \ref{single-multi}. However, in Figure \ref{reduction}, we consider the first choice as well. This is because it is possible to rederive Figure \ref{single-multi} under the requirement that $\epsilon$ is non-degenerate \cite{kpsv16}. As numerical accuracy is improved, we expect the blue line to precisely meet the dotted line at $m = 4$.\footnote{Varying the spatial dimension provides one indication that numerical errors have a large effect. Evaluating conformal blocks at $d = 2.01$ instead of $d = 2$ results in a much smaller bound on $\Delta_*$.}

In Figure \ref{reduction}, the red curve tells us that (\ref{upper-bound}) is admissible whenever we treat $(\Delta_\sigma, \Delta_\epsilon)$ as allowed dimensions and perform a two-parameter scan. The blue curve tells us that (\ref{upper-bound}) will still be admissible when we fix $\epsilon$ as a single operator at angle $\theta = \mathrm{arctan} \left ( \frac{\lambda_{\sigma\sigma\epsilon}}{\lambda_{\epsilon\epsilon\epsilon}} \right )$ and scan over $(\Delta_\sigma, \Delta_\epsilon, \theta)$. Finally, one may contemplate the effect of imposing $\theta = \frac{\pi}{2}$ which is one consequence of Kramers-Wannier duality in the Ising model. Although this question cannot be answered with Figure \ref{reduction}, we have found that (\ref{upper-bound}) persists yet again. In all three cases, (\ref{upper-bound}) is not just an allowed line --- it is the \textit{maximal} allowed line. From this, we must conclude that many solutions to crossing, labelled by values of $\Delta_*$ between the red and blue lines, lie along the bound on the left side of Figure \ref{single-multi}. This signals the presence of a flat direction, \textit{e.g.} a bound in $(\Delta_\sigma, \Delta_\epsilon, \theta)$ space which is independent of $\theta$. A flat direction in the modular bootstrap was previously seen in \cite{dfx17}. The main argument for unique extremal functionals comes from \cite{ep16}, in which the multi-correlator bootstrap equations were augmented with angles for each operator in OPE space. To extract a spectrum in this formulation, one would have to look for zeros of these functionals on the entire $(\Delta, \theta)$ plane. We suspect that the flat direction here corresponds to a zero being achieved on a codimension-one locus.

\begin{figure}[h]
\centering
\includegraphics[scale=0.6]{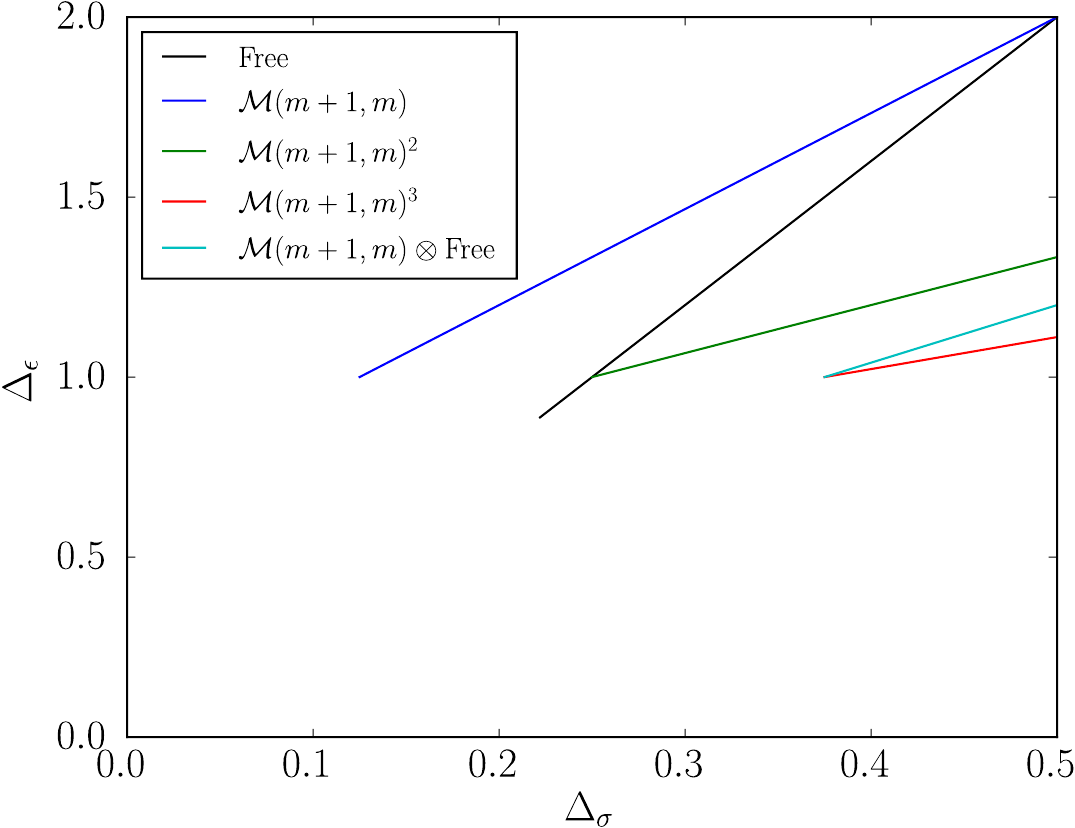}
\caption{Tensor product theories that are allowed in both sides of Figure \ref{single-multi}. Since we must have $\Delta_\epsilon \geq 1$, all other tensor products involving the free field vertex operators necessarily lie to the right.}
\label{tensor}
\end{figure}

Because the generalized minimal model line (\ref{upper-bound}) is allowed by the bootstrap, there are several lines in the interior of Figure \ref{single-multi} that must be allowed as well. These can be constructed through one or more tensor products. If we multiply two generalized minimal models for instance, the only non-trivial operator whose dimension lies to the left of $\Delta_\sigma = \frac{1}{2}$ is $\sigma \otimes \sigma$. Writing its OPE schematically,
\begin{equation}
(\sigma \otimes \sigma) \times (\sigma \otimes \sigma) = (I \otimes I) + (I \otimes \epsilon) + (\epsilon \otimes I) + (\epsilon \otimes \epsilon) + \dots \label{tensor-ope}
\end{equation}
includes two relevant operators. In order for these to have the same scaling dimension, the $\mathcal{M}(m + 1, m) \otimes \mathcal{M}(m^\prime + 1, m^\prime)$ product must have $m = m^{\prime}$. Expanding the search to include free theories and generalized free theories, it is a simple exercise to check that the lines in Figure \ref{tensor} all have one relevant $\mathbb{Z}_2$-odd scaling dimension and one relevant $\mathbb{Z}_2$-even scaling dimension.

Evidently, it is not possible to isolate particular minimal models in a three-correlator bootstrap by specifying the number of relevant operators. One has to consider more stringent assumptions or add more correlators. The former approach was discussed already in \cite{r11}, where it was found that the one-correlator Ising kink sharpens considerably when scalars are restricted to lie in $\mathcal{S} = \{ \Delta > 3 \} \cup \{ \Delta = \Delta_\epsilon \}$. This kink becomes an island when a similar restriction is made for a three-correlator system. Following \cite{ls17}, it is likely that one can obtain this island from a single correlator by imposing large gaps in the spin-0 and spin-2 sectors.

\begin{table}[h]
\centering
\begin{tabular}{l|l}
Coefficients & Signs \\
\hline
$C^{(1, 3)}_{(1, 4)(1, 4)}$ & Positive \\
$C^{(1, 5)}_{(1, 4)(1, 4)}$ & Negative for $m < 4$ \\
$C^{(1, 7)}_{(1, 4)(1, 4)}$ & Negative for $m < 6$ \\
$C^{(1, 6)}_{(1, 3)(1, 4)}$ & Negative for $m < 5$ \\
$C^{(1, 5)}_{(1, 2)(1, 4)}$ & Negative for $m < 4$
\end{tabular}
\caption{Virasoro block coefficients (other than the ones in Figure \ref{ope-plot}) appearing in four-point functions made from $\sigma$, $\epsilon$ and $\sigma^\prime$. Only one is non-negative for all $m \geq 3$.}
\label{new-coeffs}
\end{table}

A more ambitious goal is to produce islands under minimal assumptions by introducing a third external scalar. Taking this external scalar to be odd, plots along the lines of Figure \ref{ope-plot} offer some preliminary insight.\footnote{The lightest $\mathbb{Z}_2$-odd scalar after $\sigma$ is $\sigma^\prime \equiv \phi_{1, 4}$. The lightest $\mathbb{Z}_2$-even scalar after $\epsilon$ is $\epsilon^\prime \equiv [L_{-2}\bar{L}_{-2}, \phi_{1, 1}]$. We have made the choice in which all operator fusions are between Virasoro primaries.} Regions where all $C_{(r_1,s_1)(r_2,s_2)}^{(r_3,s_3)} > 0$ are likely to survive, but as Table \ref{new-coeffs} shows, there can be several negative structure constants with more than three correlators. We have seen that the negative constant $C^{(1, 5)}_{(1, 3)(1, 3)}$ in the three-correlator system was innocuous because it did not appear in any mixed correlators. It is therefore encouraging that the coefficients $C^{(1, 5)}_{(1, 2)(1, 4)}$ and $C^{(1, 6)}_{(1, 3)(1, 4)}$, which have first-order zeros, participate in $\left < \sigma\sigma\sigma^\prime\sigma^\prime \right >$ and $\left < \epsilon\epsilon\sigma^\prime\sigma^\prime \right >$ respectively. Lest we become too encouraged, it is important to note that $\Delta_{\sigma^\prime}$ is defined as the starting point for a continuum of irrelevant operators. This represents a fundamental difference as compared to the one-correlator and three-correlator analysis. It remains to be seen whether we can still derive strong bounds from a scan over two isolated scaling dimensions and one non-isolated scaling dimension.

\subsection{Conformal manifolds}
Even though the tensor product lines above all have $\Delta_\sigma > \frac{1}{4}$, there could be other CFTs in $\frac{1}{8} < \Delta_\sigma < \frac{1}{5}$ with sufficiently few relevant operators to survive the constraints of $\sigma$, $\epsilon$ and $\sigma^\prime$. An interesting possibility, that would complicate the search for islands, is a continuous line of theories ending somewhere close to the Ising model. We may search for an example by using the extremal functional method on solutions to crossing that involve a scalar $\Phi$ of dimension 2.\footnote{The idea of building up a previously unknown conformal manifold was discussed in \cite{bbr17}. Their method uses large-$N$ perturbation theory to construct the holographic dual of a bulk action with shift-symmetric couplings. For a recent bootstrap study of a known conformal manifold, see \cite{bbclp17}.} This provides another opportunity to predict a mixed correlator result without actually bootstrapping more than one correlator.

The first step is to maximize $\lambda^2_{\sigma\sigma\Phi}$ in the range $\frac{1}{8} < \Delta_\sigma < \frac{1}{2}$ with the constraint that all scalar dimensions are above $2\Delta_\sigma$. Each point saturating the bound yields a spectrum with the marginal deformation $\Phi$. To see that this set of solutions is not a conformal manifold, we may check that the central charge varies with $\Delta_\sigma$. Specifically, it reaches a minimum value of $c \approx 1.12$. It is then straightforward to force our putative theories to have this central charge (or any larger value) by taking
\begin{equation}
F^{\sigma\sigma ; \sigma\sigma}_{-,0,0}(u, v) \mapsto F^{\sigma\sigma ; \sigma\sigma}_{-,0,0}(u, v) + \frac{\Delta_\sigma^2}{c} F^{\sigma\sigma ; \sigma\sigma}_{-,2,2}(u, v) \label{forcing-c}
\end{equation}
in the usual one-correlator crossing equation. To prevent (\ref{forcing-c}) from being undone by a second copy of the stress tensor, we have imposed the gap $\Delta_T > 2.1$ on spin-2 operators. Performing a second pass in this way, we have found low-lying dimensions and OPE coefficients with examples plotted in Figure \ref{manifold}. With constant $c$ and $\Delta_\Phi$, these CFT data appear to satisfy all constraints of a conformal manifold that are directly accessible to $\left < \sigma\sigma\sigma\sigma \right >$. Despite this, there is a method from \cite{ep12} that can be revived to gain indirect information about the next four-point function, which we will call $\left < \epsilon\epsilon\epsilon\epsilon \right >$.

\begin{figure}[t!]
\centering
\subfloat[][Scaling dimensions]{\includegraphics[scale=0.45]{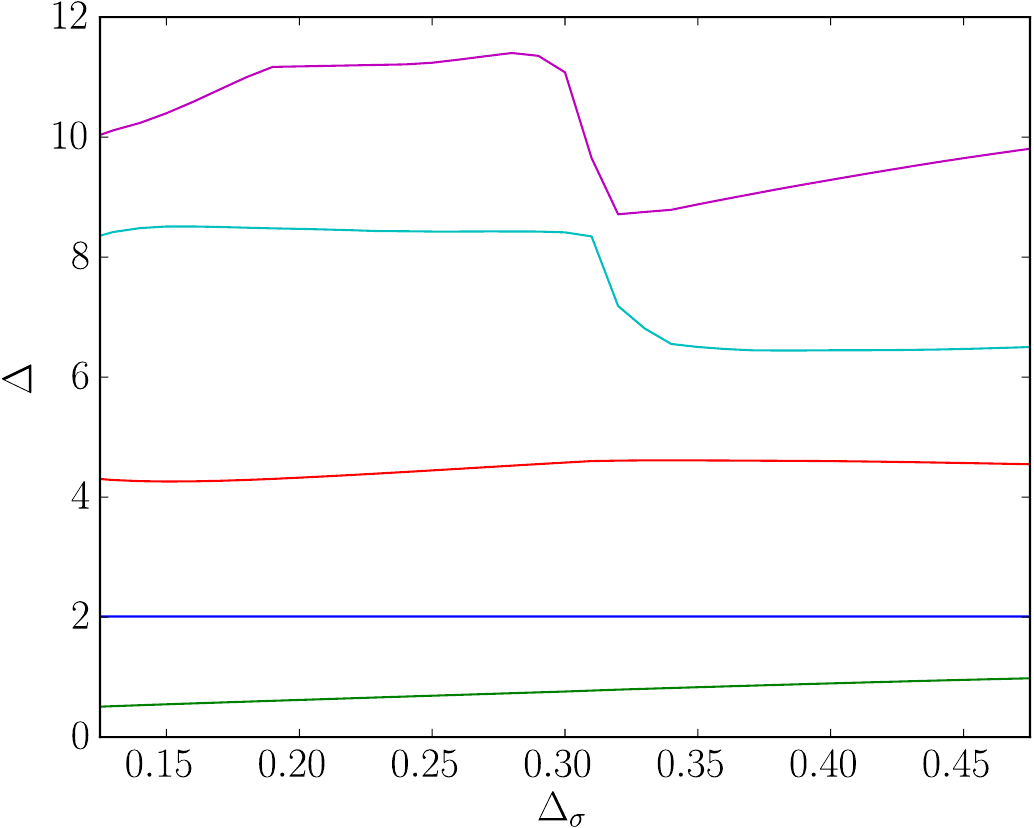}}
\subfloat[][OPE coefficients]{\includegraphics[scale=0.45]{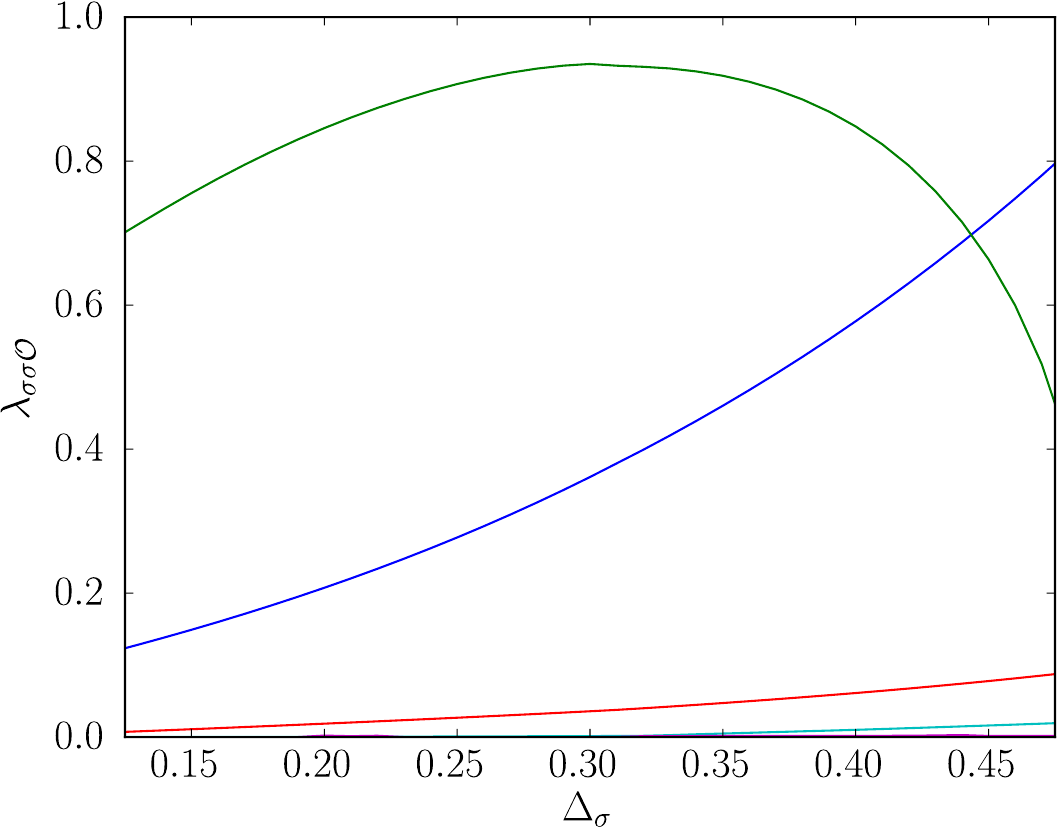}}
\caption{Dimensions and OPE coefficients for scalars in the spectrum maximizing $\lambda^2_{\sigma\sigma\Phi}$. The green and blue lines are $\epsilon$ and $\Phi$ respectively.}
\label{manifold}
\end{figure}

The key is that the dimensions in $\sigma \times \sigma$ are also the dimensions in any OPE between identical scalars. Given a sufficiently long list of scaling dimensions, there is no reason why a search for optimal OPE coefficients has to be done for $\sigma \times \sigma$ rather than $\epsilon \times \epsilon$. We confirm this in Appendix B, by taking an approximate Ising model spectrum and performing the same type of fit that was done in \cite{ep12}. Our results show that the estimation of $\lambda^2_{\sigma\sigma\mathcal{O}}$ could have been extended to $\lambda^2_{\epsilon\epsilon\mathcal{O}}$ without the three-correlator crossing equations. The fitting procedure therefore differs from the primal method of \cite{epprsv14}, which is more accurate but limited to the direct correlators under study. Assuming that the method can be trusted, at least for 2D theories, we now have access to $\epsilon \times \epsilon$ coefficients including $\lambda_{\epsilon\epsilon\Phi}$. This allows another property of conformal manifolds to be tested.

\begin{figure}[t!]
\centering
\includegraphics[scale=0.6]{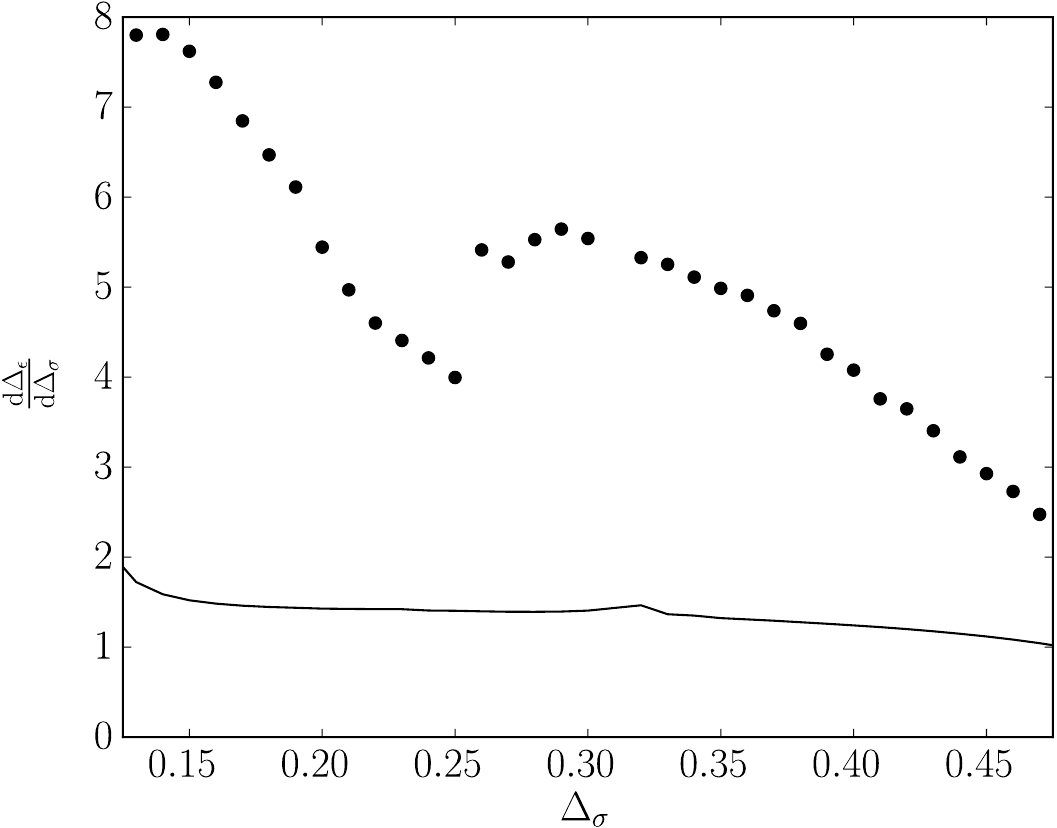}
\caption{The slope of $\Delta_\epsilon$ as one goes through the artificial spectra parameterized by $\Delta_\sigma$. The points, predicted by (\ref{flow-equation}), would have to match if these operators were to come from a genuine theory with a dimensionless coupling.}
\label{non-manifold}
\end{figure}

When CFTs are parameterized by an exactly marginal operator with coupling $g$, there is a complicated set of differential equations that their local data must satisfy \cite{b17,s17b}. The simplest of these is $\frac{\textup{d} \Delta_i}{\textup{d} g} = -S_{d - 1} \lambda_{ii\Phi}$ which reads $\frac{\textup{d} \Delta_i}{\textup{d} g} = -2\pi \lambda_{ii\Phi}$ in two dimensions. The constraint
\begin{equation}
\frac{\textup{d}\Delta_\epsilon}{\textup{d}\Delta_\sigma} = \frac{\lambda_{\epsilon\epsilon\Phi}}{\lambda_{\sigma\sigma\Phi}} \label{flow-equation}
\end{equation}
immediately follows. The two sides of this equation have been plotted in Figure \ref{non-manifold}. Even though $\lambda_{\epsilon\epsilon\Phi}$ is not known with high precision, the noise in the data seems much too small to explain the violation of (\ref{flow-equation}). We must conclude that if a conformal manifold in $\frac{1}{8} < \Delta_\sigma < \frac{1}{2}$ allowed by Figure \ref{single-multi} exists, it is not privileged enough to be found by this one-correlator exercise.

It would be very interesting to find other applications for the dual method in \cite{ep12}. Unfortunately, Appendix B shows that results become much less reliable above $d = 2$. A possible reason for this was given in \cite{s17a} which noticed a surprising preference for double-twist operators in the 3D numerical bootstrap. To review, double-twist families have the following schematic form for $\ell \rightarrow \infty$:
\begin{eqnarray}
[\phi \phi]_n &=& \phi \partial^{\mu_1} \dots \partial^{\mu_\ell} \Box^n \phi \nonumber \\
\tau &=& 2\tau_\phi + 2n + O(\ell^{-1}) \; . \label{double-twist}
\end{eqnarray}
The lightcone bootstrap requires them to appear in any CFT with a twist gap \cite{fkps12, kz12}. When $d > 2$, every operator except the identity has positive twist and the extremal functional method is able to find spectra dominated by (\ref{double-twist}). The absence of other operators interferes with our ability to fit OPE coefficients. Convsersely, 2D theories have no need for double-twist operators as the identity, stress tensor and higher-spin currents all have $\tau = 0$. The quality of our fit suggests that this resolves the main source of bias in the extremal spectrum.

\subsection{Supersymmetric minimal models}
So far, we have been concerned with finding the necessary correlators to single out theories that are minimal with respect to the Virasoro algebra. However, the unitary representations of the $\mathcal{N} = 1$ super-Virasoro algebra also admit a discrete series for central charges below that of the free field. In analogy with $\mathcal{M}(m + 1, m)$, we can continue them to $\mathcal{SM}(m + 2, m)$ where non-integer $m$ breaks unitarity but preserves crossing symmetry. Studying these solved theories offers another route toward understanding the systematics of the bootstrap. It could also be interesting to compare results for the tricritical Ising model since this is the lowest model of $\mathcal{SM}(m + 2, m)$ but the second lowest model of $\mathcal{M}(m + 1, m)$.

The super-Virasoro graded commutation relations
\begin{eqnarray}
\left [ L_m, L_n \right ] &=& (m - n) L_{m + n} + \frac{c}{12} m(m-1)(m+1) \delta_{m + n, 0} \nonumber \\
\left [ L_m, G_r \right ] &=& \left ( \frac{m}{2} - r \right ) G_{m + r} \nonumber \\
\{ G_r, G_s \} &=& 2 L_{r + s} + \frac{c}{3} \left ( r - \frac{1}{2} \right ) \left ( r + \frac{1}{2} \right ) \delta_{r + s, 0} \label{super-virasoro}
\end{eqnarray}
actually describe two algebras since fermions do not have to be periodic in radial quantization. When the indices on $G_r$ are integers, (\ref{super-virasoro}) is the Ramond superalgebra, otherwise it is the Neveu-Schwarz superalgebra. The super-Virasoro minimal models, which contain representations of each, have a Kac formula given by
\begin{eqnarray}
c &=& \frac{3}{2} - \frac{12}{m(m + 2)} \;\;\;\;\;\;\;\;\;\;\;\;\;\;\;\;\;\;\;\;\;\;\;\;\;\;\;\;\;\;\;\;\;\;\;\;\;\;\;\;\;\; m > 2 \nonumber \\
h_{r, s} &=& \frac{[(m + 2)r - ms]^2 - 4}{8m(m + 2)} + \frac{1}{32} [1 - (-1)^{r - s}] \;\;\; r, s \in \mathbb{Z}_{>0} \; . \label{super-kac-table}
\end{eqnarray}
If $r - s$ is even (odd), this is a Neveu-Schwarz (Ramond) degenerate weight \cite{fqs85}. In either case, it is degenerate at level $rs / 2$. It is clear by inspection that $G_{\pm \frac{1}{2}}$ generate a subalgebra of (\ref{super-virasoro}) that is independent of $c$. For integer indices, on the other hand, no global subalgebra exists. A numerical bootstrap approach is therefore most readily accessible for the Neveu-Schwarz sectors of $\mathcal{N} = (1, 1)$ theories.

In the global algebra, which is $\mathfrak{osp}(2 | 1)$, primary operators may be written as superfields; $\Phi(z, \theta) = \phi(z) + \theta \psi(z)$. The superspace distance, which enters in correlation functions, is $Z_{ij} \equiv z_i - z_j - \theta_i \theta_j$. Even though cross-ratios in $\mathbb{R}^d$ all involve at least four points, invariant combinations in superspace may be built using three points as well. The quantity
\begin{equation}
\eta = \frac{\theta_1 Z_{23} + \theta_2 Z_{31} + \theta_3 Z_{12} + \theta_1 \theta_2 \theta_3}{\sqrt{Z_{12}Z_{23}Z_{31}}} \label{super-3pt1}
\end{equation}
is invariant under $\mathfrak{osp}(2 | 1)$ \cite{q86}. As a result, the three-point function depends on more than just an OPE coefficient. The general expression for the chiral half is
\begin{equation}
\left < \Phi_1(z_1, \theta_1) \Phi_2(z_2, \theta_2) \Phi_3(z_3, \theta_3) \right > = \frac{\lambda_{123}(1 + \zeta \eta)}{Z_{12}^{h_1 + h_2 - h_3} Z_{23}^{h_2 + h_3 - h_1} Z_{31}^{h_3 + h_1 - h_2}} \label{super-3pt2}
\end{equation}
where $\zeta$ is an arbitrary Grassman number. Until recently, such extra parameters were eliminated by restricting the superconformal bootstrap to correlators of BPS operators \cite{ps10,psv12}. We may indeed impose shortening conditions on (\ref{super-3pt2}), but due to the small amount of supersymmetry, this would require us to give up a lot. To be annihilated by a supercharge, each external operator would have to be the identity in at least one of the $\mathfrak{osp}(2 | 1)$ factors. It is therefore preferable to leave (\ref{super-3pt2}) in its most general form and use superconformal blocks that include unknown coefficients reflecting the presence of $\zeta$.

The authors of \cite{cls17} computed some of the necessary blocks and introduced a framework that still allows the bootstrap to proceed. Their idea is to consider an entire multiplet at once, with the external correlators involving all combinations of a primary and its super-descendants. When this is carried out for $\mathcal{N} = (1, 1)$, the allowed regions will have to include all points corresponding to the $\mathcal{SM}(m + 2, m)$. The strongest statement we can make from this is that the line
\begin{equation}
\Delta_\epsilon = \frac{8}{3} \Delta_\sigma \label{super-upper-bound}
\end{equation}
must be inside the bound for $\Delta_\sigma < \frac{1}{8}$. This comes from choosing the Neveu-Schwarz fields $\sigma \equiv \phi_{2, 2}$ and $\epsilon \equiv \phi_{3, 3}$. Right at $\Delta_\sigma = \frac{1}{8}$, we find ourselves in the $c = 1$ model where $\phi_{3, 3} = \phi_{1, 3}$ and the entire level-$\frac{3}{2}$ subspace decouples. To see this, we may check that
\begin{equation}
\left | \chi \right > = \left [ G_{-\frac{3}{2}} - \frac{2}{2h + 1} L_{-1} G_{-\frac{1}{2}} \right ] \left | h \right > \label{super-null1}
\end{equation}
is the unique quasiprimary state. Computing the norm and setting $h \mapsto h_{3,3}$, we find
\begin{eqnarray}
\left < \chi | \chi \right > &=& \frac{2(2ch + c + 6h^2 - 9h)}{3(2h + 1)} \nonumber \\
&\mapsto& \frac{(m - 4)(m + 6)}{m^2 + 2m + 8} \label{super-null2}
\end{eqnarray}
with a first-order zero. This is the behaviour that we saw for $m = 3$ in the bosonic case, but now it occurs for $m = 4$. The tricritical Ising model, which has no $\phi_{3, 3}$ operator, lives at the point $(\Delta_\sigma, \Delta_\epsilon) = \left ( \frac{1}{10}, \frac{1}{10} \right )$.

The above calculation shows that a global block coefficient in the supersymmetric generalized minimal model line becomes negative for $m < 4$. If we are to see an associated kink, this line must saturate the bound on operator dimensions from the long multiplet bootstrap. This brings us to a crucial difference between $\mathcal{M}(m + 1, m)$ and the Neveu-Schwarz sector of $\mathcal{SM}(m + 2, m)$. In the former case, we saw the correct saturation with (\ref{upper-bound}). The same cannot hold for (\ref{super-upper-bound}) because it is strictly below the line for vertex operators. Writing a CFT vertex operator as $e^{iq\phi(z)}$ and an SCFT vertex operator as $e^{iq\Phi(z, \theta)}$, the two important properties are $\Delta \propto q^2$ and additivity of $q$. These lead to $\Delta_\epsilon = 4\Delta_\sigma$ which is allowed by the one-correlator region of Figure \ref{single-multi}. Due to the restriction on the number of relevant operators, the three-correlator region omits this line until $\Delta_\sigma = \frac{2}{9}$. It is therefore clear that treating four copies of the same multiplet with the methods of \cite{cls17} is not enough. If our goal is to see a minimal model kink, the $\mathcal{N} = (1, 1)$ bootstrap will require multiple correlators at the superspace level. Since each of these must separately expand to a mixed correlator system, the resulting problem is likely to be numerically intensive.

\section{Conclusion}
Through a combination of analytic techniques and one-correlator numerics, we have explained an important aspect of Figure \ref{single-multi} --- we have shown that the constraints on four-point functions of relevant operators are not strong enough to exclude the generalized minimal models. Constraints that are strong enough may be used in the conformal or superconformal bootstrap, but only when the number of four-point functions exceeds the maximum system size that has been tested to date.

Our analysis proceeded correlator-by-correlator. For $\left < \sigma\sigma\sigma\sigma \right >$, the results were exact. The expressions (\ref{4sigma-result-id}) and (\ref{4sigma-result-eps}) ensured that the global block coefficients involved were all positive. For $\left < \sigma\sigma\epsilon\epsilon \right >$, all but one coefficient appeared to be positive upon using the approach of \cite{lrv13} --- we simply expanded in $\mathfrak{sl}(2)$ blocks to high order and conjectured that the pattern continues to hold. For $\left < \epsilon\epsilon\epsilon\epsilon \right >$, two pieces of the $\mathcal{M}(m + 1, m)$ line had to be treated separately. In $4 < m < \infty$, it was enough to compute coefficients again and see that there were no obvious signs of unitarity violation. However, $3 < m < 4$ required an exotic solution having only partial overlap with the CFT data of a generalized minimal model. It was possible to see evidence of this in the one-correlator bootstrap because of a special property of $\phi_{1, 5}$ operators, namely their absence from the superselection sector $\sigma \times \sigma$. This had interesting implications for the uniqueness results in \cite{ep16}.

A possible future endeavour is to put all correlators on the same footing as $\left < \sigma\sigma\sigma\sigma \right >$. In order to do this, one does not necessarily have to solve for $c^{\epsilon\epsilon(1, s)\epsilon\epsilon}_{2n}$ as a known special function. The $c^{\sigma\sigma(1, s)\sigma\sigma}_{2n}$ were ultimately shown to be positive using only the recursion for Wilson polynomials. Instead, the main challenge in extending the positivity proof is expressing global block coefficients as solutions of one recurrence relation instead of two. Our current approach, based on the BPZ equation, is awkward in this respect. It uses the Frobenius method to evaluate Taylor coefficients and then feeds these into a second recursion to obtain global block coefficients. It is worth checking if hypergeometric identities can be used to derive a recursion that operates on global block coefficients directly. It would also be interesting to take a closer look at super-BPZ equations. Our discussion surrounding the vanishing norm (\ref{super-null2}) can be made more systematic if we also check how other $\mathfrak{osp}(2 | 1)$ blocks appear. BPZ differential equations in superspace have been studied in \cite{fqs85, q86} and some of them are second-order. This is exactly what we need to go beyond recurrence relations and apply the decompositions in \cite{hv17,h17}. Additionally, these methods could be applicable to the KZ equations associated with extended chiral algebras. A realistic hope is using them to explain a numerical bound in \cite{rs17} which interpolates between $\mathcal{W}_3$-minimal models.

The last possibility we have discussed is an extension of the extremal functional method --- reviving the fit in \cite{ep12} to estimate more OPE coefficients than the ones that are known to high precision. While we only used this to demonstrate a null result, it would be interesting to find a further use for it in two dimensions. Complications in higher dimensions arise due to the privileged role of double-twist operators in the numerical bootstrap; a result that is not fully understood \cite{s17a}. One should be able to get a sense of how robust it is by studying alternative bootstrap algorithms such as the one in \cite{ehs16}.

\section*{Acknowledgements}
This work was partially supported by the Natural Sciences and Engineering Research Council of Canada. Numerical results in this paper were obtained using the high-performance computing system at the Institute for Advanced Computational Science at Stony Brook University. I am grateful for discussions with Shai Chester, Liam Fitzpatrick, Anton de la Fuente, Matthijs Hogvervorst, Jaehoon Lee, Madalena Lemos, Dalimil Maz\'{a}\v{c}, Miguel Paulos, Leonardo Rastelli, Slava Rychkov, David Simmons-Duffin, Balt van Rees and Alessandro Vichi. Most of these took place at various meetings of the Simons Collaboration on the Nonperturbative Bootstrap.

\appendix
\section{Linear difference equations}
We have regularly encountered linear recursions with three terms such as (\ref{continuous-hahn-recursion}) and (\ref{wilson-recursion}). Asymptotic analysis of sequences obeying these relations is a well understood subject, going by the name Birkhoff-Trjitzinsky theory. The following theorem summarizes a number of results from it \cite{w04}.
\begin{theorem}
Let $y_1(n)$ and $y_2(n)$ be the two linearly independent solutions of the difference equation
\begin{equation}
y(n + 2) + a(n)y(n + 1) + b(n)y(n) = 0 \label{diffeq1}
\end{equation}
where the coefficients have asymptotic expansions $a(n) \sim \sum_{s = 0}^\infty \frac{a_s}{n^s}$ and $b(n) \sim \sum_{s = 0}^\infty \frac{b_s}{n^s}$.
\begin{enumerate}
\item
If the \textit{characteristic equation} $\rho^2 + a_0 \rho + b_0 = 0$ has two distinct roots $\rho_1$ and $\rho_2$, the solutions satisfy $y_j(n) \sim \rho_j^n n^{\alpha_j} \sum_{s = 0}^\infty \frac{c_{s, j}}{n^s}$ where $\alpha_j = \frac{a_1 \rho_j + b_1}{a_0 \rho_j + 2b_0}$.
\item
Otherwise, consider the double root $\rho$. If the \textit{auxiliary equation} $a_1 \rho + b_1 = 0$ is not satisfied, the solutions satisfy $y_j(n) \sim \rho^n e^{(-1)^j \beta \sqrt{n}} n^\alpha \sum_{s = 0}^\infty (-1)^{js} \frac{c_s}{n^{s / 2}}$ where $\alpha = \frac{1}{4} + \frac{b_1}{2b_0}$ and $\beta = 2\sqrt{\frac{a_0a_1 - 2b_1}{2b_0}}$.
\item
Otherwise, consider the roots $\alpha_1$ and $\alpha_2$ of the \textit{indicial equation} $\alpha(\alpha - 1)\rho^2 + (a_1 \alpha + a_2)\rho + b_2 = 0$, ordered according to $\Re \alpha_2 \geq \Re \alpha_1$. If $\alpha_2 - \alpha_1 \notin \mathbb{Z}_{\geq 0}$, the solutions satisfy $y_j(n) \sim \rho^n n^{\alpha_j} \sum_{s = 0}^\infty \frac{c_{s, j}}{n^s}$.
\item
Otherwise, let $m = \alpha_2 - \alpha_1$. The asymptotic expansion for the first solution is unchanged from the previous case but for the second solution we must use $y_2(n) \sim \rho^n n^{\alpha_2} \left ( \sum_{s = 0}^\infty \frac{d_s}{n^s} - \frac{d_m}{n^m} \right ) + c\log(n)y_1(n)$.
\end{enumerate}
\end{theorem}
Taking (\ref{diffeq1}) to be the recurrence relation for Wilson polynomials, we find the roots $\rho = 1$, $\alpha_1 = -2(a + x)$ and $\alpha_2 = -2(a - x)$ in the third case of the theorem. Each coefficient in (\ref{wilson-poly}) can then be written as a linear combination of two functions asymptotic to power laws. As seen in Table \ref{rates}, all of them decay to zero. When comparing to the results of \cite{prer12, ry15}, one must remember that $c^{\sigma\sigma(1, 1)\sigma\sigma}_{2n}$ and $c^{\sigma\sigma(1, 3)\sigma\sigma}_{2n}$ include many squared OPE coefficients due to the increasing amount of degeneracy at each level of a Verma module.
\begin{table}[h]
\centering
\begin{tabular}{l|l}
Coefficient & Leading rates \\
\hline
$c^{\sigma\sigma(1, 1)\sigma\sigma}_{2n}$ & $n^{-\frac{8 + 2\Delta_\sigma}{3}}$ and $n^{2\Delta_\sigma - 2}$ \\
$c^{\sigma\sigma(1, 3)\sigma\sigma}_{2n}$ & $n^{-\frac{1 - 2\Delta_\sigma}{3}}$ and $n^{-2\Delta_\sigma - 1}$
\end{tabular}
\caption{Decay rates of the fundamental solutions that comprise two of our main results. We have not included the prefactors in (\ref{wilson-poly}). These will make the convergence much faster, namely $(1 / 16)^n$, which can be predicted from the growth rate of $K_{2n}(1)$.}
\label{rates}
\end{table}

Our main claim about $c^{\sigma\sigma(1, 1)\sigma\sigma}_{2n}$ and $c^{\sigma\sigma(1, 3)\sigma\sigma}_{2n}$ --- that they are positive for finite $n$ --- cannot be proven with asymptotics. Instead, we will use a theorem from \cite{xy11} which bounds the ratio between neighbouring terms in a sequence.
\begin{theorem}
Let $x(n)$ be a solution of
\begin{equation}
x(n) \geq \frac{a(n)}{b(n)} x(n - 1) - \frac{c(n)}{d(n)} x(n - 2) \label{diffeq2}
\end{equation}
where $a(n)$, $b(n)$, $c(n)$ and $d(n)$ are degree-$k$ polynomials with positive leading terms. Also define $f(n) = a(n + 1) d(n + 1) - \frac{2b_k}{a_k} b(n + 1) c(n + 1) - \frac{a_k}{2b_k} b(n + 1) d(n + 1)$. Finally, let $m$ be an integer large enough to guarantee that $b(n)$, $c(n)$, $d(n)$ and $f(n)$ have positive values for $n \geq m$. If $\frac{x(m)}{x(m - 1)} > \frac{a_k}{2b_k}$ then $\frac{x(n)}{x(n - 1)} > \frac{a_k}{2b_k}$ for $n \geq m$.
\end{theorem}
\begin{proposition}
The sequences $y_1(n) = P_n \left ( \frac{7 - 2\Delta_\sigma}{6}, \frac{4 - 2\Delta_\sigma}{6}, -\frac{1 - 2\Delta_\sigma}{6}, \frac{5 + 2\Delta_\sigma}{6} ; \frac{1 + 4\Delta_\sigma}{6} \right )$ and $y_2(n) = P_n \left ( \frac{1 + \Delta_\sigma}{3}, \frac{5 + 2\Delta_\sigma}{6}, -\frac{1 - 2\Delta_\sigma}{6}, -\frac{1 - 2\Delta_\sigma}{6} ; \frac{1 + 4\Delta_\sigma}{6} \right )$ of Wilson polynomials are positive for $\frac{1}{8} \leq \Delta_\sigma \leq \frac{1}{2}$.
\end{proposition}
\begin{proof}
The theorem above is more effective at identifying increasing sequences than ruling out changes of sign directly. Therefore, we will work with $x_1(n) = n^2 y_1(n)$ and $x_2(n) = n y_2(n)$ --- expressions that have been guided by the asymptotics in Table \ref{rates}.

Looking at the more difficult case first, $x_1(n)$ satisfies (\ref{diffeq2}) with
\begin{eqnarray}
b(n) &=& (n - 1)^2(n + 1)(2n + 1)(4n - 3)(6n - 4\Delta_\sigma + 5) \nonumber \\
c(n) &=& n(n - 1)(2n - 1)(2n - 3)(4n + 1)(3n + 2\Delta_\sigma - 4) \nonumber \\
d(n) &=& (n - 2)^2(n + 1)(2n + 1)(4n - 3)(6n - 4\Delta_\sigma + 5) \; . \label{bcd-seq1}
\end{eqnarray}
These are clearly positive for $n > 2$. Also, by writing out the polynomial for $a(n)$, we find a leading coefficient of $a_6 = 2b_6 = 96$. It remains to check $f(n)$ or equivalently $(n + 2)^{-1}(2n + 3)^{-1}(4n + 1)^{-1}(6n - 4\Delta_\sigma + 11)^{-1} f(n)$. This is a sixth degree polynomial in which $64 \Delta_\sigma (\Delta_\sigma + 1) n^6$ is followed immediately by negative coefficients. From this we see that the critical value of $n$, beyond which $f(n) > 0$, increases without bound as $\Delta_\sigma \rightarrow 0$. This reflects the fact that the $n^2$ we introduced is only able to overpower $n^{2\Delta_\sigma - 2}$ for strictly positive $\Delta_\sigma$. Fortunately, the smallest value of $\Delta_\sigma$ that we consider is $\frac{1}{8}$ leading to a critical value of $n = 44$. Since $x_1(44) > x_1(43)$, we establish positivity of the entire sequence $x_1(n)$ by checking its first 44 terms.

Things will be easier for $x_2(n)$ which satisfies (\ref{diffeq2}) for
\begin{eqnarray}
b(n) &=& (n - 1)(6n + 4\Delta_\sigma - 5)^2(6n + 4\Delta_\sigma + 1)(6n + 8\Delta_\sigma - 7)(12n + 8\Delta_\sigma - 19) \nonumber \\
c(n) &=& 48 n(n - 1)(3n + 2\Delta_\sigma - 7)(3n + 2\Delta_\sigma - 4)^2(12n + 8\Delta_\sigma - 7) \label{bcd-seq2} \\
d(n) &=& (n - 2)(6n + 4\Delta_\sigma - 5)^2(6n + 4\Delta_\sigma + 1)(6n + 8\Delta_\sigma - 7)(12n + 8\Delta_\sigma - 19) \; . \nonumber
\end{eqnarray}
Additionally, we find $a_6 = 2b_6 = 31104$ and an $f(n)$ proportional to $(6n + 4\Delta_\sigma + 1)^2(6n + 4\Delta_\sigma + 7)(6n + 8\Delta_\sigma - 1)(12n + 8\Delta_\sigma - 7)$. The non-trivial factor of $f(n)$ begins with $20736 \Delta_\sigma (\Delta_\sigma + 1) n^5$ which is strictly positive in the considered range. From looking at the next coefficients, we find that $f(n)$ is positive for $n > 3$ just like the polynomials above. Checking that $x_2(3) > x_2(2) > x_2(1) > 0$, positivity of $x_2(n)$ has been proven as well.
\end{proof}

Although we do not have closed-form solutions for them, it is possible that $c^{\epsilon\epsilon(1, 1)\epsilon\epsilon}_{2n}$, $c^{\epsilon\epsilon(1, 3)\epsilon\epsilon}_{2n}$ and $c^{\epsilon\epsilon(1, 5)\epsilon\epsilon}_{2n}$ are positive as well. The main hint of this, which we now prove, is that the Virasoro blocks containing them have positive Taylor coefficients around $z = 0$. The fact that this is a necessary condition follows trivially from expanding $g(z) = \sum_{n = 0}^\infty c_{2n} K_{r + 2n}(z)$. The analogous statement in higher dimensions was proven in \cite{hjk16}.
\begin{proposition}
Let $b_k$ be a sequence starting at $b_0 = 1$ with the rest of the terms given by (\ref{4epsilon-recursion}). If $\Delta_\epsilon > 1$, the sequence monotonically increases.
\end{proposition}
\begin{proof}
Defining $K = k + r$ for brevity, we have
\begin{eqnarray}
&& [4K^3 + 8(1 - \Delta_\epsilon)K^2 + (3\Delta_\epsilon^2 - 14\Delta_\epsilon + 4)K + 3\Delta_\epsilon(\Delta_\epsilon - 2)]b_{k + 1} = \label{rel1} \\
&& [12K^2 - 4(5\Delta_\epsilon + 1)K + 6\Delta_\epsilon^2]Kb_k \nonumber \\
&& - [12K^3 - 16(\Delta_\epsilon + 2)K^2 + 2(\Delta_\epsilon^2 + 12\Delta_\epsilon + 14)K + 2(\Delta_\epsilon - 2)(\Delta_\epsilon + 1)(\Delta_\epsilon + 2)]b_{k - 1} \nonumber \\
&& + [4K^3 - 4(\Delta_\epsilon + 5)K^2 - (\Delta_\epsilon^2 - 10\Delta_\epsilon - 32)K + (\Delta_\epsilon - 2)(\Delta_\epsilon + 2)(\Delta_\epsilon + 4)]b_{k - 2} \; . \nonumber
\end{eqnarray}
Although this has four terms, we will convert it to a simpler recursion having only three. We do this by assuming that the piece with $b_{k - 1}$ and $b_{k - 2}$ is bounded by some function of $b_k$ and $b_{k - 1}$. Our ansatz for this function is $[4K^2 - 12(\Delta_\epsilon + 1)K + M]Kb_k + 4K^3b_{k - 1}$. In other words, we need to show that
\begin{eqnarray}
&& [4K^2 - 12(\Delta_\epsilon + 1)K + M]Kb_k > \label{assumption} \\
&& [8K^3 - 16(\Delta_\epsilon + 2)K^2 + 2(\Delta_\epsilon^2 + 12\Delta_\epsilon + 14)K + 2(\Delta_\epsilon - 2)(\Delta_\epsilon + 1)(\Delta_\epsilon + 2)]b_{k - 1} \nonumber \\
&& - [4K^3 - 4(\Delta_\epsilon + 5)K^2 - (\Delta_\epsilon^2 - 10\Delta_\epsilon - 32)K + (\Delta_\epsilon - 2)(\Delta_\epsilon + 2)(\Delta_\epsilon + 4)]b_{k - 2} \; . \nonumber
\end{eqnarray}
Using (\ref{assumption}) in (\ref{rel1}), we find an expression of the form $R_+(K) b_{k + 1} > R_0(K) b_k + R_-(K) b_{k - 1}$. We will perform a rescaling to instead write this as
\begin{eqnarray}
S_+(K) b_{k + 1} &>& S_0(K) b_k + S_-(K) b_{k - 1} \label{rel2} \\
S_i(K) &\equiv& [4(K + 1)^2 - 12(\Delta_\epsilon + 1)(K + 1) + M](K + 1) R_+(K)^{-1} R_i(K) \; . \nonumber
\end{eqnarray}
For our assumption to be true, the fractional coefficients in (\ref{rel2}) must exceed the $K \mapsto K + 1$ versions of the ones in (\ref{assumption}). These two conditions each give one side of an inequality for $M$. The result is $10\Delta_\epsilon^2 + 20\Delta_\epsilon + 4 < M < 10\Delta_\epsilon^2 + 29\Delta_\epsilon + 4$ which may be satisfied for any $\Delta_\epsilon > 0$. Having chosen $M$ appropriately, we have moved the problem into the domain of the theorem above. The monotonicity proof for $b_k$ is now identical to the one for $n^2 c^{\sigma\sigma(1, 1)\sigma\sigma}_{2n}$ and $n^2 c^{\sigma\sigma(1, 3)\sigma\sigma}_{2n}$.
\end{proof}

\section{Implementation details}
\subsection{The semidefinite program}
The conformal bootstrap is any technique for demanding that crossing symmetry and unitarity hold for the four-point function:
\begin{equation}
\left < \phi_i(x_1) \phi_j(x_2) \phi_k(x_3) \phi_l(x_4) \right > = \left ( \frac{|x_{24}|}{|x_{14}|} \right )^{\Delta_{ij}} \left ( \frac{|x_{14}|}{|x_{13}|} \right )^{\Delta_{kl}} \frac{\sum_{\mathcal{O}} \lambda_{ij\mathcal{O}} \lambda_{kl \mathcal{O}} g^{\Delta_{ij}, \Delta_{kl}}_{\mathcal{O}}(u, v)}{|x_{12}|^{\Delta_i + \Delta_j}|x_{34}|^{\Delta_k + \Delta_l}} \; . \label{4pt}
\end{equation}
The conformal blocks $g^{\Delta_{ij}, \Delta_{kl}}_{\mathcal{O}}(u, v)$ are functions of the cross-ratios $u = \frac{x_{12}^2 x_{34}^2}{x_{13}^2 x_{24}^2}$ and $v = \frac{x_{14}^2 x_{23}^2}{x_{13}^2 x_{24}^2}$. Invariance under $(1, i) \leftrightarrow (3, k)$, which relates two channels of crossing symmetry, leads to the following sum rule \cite{kps14}.
\begin{eqnarray}
&& \sum_{\mathcal{O}} \left [ \lambda_{ij\mathcal{O}} \lambda_{kl\mathcal{O}} F^{ij ; kl}_{\mp,\mathcal{O}}(u,v) \pm \lambda_{kj\mathcal{O}} \lambda_{il\mathcal{O}} F^{kj ; il}_{\mp,\mathcal{O}}(u,v) \right ] = 0 \label{rule} \\
&& F^{ij ; kl}_{\pm,\mathcal{O}} \equiv v^{\frac{\Delta_k + \Delta_j}{2}} g^{\Delta_{ij}, \Delta_{kl}}_{\mathcal{O}}(u, v) \pm u^{\frac{\Delta_k + \Delta_j}{2}} g^{\Delta_{ij}, \Delta_{kl}}_{\mathcal{O}}(v, u) \nonumber
\end{eqnarray}
To apply this rule, we choose an odd scalar $\sigma$ and an even scalar $\epsilon$ and let our external operators run over all admissible combinations of these. Using $\lambda_{\sigma\epsilon\mathcal{O}} = (-1)^\ell \lambda_{\epsilon\sigma\mathcal{O}}$, this yields
\begin{equation}
\sum_{\mathcal{O}, 2 | \ell} \left ( \lambda_{\sigma\sigma\mathcal{O}} \; \lambda_{\epsilon\epsilon\mathcal{O}} \right ) V_{+, \Delta, \ell} \left (
\begin{tabular}{c}
$\lambda_{\sigma\sigma\mathcal{O}}$ \\ $\lambda_{\epsilon\epsilon\mathcal{O}}$
\end{tabular}
\right ) + \sum_{\mathcal{O}} \lambda^2_{\sigma \epsilon \mathcal{O}} V_{-, \Delta, \ell} = 0 \label{mixed-rule}
\end{equation}
where
\begin{equation}
V_{+, \Delta, \ell} = \left [
\begin{tabular}{c}
$\left ( \begin{tabular}{cc} $F_{-, \Delta, \ell}^{\sigma \sigma ; \sigma \sigma}$ & $0$ \\ $0$ & $0$ \end{tabular} \right )$ \\
$\left ( \begin{tabular}{cc} $0$ & $0$ \\ $0$ & $F_{-, \Delta, \ell}^{\epsilon \epsilon ; \epsilon \epsilon}$ \end{tabular} \right )$ \\
$\left ( \begin{tabular}{cc} $0$ & $0$ \\ $0$ & $0$ \end{tabular} \right )$ \\
$\left ( \begin{tabular}{cc} $0$ & $\frac{1}{2} F_{-, \Delta, \ell}^{\sigma \sigma ; \epsilon \epsilon}$ \\ $\frac{1}{2} F_{-, \Delta, \ell}^{\sigma \sigma ; \epsilon \epsilon}$ & $0$ \end{tabular} \right )$ \\
$\left ( \begin{tabular}{cc} $0$ & $\frac{1}{2} F_{+, \Delta, \ell}^{\sigma \sigma ; \epsilon \epsilon}$ \\ $\frac{1}{2} F_{+, \Delta, \ell}^{\sigma \sigma ; \epsilon \epsilon}$ & $0$ \end{tabular} \right )$
\end{tabular}
\right ] \; , \; 
V_{-, \Delta, \ell} = \left [ \begin{tabular}{c} $0$ \\ $0$ \\ $F_{-, \Delta, \ell}^{\sigma \epsilon ; \sigma \epsilon}$ \\ $(-1)^\ell F_{-, \Delta, \ell}^{\epsilon \sigma ; \sigma \epsilon}$ \\ $-(-1)^\ell F_{+, \Delta, \ell}^{\epsilon \sigma ; \sigma \epsilon}$ \end{tabular} \right ] \; . \label{mixed-notation}
\end{equation}
In ruling out solutions to (\ref{mixed-rule}), which is a set of five functional equations, we must approximate each row as a finite-dimensional vector. The standard way to do this is to expand around the point $(u, v) = \left ( \frac{1}{4}, \frac{1}{4} \right )$. We may either take derivatives with respect to $z, \bar{z}$, defined by $u = |z|^2, v = |1 - z|^2$, or the diagonal / off-diagonal variables $a = z + \bar{z}, b = (z - \bar{z})^2$ \cite{ep12}. We choose $a, b$ and control the order of our derivatives $\frac{\partial^{m + n} g_{\Delta, \ell}}{\partial a^m \partial b^n}$ with two parameters $m_{\mathrm{max}}$ and $n_{\mathrm{max}}$:
\begin{eqnarray}
n &\in& \{ 0, \dots, n_{\mathrm{max}} \} \nonumber \\
m &\in& \{ 0, \dots, 2(n_{\mathrm{max}} - n) + m_{\mathrm{max}} \} \; . \label{mn-params}
\end{eqnarray}
Since half of the derivatives vanish when our conformal blocks are added or subtracted, the resulting number of components is
\begin{equation}
N = \lfloor (n_{\mathrm{max}} + 1)(m_{\mathrm{max}} + n_{\mathrm{max}} + 1) / 2 \rfloor \; . \label{mn-comps}
\end{equation}
The three-correlator plot in this work was obtained with $(m_{\mathrm{max}}, n_{\mathrm{max}}) = (3, 5)$. All one-correlator results, on the other hand, use $(m_{\mathrm{max}}, n_{\mathrm{max}}) = (5, 10)$. There are two additional parameters needed to turn (\ref{mixed-rule}) into a concrete bootstrapping problem. One is a cutoff on the number of spins, which we take to be $\ell_{\mathrm{max}} = 30$. The other is the accuracy parameter for a single conformal block, which we take to be $k_{\mathrm{max}} = 40$. This controls how many poles from the triple series
\begin{equation}
\begin{tabular}{ll}
$\Delta_1(\ell) = 1 - \ell - k$ & $k = 1, 2, \dots$ \\
$\Delta_2(\ell) = \frac{d}{2} - k$ & $k = 1, 2, \dots$ \\
$\Delta_3(\ell) = d - 1 + \ell - k$ & $k = 1, 2, \dots, \ell$
\end{tabular} \label{block-poles}
\end{equation}
appear in the function
\begin{eqnarray}
\chi_\ell(\Delta) &=& \frac{r_*^\Delta}{\Pi_i (\Delta - \Delta_i(\ell))} \label{block-prefactor} \\
r_* &\equiv& 3 - 2\sqrt{2} \; . \nonumber
\end{eqnarray}
As explained in \cite{hr13, hor13, kps14}, there are algorithms for explicitly constructing each conformal block derivative as a rational approximation:
\begin{equation}
\frac{\partial^{m + n}}{\partial a^m \partial b^n} F^{ij ; kl}_{\pm, \Delta, \ell}(a = 1, b = 0) = \chi_\ell(\Delta) P^{ij ; kl ; mn}_{\pm, \ell}(\Delta) \; . \label{block-approx}
\end{equation}
Here, $P^{ij ; kl ; mn}_{\pm, \ell}$ is a polynomial with the same degree as $\chi_\ell$ for $m = n = 0$. Its degree goes up by one whenever the derivative order is increased. Our task of inputting $(k_{\mathrm{max}}, \ell_{\mathrm{max}}, m_{\mathrm{max}}, n_{\mathrm{max}})$ and computing a table suitable for approximating (\ref{mixed-rule}) is accomplished with the program \texttt{PyCFTBoot} \cite{b16}.

With the truncations described above, problems of this form are tractable with semidefinite programming \cite{s15}. In the dual formulation, one wishes to find a linear functional $\textbf{y}$ which sends each term of (\ref{mixed-rule}) to a positive-definite matrix, thereby certifying that no solution to crossing symmetry exists. For illustrative purposes, we consider a single correlator problem which allows us to drop the $ij ; kl$ and $\pm$ labels on $P^{ij ; kl ; mn}_{\pm, \ell}$. We will also drop $mn$ through our understanding that $\textbf{P}_\ell$ is a vector with components $P^i_\ell$. If we single out the contribution of the identity operator as $\textbf{n}$, we arrive at the polynomial matrix program (PMP) where we include an objective $\textbf{b}$ for generality.
\begin{eqnarray}
&& \mathrm{maximize} \; \textbf{b}^{\mathrm{T}} \textbf{y} \; \mathrm{over} \; \textbf{n}^{\mathrm{T}} \textbf{y} = 1 \nonumber \\
&& \mathrm{such \; that} \; \textbf{P}_\ell(\Delta)^{\mathrm{T}} \textbf{y} \geq 0 \; \mathrm{for \; all} \; \ell \leq \ell_{\mathrm{max}}, \Delta \geq \Delta_{\mathrm{min}} \label{pmp1}
\end{eqnarray}
After all, the crossing equation
\begin{equation}
\sum_{k, \ell} \lambda^2_{k, \ell} \textbf{P}_\ell(\Delta_k) = \textbf{n} \label{toy-rule}
\end{equation}
becomes a contradiction when $\textbf{y}$ solving the above conditions is applied to both sides. If we reshuffle each vector according to
\begin{eqnarray}
\tilde{P}^0_\ell &=& \frac{1}{n^0} P^0_\ell \nonumber \\
\tilde{P}^i_\ell &=& P^i_\ell - \frac{n^i}{n^0} P^0_{\ell} \; , \label{reshuffling}
\end{eqnarray}
dotting $\textbf{P}_\ell$ with a functional whose action on $\textbf{n}$ is 1 becomes the same as dotting $\tilde{\textbf{P}}_\ell$ with a functional whose leading component is 1. This is precisely the choice to work with crossing equations projectively as (\ref{toy-rule}) becomes
\begin{equation}
\sum_{k, \ell} \lambda^2_{k, \ell} \left [ \begin{tabular}{c} $\tilde{P}^0_\ell(\Delta_k)$ \\ $\tilde{P}^i_\ell(\Delta_k)$ \end{tabular} \right ] = \left [ \begin{tabular}{c} $1$ \\ $0$ \end{tabular} \right ] \label{reshuffled-toy-rule}
\end{equation}
after reshuffling both sides. Any spectrum satisfying the bottom row can automatically be made to satisfy the top row through a rescaling. From now on, we will denote the polynomial vector in the bottom row of (\ref{reshuffled-toy-rule}) by $\textbf{P}_\ell$ to rewrite the PMP.
\begin{eqnarray}
&& \mathrm{maximize} \; \textbf{b}^{\mathrm{T}} \textbf{y} \nonumber \\
&& \mathrm{such \; that} \; P^0_\ell(x) + \textbf{P}_\ell(x)^{\mathrm{T}} \textbf{y} \geq 0 \; \mathrm{for \; all} \; \ell \leq \ell_{\mathrm{max}}, x \geq 0 \label{pmp2}
\end{eqnarray}
Here, $x = \Delta - \Delta_{\mathrm{min}}$. To solve this type of problem efficiently, we use the program \texttt{SDPB} \cite{s15}. Because we will see an alternative choice shortly, we briefly review the process by which \texttt{SDPB} translates (\ref{pmp2}) into a semidefinite program (SDP).

Positivity of $P^0_\ell(x) + \textbf{P}_\ell(x)^{\mathrm{T}} \textbf{y}$ on the half-line is equivalent to the requirement that it be equal to
\begin{equation}
\mathrm{Tr} \left ( \left [ \begin{tabular}{cc} $\textbf{q}_\ell(x) \textbf{q}_\ell^{\mathrm{T}}(x)$ & $0$ \\ $0$ & $x\textbf{q}_\ell(x) \textbf{q}_\ell^{\mathrm{T}}(x)$ \end{tabular} \right ] Y_\ell \right ) \label{constraint1}
\end{equation}
where $Y_\ell$ is positive-definite and $\textbf{q}_\ell$ is a vector of orthogonal polynomials. We have abused notation slightly since the maximum power of $x$ that (\ref{constraint1}) needs to express may be even or odd. Because of this, the first $\textbf{q}_\ell$ might have one more component than the second $\textbf{q}_\ell$. It is sufficient to demand this equality on a set of sample points which we denote $x_k$. It is also possible to combine all $Y_{(k, \ell)}$ into a single matrix $Y$. Making the identifications
\begin{eqnarray}
A_{(k, \ell)} &=& \mathrm{diag} \left ( 0, \dots, 0,  \left [ \begin{tabular}{cc} $\textbf{q}_\ell(x_k) \textbf{q}_\ell^{\mathrm{T}}(x_k)$ & $0$ \\ $0$ & $x_k\textbf{q}_\ell(x_k) \textbf{q}_\ell^{\mathrm{T}}(x_k)$ \end{tabular} \right ], 0, \dots, 0 \right ) \nonumber \\
B_{(k, \ell), i} &=& -P_\ell^i(x_k) \label{constraint2} \\
c_{(k, \ell)} &=& P_\ell^0(x_k) \nonumber \\
C &=& 0 \; , \nonumber
\end{eqnarray}
(\ref{pmp2}) becomes
\begin{eqnarray}
&& \mathrm{maximize} \; \mathrm{Tr}(CY) + \textbf{b}^{\mathrm{T}} \textbf{y} \; \mathrm{over} \; Y \succeq 0 \nonumber \\
&& \mathrm{such \; that} \; \mathrm{Tr}(A_*Y) + B\textbf{y} = \textbf{c} . \label{sdp}
\end{eqnarray}
In the numerical bootstrap, (\ref{sdp}) and the primal problem corresponding to it are typically solved together, in order to see which one becomes feasible first \cite{s15}.

\subsection{The extremal functional method}
It was shown in \cite{ep12} that solutions to crossing symmetry may be built by locating the zeros of $\textbf{y}$. This functional may be found either by ruling out CFTs just outside the allowed region or by maximizing an OPE coefficient just inside it. Ideally, elements of the spin-$\ell$ spectrum are dimensions $\Delta_k$ such that
\begin{equation}
\mathrm{det} \left ( \textbf{y}^{\mathrm{T}} \left [ \begin{tabular}{ccc} $(\textbf{P}_\ell(\Delta_k))_{0,0}$ & $\dots$ & $(\textbf{P}_\ell(\Delta_k))_{0,n}$ \\ $\vdots$ & $\ddots$ & $\vdots$ \\ $(\textbf{P}_\ell(\Delta_k))_{n,0}$ & $\dots$ & $(\textbf{P}_\ell(\Delta_k))_{n,n}$ \end{tabular} \right ] \right ) = 0 \label{ideal-zero}
\end{equation}
where $\textbf{P}_\ell$ is one of the polynomial vectors appearing in (\ref{block-approx}). On the other hand, numerical errors usually prevent this polynomial from ever reaching zero for real $\Delta$. There are two approaches, both based on the \texttt{spectrum.py} script \cite{s17a}, which make it easy to account for this. The first is to run the script as written after spending several iterations to bring the primal and dual solutions close together. A highly converged functional is needed since \texttt{spectrum.py} assumes that the would-be zeros are close to the local minima of (\ref{ideal-zero}) which uses only the polynomial numerator. The second is to modify the script to use a non-polynomial function minimizer, allowing us to multiply (\ref{ideal-zero}) by the prefactor $\chi_\ell$ from (\ref{block-approx}). The advantage is that when $\textbf{y}$ acts on a full convolved block, the local minima are closer to the physical $\Delta_k$. As a result, punishing \texttt{SDPB} parameters are no longer required. For the data in Figure \ref{manifold}, we have used the first approach and specified
\begin{eqnarray}
&& \texttt{--precision=660} \nonumber \\
&& \texttt{--dualityGapThreshold=1e-75} \nonumber
\end{eqnarray}
as the non-default parameters. The rest of the script obtains high precision OPE coefficients directly from the primal solution.\footnote{For an implementation with linear programming, see \cite{epprsv14}. The semidefinite programming version first appeared in \cite{bllrv14} which uses opposite conventions for what the primal and dual problems are.}

To plot the points in Figure \ref{non-manifold}, we have used an older method which minimizes the error in a set of crossing equations with known scaling dimensions.\footnote{An analogous approach to the severely truncated bootstrap --- fitting operator dimensions and OPE coefficients at the same time --- was recently explored in \cite{l17}.} Once a set of $Z$ stable operators has been found, we may consider a truncated crossing equation of the form
\begin{equation}
\sum_{k = 1}^Z a_k \textbf{F}_k = \textbf{n} \; . \label{rows}
\end{equation}
When studying a single correlator involving $\mathbb{Z}_2$-even operators, we make the identifications
\begin{eqnarray}
a_k &=& \lambda^2_{\phi\phi\mathcal{O}_k} \nonumber \\
\textbf{F}_k &=& \textbf{F}^{\phi\phi ; \phi\phi}_{-,\Delta_k,\ell_k} \nonumber \\
\textbf{n} &=& -\textbf{F}^{\phi\phi ; \phi\phi}_{-,0,0} \; . \label{rows-decode1}
\end{eqnarray}
The point is that $\phi$ may denote $\sigma$, $\epsilon$ or indeed any scalar whose self-OPE is well approximated by the $Z$ operators in our set. We regard (\ref{rows}) as a set of $N > Z$ linear equations with $N$ given by (\ref{mn-comps}). A naive approach is to remove rows corresponding to high derivatives bringing the number of equations down to $Z$. This often leads to $a_k$ coefficients that are negative, even when we use dimensions that are accurate to four digits. To overcome this, we follow \cite{ep12} and leave one extra equation so that $N = Z + 1$. With this overdetermined system, we will not be able to find $a_k$ that satisfy (\ref{rows}) exactly. Rather, we find the $a_k$ that minimize the distance between the left and right sides of (\ref{rows}). In this fit, the constraint $a_k \geq 0$ may be specified by hand. If our norm for this is the 1-norm, it is convenient to introduce a vector of positive entries $\textbf{t}$ such that
\begin{eqnarray}
&& -\textbf{t} \leq \textbf{n} - \sum_{k = 1}^Z a_k \textbf{F}_k \leq \textbf{t} \nonumber \\
&& \left\Vert \textbf{n} - \sum_{k = 1}^Z a_k \textbf{F}_k \right\Vert \leq \sum_{k = 1}^{Z + 1} t_k \; . \label{norm}
\end{eqnarray}

We may now concatenate $\textbf{t}$ and $\textbf{a}$ into a vector $\textbf{y}$ and recognize that the norm (\ref{norm}) is $\textbf{b}^{\mathrm{T}} \textbf{y}$ with the objective $\textbf{b}^{\mathrm{T}} = [1, \dots, 1, 0, \dots, 0]$. Under the identifications
\begin{eqnarray}
B &=& \left [ \begin{tabular}{rrrr} $-I_{Z + 1}$ & $\textbf{F}_1$ & $\dots$ & $\textbf{F}_Z$ \\ $-I_{Z + 1}$ & $-\textbf{F}_1$ & $\dots$ & $-\textbf{F}_Z$ \end{tabular} \right ] \nonumber \\
\textbf{c} &=& \left [ \begin{tabular}{r} $\textbf{n}$ \\ $-\textbf{n}$ \end{tabular} \right ] \; , \label{constraint2}
\end{eqnarray}
this becomes the linear program (LP):
\begin{eqnarray}
&& \mathrm{minimize} \; \textbf{b}^{\mathrm{T}} \textbf{y} \; \mathrm{over} \; \textbf{y} \geq 0 \nonumber \\
&& \mathrm{such \; that} \; B\textbf{y} \leq \textbf{c} . \label{lp}
\end{eqnarray}

It is amusing to point out that a reader using \texttt{SDPB} may continue to use it for solving (\ref{lp}) because every linear program is also a semidefinite program. To do this, the $(2Z + 2) \times (2Z + 1)$ and $(2Z + 2) \times 1$ matrices $B$ and $\textbf{c}$ must be enlarged to $(4Z + 3) \times (2Z + 1)$ and $(4Z + 3) \times 1$ so that they encode component-wise positivity of $\textbf{y}$ in addition to (\ref{norm}). The next necessary step is replacing $\textbf{b}$ with $-\textbf{b}$ (and prepending an arbitrary number) as \texttt{SDPB}'s objective is maximized instead of minimized. Finally, the orthogonal polynomials should all be $1$. With this choice, constraint matrices in (\ref{sdp}) simply pull off a diagonal component of $Y$ and $\textbf{c} - B\textbf{y} = \mathrm{Tr}(A_*Y)$ reads $\textbf{c} - B\textbf{y} \geq 0$. In any case, whether (\ref{lp}) is solved as an LP or SDP, it is clear that the returned $\textbf{y}$ will give us the OPE coefficients $\lambda^2_{\phi\phi\mathcal{O}_k}$.

\begin{figure}[h]
\centering
\subfloat[][$|\lambda_{\sigma\sigma\mathcal{O}}|$ comparison]{\includegraphics[scale=0.45]{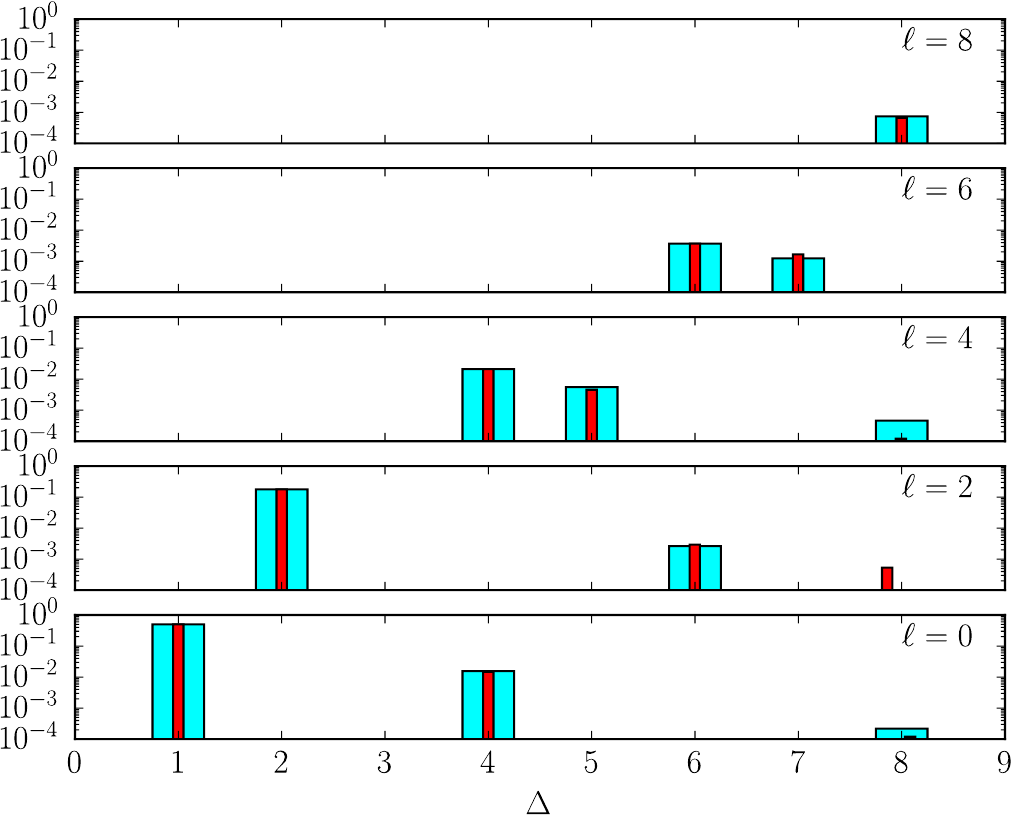}}
\subfloat[][$|\lambda_{\epsilon\epsilon\mathcal{O}}|$ comparison]{\includegraphics[scale=0.45]{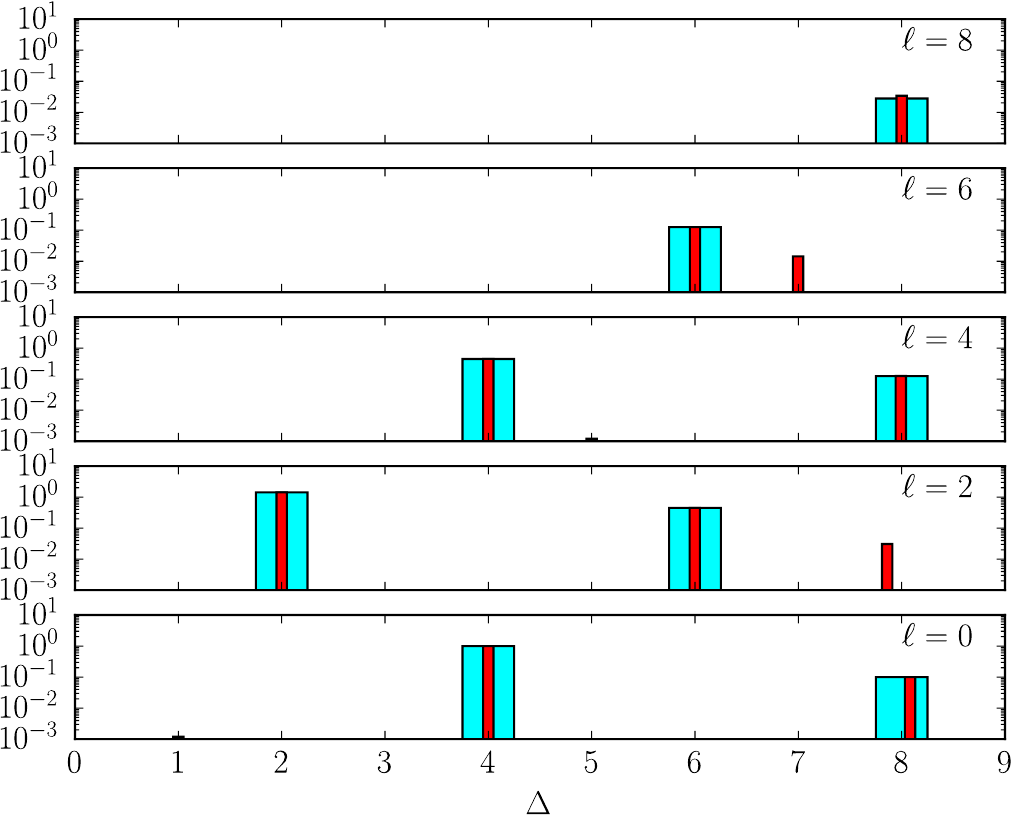}}
\caption{OPE coefficients in the 2D Ising model for operators with $\Delta \leq 8$. Wide blue bars show exact values and narrow red bars show estimates from the fit used in \cite{ep12}.}
\label{2d-fit}
\end{figure}

To verify that this gives reliable values for $a_k$ in two dimensions, we have compared exact and approximate results in the Ising model where $\epsilon \times \epsilon$ operators are also in $\sigma \times \sigma$. Starting with $\Delta_\phi = \Delta_\sigma$, the $\lambda^2_{\sigma\sigma\mathcal{O}}$ coefficients are very close to the ones found in \cite{ep12}. Changing the external dimension to $\Delta_\phi = \Delta_\epsilon$, the same algorithm produces accurate $\lambda^2_{\epsilon\epsilon\mathcal{O}}$ coefficients as well. These are shown in Figure \ref{2d-fit}.

\begin{figure}[h]
\centering
\subfloat[][$|\lambda_{\sigma\sigma\mathcal{O}}|$ comparison]{\includegraphics[scale=0.45]{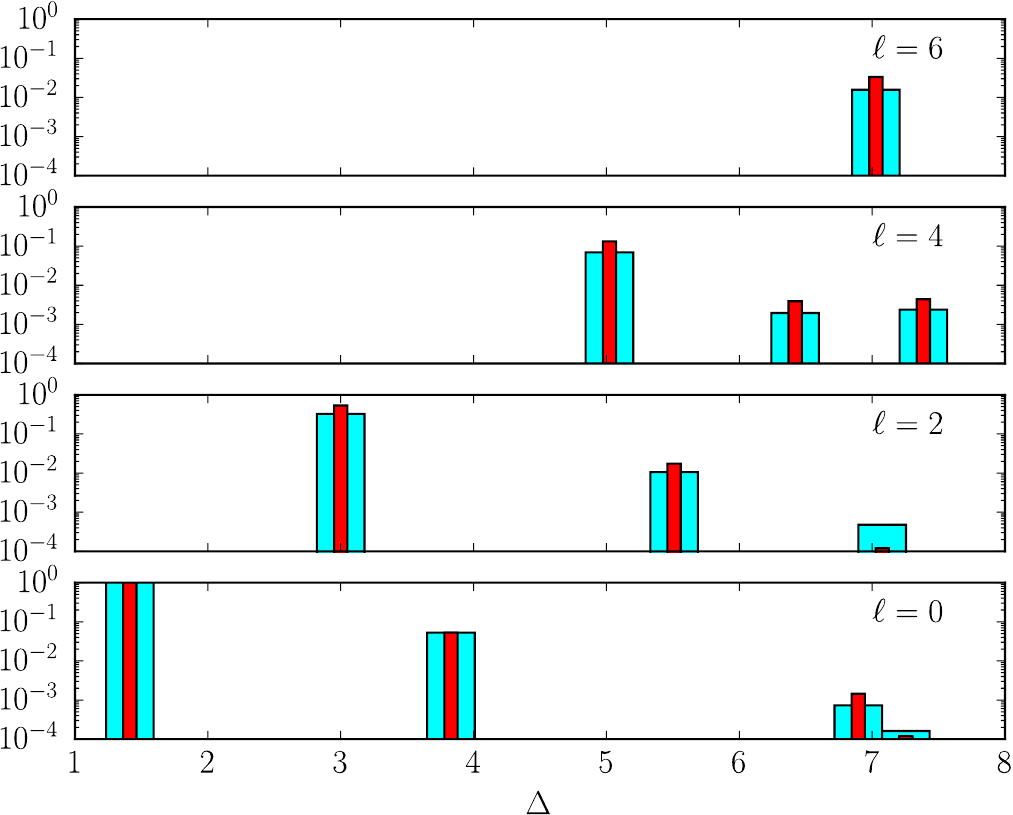}}
\subfloat[][$|\lambda_{\epsilon\epsilon\mathcal{O}}|$ comparison]{\includegraphics[scale=0.45]{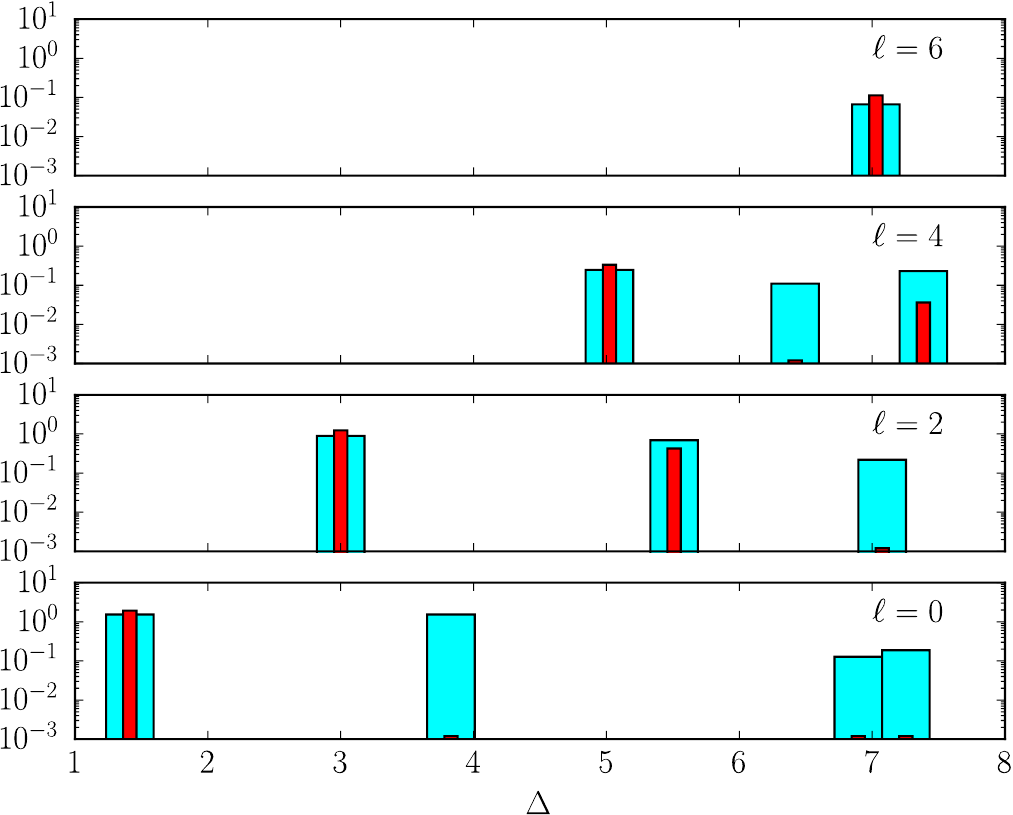}}
\caption{The analogue of Figure \ref{2d-fit} for the 3D Ising model. In the fit, we have used the $\mathbb{Z}_2$-even operator dimensions from \cite{s17a}. This is already a longer list than anything that is likely to come from a one-correlator bootstrap. Due to a remaining bias in the spectrum, there is strong disagreement between the primal and dual methods.}
\label{3d-fit}
\end{figure}

Turning to three dimensions, exact Ising CFT data are not available, but the tables in \cite{s17a} have negligible error for our purposes. Despite all the progress in isolating this model, it appears that several spin-$\ell$ operators are still missing from these spectra: the ones that do not fall into $[\sigma\sigma]_n$, $[\epsilon\epsilon]_n$ and $[\sigma\epsilon]_n$ twist families. As evidence of this, Figure \ref{3d-fit} shows that OPE coefficients fit with the dual method differ greatly from the ones returned by \texttt{spectrum.py}. To guess where the first missing operator appears, we can take the naive view that twist families exist all the way down to $\ell = 0$. Constructing one out of an irrelevant operator $\Phi$, we have $\Delta \approx 2\Delta_\Phi + 2n + \ell > 2d$. Four-point functions approximated by bootstrap data have already proven useful for conformal perturbation theory and measuring non-Gaussianity \cite{ks16, rsz17}. As it is significantly affected by the incompleteness of the spectrum above $\Delta_* \approx 2d$, a fit to the crossing equations must be a less forgiving problem.

Although we have not done so in this work, it should also be possible to fit mixed OPE coefficients in a 2D theory. The last two rows of (\ref{mixed-rule}) would lead to coupled quadratic equations which rapidly become impractical to solve as our system grows. Therefore, we will focus on the third row. If only the identity has been singled out, this is a homogenous equation and its parameters are ambiguous up to a rescaling. This is why it helps to find OPE coefficients in a prescribed order. As long as $\lambda_{\phi\phi\phi^2}$ has been found from the $\phi \times \phi$ fit described above, permutation symmetry of this coefficient may be used to make our equation inhomogeneous. In this case, it takes the form of (\ref{rows}) but now with
\begin{eqnarray}
a_k &=& \lambda^2_{\phi\phi^2\mathcal{O}_k} \nonumber \\
\textbf{F}_k &=& \textbf{F}^{\phi\phi^2 ; \phi\phi^2}_{-,\Delta_k,\ell_k} \nonumber \\
\textbf{n} &=& -\lambda^2_{\phi\phi^2\phi}\textbf{F}^{\phi\phi^2 ; \phi\phi^2}_{-,\Delta_\phi,0} \; . \label{rows-decode2}
\end{eqnarray}
Essentially, $\phi$ plays the same role that the identity operator played before.

\bibliographystyle{unsrt}
\bibliography{references}

\begin{thebibliography}{10}

\bibitem{bpz84}
A.~A. Belavin; A. M. Polyakov; A.~B. Zamolodchikov.
\newblock Infinite conformal symmetry in two-dimensional quantum field theory.
\newblock {\em Nuclear Physics B}, 241:333--380, 1984.

\bibitem{fqs84}
D.~Friedan; Z. Qiu;~S. Shenker.
\newblock Conformal invariance, unitarity, and critical exponents in two
  dimensions.
\newblock {\em Physical Review Letters}, 52(18):1575--1578, 1984.

\bibitem{fqs86}
D.~Friedan; Z. Qiu;~S. Shenker.
\newblock Details of the non-unitarity proof for highest weight representations
  of the {V}irasoro algebra.
\newblock {\em Communications in Mathematical Physics}, 107:535--542, 1986.

\bibitem{gko85}
P.~Goddard; A. Kent;~D. Olive.
\newblock Virasoro algebras and coset space models.
\newblock {\em Physics Letters B}, 152(88), 1985.

\bibitem{gko86}
P.~Goddard; A. Kent;~D. Olive.
\newblock Unitary representations of the {V}irasor and super-{V}irasoro
  algebras.
\newblock {\em Communications in Mathematical Physics}, 103:105--119, 1986.

\bibitem{rrtv08}
R.~Rattazzi; S. Rychkov; E. Tonni;~A. Vichi.
\newblock Bounding scalar operator dimensions in 4{D} {CFT}.
\newblock {\em Journal of High Energy Physics}, 12(31), 2008.
\newblock arXiv:0807.0004.

\bibitem{rv09}
S.~Rychkov;~A. Vichi.
\newblock Universal constraints on conformal operator dimensions.
\newblock {\em Physical Review D}, 80(4), 2009.
\newblock arXiv:0905.2211.

\bibitem{cr09}
F.~Caracciolo;~S. Rychkov.
\newblock Rigorous limits on the interaction strength in quantum field theory.
\newblock {\em Physical Review D}, 81(8), 2010.
\newblock arXiv:0912.2726.

\bibitem{rrv10a}
R.~Rattazzil S. Rychkov;~A. Vichi.
\newblock Central charge bounds in 4{D} conformal field theory.
\newblock {\em Physical Review D}, 83(4), 2011.
\newblock arXiv:1009.2725.

\bibitem{rrv10b}
R.~Rattazzil S. Rychkov;~A. Vichi.
\newblock Bounds in 4{D} conformal field theories with global symmetry.
\newblock {\em Journal of Physics A}, 44(3), 2011.
\newblock arXiv:1009.5985.

\bibitem{v12}
A.~Vichi.
\newblock Improved bounds for {CFT}s with global symmetries.
\newblock {\em Journal of High Energy Physics}, 1(162), 2012.
\newblock arXiv:1106.4037.

\bibitem{ps10}
D.~Poland;~D. Simmons-Duffin.
\newblock Bounds on 4{D} conformal and superconformal field theories.
\newblock {\em Journal of High Energy Physics}, 5(17), 2010.
\newblock arXiv:1009.2087.

\bibitem{psv12}
D.~Poland; D. Simmons-Duffin;~A. Vichi.
\newblock Carving out the space of 4{D} {CFT}s.
\newblock {\em Journal of High Energy Physics}, 5(110), 2012.
\newblock arXiv:1109.5176.

\bibitem{kps14}
F.~Kos; D. Poland;~D. Simmons-Duffin.
\newblock Bootstrapping mixed correlators in the 3{D} {I}sing model.
\newblock {\em Journal of High Energy Phyiscs}, 11(109), 2014.
\newblock arXiv:1406.4858.

\bibitem{s15}
D.~Simmons-Duffin.
\newblock A semidefinite program solver for the conformal bootstrap.
\newblock {\em Journal of High Energy Physics}, 6(174), 2015.
\newblock arXiv:1502.02033.

\bibitem{kpsv16}
F.~Kos; D. Poland; D. Simmons-Duffin;~A. Vichi.
\newblock Precision islands in the {I}sing and {O(N)} models.
\newblock {\em Journal of High Energy Physics}, 8(36), 2016.
\newblock arXiv:1603.04436.

\bibitem{lrv13}
P.~Liendo; L. Rastelli; B.~C. van Rees.
\newblock The bootstrap program for boundary {CFT}${}_d$.
\newblock {\em Journal of High Energy Physics}, 7(113), 2013.
\newblock arXiv:1210.4258.

\bibitem{z05}
Al.~B. Zamolodchikov.
\newblock Three-point function in the minimal {L}iouville gravity.
\newblock {\em Theoretical and Mathematical Physics}, 142(2):183--196, 2005.

\bibitem{bz05}
A.~A. Belavin; A.~B. Zamolodchikov.
\newblock Moduli integrals and ground ring in minimal {L}iouville gravity.
\newblock {\em Journal of Experimental and Theoretical Physics Letters},
  82(1):7--13, 2005.

\bibitem{r14}
S.~Ribault.
\newblock Conformal field theory on the plane.
\newblock 2014.
\newblock arXiv:1406.4290.

\bibitem{hv17}
M.~Hogervorst; B.~C. van Rees.
\newblock Crossing symmetry in alpha space.
\newblock {\em Journal of High Energy Physics}, 11(193), 2017.
\newblock arXiv:1702.08471.

\bibitem{h17}
M.~Hogervorst.
\newblock Crossing kernels for boundary and crosscap {CFT}s.
\newblock 2017.
\newblock arXiv:1703.08159.

\bibitem{efr16}
I.~Esterlis; A. L. Fitzpatrick;~D. Ramirez.
\newblock Closure of the operator product expansion in the non-unitary
  bootstrap.
\newblock {\em Journal of High Energy Physics}, 11(30), 2016.
\newblock arXiv:1606.07458.

\bibitem{g13}
F.~Gliozzi.
\newblock More constraining conformal bootstrap.
\newblock {\em Physical Review Letters}, 111:161602, 2013.
\newblock arXiv:1307.3111.

\bibitem{gr14}
F.~Gliozzi;~A. Rago.
\newblock Critical exponents of the 3d {I}sing and related models from the
  conformal bootstrap.
\newblock {\em Journal of High Energy Physics}, 10(42), 2014.
\newblock arXiv:1403.6003.

\bibitem{glmr15}
F.~Gliozzi; P. Liendo; M. Meineri;~A. Rago.
\newblock Boundary and interface {CFT}s from the conformal bootstrap.
\newblock {\em Journal of High Energy Physics}, 5(36), 2015.
\newblock arXiv:1502.07217.

\bibitem{s16}
D.~Simmons-Duffin.
\newblock {TASI} lectures on the conformal bootstrap.
\newblock 2016.
\newblock arXiv:1602.07982.

\bibitem{v11}
A.~Vichi.
\newblock {\em A new method to explore conformal field theories in any
  dimension}.
\newblock PhD thesis, \'{E}cole Polytechnique F\'{e}d\'{e}rale de Lausanne,
  2011.

\bibitem{ep12}
S.~El-Showk; M.~F. Paulos.
\newblock Bootstrapping conformal field theories with the extremal functional
  method.
\newblock {\em Physical Review Letters}, 111:241601, 2013.
\newblock arXiv:1211.2810.

\bibitem{ks96}
R.~Koekoek; R.~F. Swarttouw.
\newblock The {A}skey-scheme of hypergeometric orthogonal polynomials and its
  q-analogue.
\newblock {\em Reports of the Faculty of Technical Mathematics and
  Informatics}, 94(5), 1996.
\newblock math/9602214.

\bibitem{gkss17}
R.~Gopakumar; A. Kaviraj; K. Sen;~A. Sinha.
\newblock A {M}ellin space approach to the conformal bootstrap.
\newblock {\em Journal of High Energy Physics}, 5(27), 2017.
\newblock arXiv:1611.08407.

\bibitem{c12}
W.~Chu.
\newblock Analytical formulae for extended ${}_3{F}_2$-series of
  {W}atson-{W}hipple-{D}ixon with two extra integer parameters.
\newblock {\em Mathematics of Computation}, 81(277):467--479, 2012.

\bibitem{m16}
D.~Maz\'{a}\v{c}.
\newblock Analytic bounds and emergence of {AdS$_2$} physics from the conformal
  bootstrap.
\newblock {\em Journal of High Energy Physics}, 4(146), 2017.
\newblock arXiv:1611.10060.

\bibitem{c17}
S.~Caron-Huot.
\newblock Analyticity in spin in conformal theories.
\newblock {\em Journal of High Energy Physics}, 9(78), 2017.
\newblock arXiv:1703.00278.

\bibitem{dlm12}
P.~Desrosiers; L. Lapointe;~P. Mathieu.
\newblock Superconformal field theory and {J}ack superpolynomials.
\newblock {\em Journal of High Energy Physics}, 9(37), 2012.
\newblock arXiv:1205.0784.

\bibitem{bbt13}
A.~A. Belavin; M. A. Bershtein; G.~M. Tarnopolsky.
\newblock Bases in coset conformal field theory from {AGT} correspondence and
  {M}acdonald polynomials at the roots of unity.
\newblock {\em Journal of High Energy Physics}, 3(19), 2013.
\newblock arXiv:1211.2788.

\bibitem{adm13}
L.~Alarie-Vezina; P. Desrosiers;~P. Mathieu.
\newblock Ramond singular vectors and {J}ack superpolynomials.
\newblock {\em Journal of Physics A}, 47(3), 2013.
\newblock arXiv:1309.7965.

\bibitem{df84}
V.~S. Dotsenko; V.~A. Fateev.
\newblock Conformal algebra and multipoint correlation functions in 2{D}
  statistical models.
\newblock {\em Nulear Physics B}, 240:312--348, 1984.

\bibitem{df85a}
V.~S. Dotsenko; V.~A. Fateev.
\newblock Four point correlation functions and the operator algebra in the
  two-dimensional conformal invariant theories with central charge $c \leq 1$.
\newblock {\em Nulear Physics B}, 251:691--734, 1985.

\bibitem{df85b}
V.~S. Dotsenko; V.~A. Fateev.
\newblock Operator algebra of two-dimensional conformal theories with central
  charge $c \leq 1$.
\newblock {\em Physics Letters B}, 154:291--295, 1985.

\bibitem{w77}
J.~A. Wilson.
\newblock Three term contiguous relations and some new orthogonal polynomials.
\newblock In E.~Saff;~R. Varga, editor, {\em Pade and Rational Approximations},
  pages 227--232. Academic Press, 1977.

\bibitem{s09}
P.~Suchanek.
\newblock {\em Recursive methods of determination of 4-point blocks in $N = 1$
  superconformal field theories}.
\newblock PhD thesis, Jagiellonian University, 2009.

\bibitem{fgp90}
P.~Furlan; A. C. Ganchev; V.~B. Petkova.
\newblock Fusion matrices and $c < 1$ (quasi) local conformal theories.
\newblock {\em International Journal of Modern Physics}, A5(14):2721--2735,
  1990.

\bibitem{pss14}
M.~Pawelkiewicz; V. Schomerus;~P. Suchanek.
\newblock The universal {R}acah-{W}igner symbol for ${U_q}({OSp}(1|2))$.
\newblock {\em Journal of High Energy Physics}, 4(79), 2014.
\newblock arXiv:1307.6866.

\bibitem{epprsv14}
S.~El-Showk; M. F. Paulos; D. Poland; S. Rychkov; D. Simmons-Duffin;~A. Vichi.
\newblock Solving the 3{D} {I}sing model with the conformal bootstrap {II}:
  c-minimization and precise critical exponents.
\newblock {\em Journal of Statistical Physics}, 157:869--914, 2014.
\newblock arXiv:1403.4545.

\bibitem{ep16}
S.~El-Showk; M.~F. Paulos.
\newblock Extremal bootstrapping: go with the flow.
\newblock 2016.
\newblock arXiv:1605.08087.

\bibitem{s17a}
D.~Simmons-Duffin.
\newblock The lightcone bootstrap and the spectrum of the 3d {I}sing {CFT}.
\newblock {\em Journal of High Energy Physics}, 3(86), 2017.
\newblock arXiv:1612.08471.

\bibitem{cl17}
C.-M. Chang; Y.-H. Lin.
\newblock Carving out the end of the world or (superconformal bootstrap in six
  dimensions).
\newblock {\em Journal of High Energy Physics}, 8(128), 2017.
\newblock arXiv:1705.05392.

\bibitem{cflw17}
C.-M. Chang; M. Fluder; Y.-H. Lin;~Y. Wang.
\newblock Spheres, charges, instantons and bootstrap: {A} five-dimensional
  odyssey.
\newblock 2017.
\newblock arXiv:1710.08418.

\bibitem{cll17}
M.~Cornagliotto; M. Lemos;~P. Liendo.
\newblock Bootstrapping the $({A}_1, {A}_2)$ {A}rgyres-{D}ouglas theory.
\newblock {\em Journal of High Energy Physics}, 3(33), 2018.
\newblock arXiv:1711.00016.

\bibitem{ckly17}
S.~Collier; P. Kravchuk; Y.-H. Lin;~X. Yin.
\newblock Bootstrapping the spectral function: {O}n the uniqueness of
  {L}iouville theory and the universality of {BPZ}.
\newblock 2017.
\newblock arXiv:1702.00423.

\bibitem{lsswy15}
Y.-H. Lin; S.-H. Shao; D. Simmons-Duffin; Y. Wang;~X. Yin.
\newblock N=4 superconformal bootstrap of the {K}3 {CFT}.
\newblock {\em Journal of High Energy Physics}, 5(126), 2017.
\newblock arXiv:1511.04065.

\bibitem{lswy16}
Y.-H. Lin; S.-H Shao; Y. Wang;~X. Yin.
\newblock $(2, 2)$ superconformal bootstrap in two dimensions.
\newblock {\em Journal of High Energy Physics}, 5(112), 2017.
\newblock arXiv:1610.05371.

\bibitem{cls17}
M.~Cornagliotto; M. Lemos;~V. Schomerus.
\newblock Long multiplet bootstrap.
\newblock {\em Journal of High Energy Physics}, 10(119), 2017.
\newblock arXiv:1702.05101.

\bibitem{prvz15}
M.~F. Paulos; S. Rychkov; B.~C. van Rees; B.~Zan.
\newblock Conformal invariance in the long-range {I}sing model.
\newblock {\em Nuclear Physics B}, 902:296--291, 2016.
\newblock arXiv:1509.00008.

\bibitem{brrz17a}
C.~Behan; L. Rastelli; S. Rychkov;~B. Zan.
\newblock A scaling theory for the long-range to short-range crossover and an
  infrared duality.
\newblock {\em Journal of Physics A}, 50(35), 2017.
\newblock arXiv:1703.05325.

\bibitem{brrz17b}
C.~Behan; L. Rastelli; S. Rychkov;~B. Zan.
\newblock Long-range critical exponents near the short-range crossover.
\newblock {\em Physical Review Letters}, 118:241601, 2017.
\newblock arXiv:1703.03430.

\bibitem{hrv16}
M.~Hogervorst; S. Rychkov; B.~C. van Rees.
\newblock Unitarity violation at the {W}ilson-{F}isher fixed point in 4-epsilon
  dimensions.
\newblock {\em Physical Review D}, 93(12), 2016.
\newblock arXiv:1512.00013.

\bibitem{dfx17}
E.~Dyer; A. L. Fitzpatrick;~Y. Xin.
\newblock Constraints on flavored 2d {CFT} partition functions.
\newblock {\em Journal of High Energy Physics}, 2(148), 2018.
\newblock arXiv:1709.01533.

\bibitem{r11}
S.~Rychkov.
\newblock Conformal bootstrap in three dimensions?
\newblock 2011.
\newblock arXiv:1111.2115.

\bibitem{ls17}
Z.~Li;~N. Su.
\newblock 3{D} {CFT} archipelago from single correlator bootstrap.
\newblock 2017.
\newblock arXiv:1706.06960.

\bibitem{bbr17}
V.~Bashmakov; M. Bertolini;~H. Raj.
\newblock On non-supersymmetric conformal manifolds: field theory and
  holography.
\newblock {\em Journal of High Energy Physics}, 11(167), 2017.
\newblock arXiv:1709.01749.

\bibitem{bbclp17}
M.~Baggio; N. Bobev; S. M. Chester; E. Lauria; S.~S. Pufu.
\newblock Decoding a three-dimensional conformal manifold.
\newblock {\em Journal of High Energy Physics}, 2(62), 2018.
\newblock arXiv:1712.02698.

\bibitem{b17}
C.~Behan.
\newblock Conformal manifolds: {ODE}s from {OPE}s.
\newblock 2017.
\newblock arXiv:1709.03967.

\bibitem{s17b}
S.~Hollands.
\newblock Action principle for {OPE}.
\newblock {\em Nuclear Physics B}, 926:614--638, 2018.
\newblock arXiv:1710.05601.

\bibitem{fkps12}
A.~L. Fitzpatrick; J. Kaplan; D. Poland;~D. Simmons-Duffin.
\newblock The analytic bootstrap and {A}d{S} superhorizon locality.
\newblock {\em Journal of High Energy Physics}, 12(4), 2013.
\newblock arXiv:1212.3616.

\bibitem{kz12}
Z.~Komargodski;~A. Zhiboedov.
\newblock Convexity and liberation at large spin.
\newblock {\em Journal of High Energy Physics}, 11(140), 2013.
\newblock arXiv:1212.4103.

\bibitem{fqs85}
D.~Friedan; Z. Qiu;~S. Shenker.
\newblock Superconformal invariance in two dimensions and the tricritical
  {I}sing model.
\newblock {\em Physics Letters B}, 151:37--43, 1985.

\bibitem{q86}
Z.~Qiu.
\newblock Supersymmetry, two-dimensional critical phenomena and the tricritical
  {I}sing model.
\newblock {\em Nuclear Physics B}, 270:205--234, 1986.

\bibitem{rs17}
J.~Rong;~N. Su.
\newblock Scalar {CFT}s and their large {N} limits.
\newblock 2017.
\newblock arXiv:1712.00985.

\bibitem{ehs16}
A.~C. Echeverri;~B. von Harling; M.~Serone.
\newblock The effective bootstrap.
\newblock {\em Journal of High Energy Physics}, 9(97), 2016.
\newblock arXiv:1606.02771.

\bibitem{w04}
R.~Wong.
\newblock Five lectures on asymptotic theory.
\newblock In Z.~Hua;~R. Wong, editor, {\em Differential equations and
  asymptotic theory in mathematical physics}, pages 189--263. World Scientific,
  2004.

\bibitem{prer12}
D.~Pappadopulo; S. Rychkov; J. Espin;~R. Rattazzi.
\newblock O{PE} convergence in conformal field theory.
\newblock {\em Physical Review D}, 86(10), 2012.
\newblock arXiv:1208.6449.

\bibitem{ry15}
S.~Rychkov;~P. Yvernay.
\newblock Remarks on the convergence properties of the conformal block
  expansion.
\newblock {\em Physics Letters B}, 753:682--686, 2015.
\newblock arXiv:1510.08486.

\bibitem{xy11}
E.~X. W. Xia; X.~M. Yao.
\newblock The signs of three-term recurrence sequences.
\newblock {\em Discrete Applied Mathematics}, 159:2290--2296, 2011.

\bibitem{hjk16}
T.~Hartman; S. Jain;~S. Kundu.
\newblock Causality constraints in conformal field theory.
\newblock {\em Journal of High Energy Physics}, 5(99), 2016.
\newblock arXiv:1509.00014.

\bibitem{hr13}
M.~Hogervorst;~S. Rychkov.
\newblock Radial coordinates for conformal blocks.
\newblock {\em Physical Review D}, 87(10), 2013.
\newblock arXiv:1303.1111.

\bibitem{hor13}
M.~Hogervorst; H. Osborn;~S. Rychkov.
\newblock Diagonal limit for conformal blocks in $d$ dimensions.
\newblock {\em Journal of High Energy Physics}, 8(14), 2013.
\newblock arXiv:1305.1321.

\bibitem{b16}
C.~Behan.
\newblock Py{CFTB}oot: {A} flexible interface for the conformal bootstrap.
\newblock {\em Communications in Computational Physics}, 22(1), 2017.
\newblock arXiv:1602.02810.

\bibitem{bllrv14}
C.~Beem; M. Lemos; P. Liendo; L. Rastelli; B.~C. van Rees.
\newblock The $\mathcal{N} = 2$ superconformal bootstrap.
\newblock {\em Journal of High Energy Physics}, 3(183), 2016.
\newblock arXiv:1412.7541.

\bibitem{l17}
W.~Li.
\newblock New method for conformal bootstrap with {OPE} truncations.
\newblock 2017.
\newblock arXiv:1711.09075.

\bibitem{ks16}
Z.~Komargodski;~D. Simmons-Duffin.
\newblock The random-bond {I}sing model in 2.01 and 3 dimensions.
\newblock {\em Journal of Physics A}, 50(15), 2017.
\newblock arXiv:1603.04444.

\bibitem{rsz17}
S.~Rychkov; D. Simmons-Duffin;~B. Zan.
\newblock Non-gaussianity of the critical 3d {I}sing model.
\newblock {\em SciPost Physics}, 2(1), 2017.
\newblock arXiv:1612.02436.

\end{thebibliography}
\end{document}